\documentclass[english,final]{LMCS}
\def\doi{8 (2:06) 2012}
\lmcsheading%
{\doi}
{1--60}
{}
{}
{Oct.~27, 2011}
{Jun.~\phantom01, 2012}
{}

\usepackage{enumerate}
\usepackage[final]{hyperref}

\usepackage{babel}

\usepackage{lmodern}

\usepackage{textcomp}           
\usepackage[T1]{fontenc}
\usepackage[utf8]{inputenc}
\usepackage{tikz}

\usepackage{stmaryrd}
\usepackage{amssymb}
\usepackage{paralist}
\usepackage{color}
\usepackage{ifthen}

\usetikzlibrary{positioning,backgrounds,calc,decorations,arrows}
\usetikzlibrary{%
  decorations.pathmorphing,%
  decorations.pathreplacing,%
  decorations.markings,%
  decorations.shapes,%
  decorations.footprints,%
  decorations.fractals,%
  decorations.text%
}%
\tikzset{%
  >=latex',%
  single step/.style={>=to,->},%
  etc/.style={edge from parent/.style={-,dotted,thick,draw}},%
  shorten/.style={shorten >=#1, shorten <=#1},%
  level distance=1cm,%
  sibling distance = 10mm,%
  edge from parent/.style={->,draw},%
  fun/.style={>=to,->},%
  node name/.style args={#1:#2}{label={[black!70]#1:\small$#2$}}%
  }%

\newcommand\wsuc\oh

\newcommand\nodeAtPos[2]{\symb{node}_{#1}(#2)}
\newcommand\posAcy[1]{\calPos^a(#1)}
\newcommand\predAcy[2]{\symb{Pre}^a_{#1}(#2)}

\newcommand\lebot{\le_\bot}
\newcommand\lbot{<_\bot}

\newcommand\leboti{\lebot^\textsf{I}}
\newcommand\lebotp{\lebot^\textsf{P}}
\newcommand\lebots{\lebot^\textsf{S}}

\newcommand\lebotr{\lebot^\textsf{R}}
\newcommand\lbotr{\lbot^\textsf{R}}

\newcommand\concat{\cdot}

\newcommand\lub{\sqcup}
\newcommand\glb{\sqcap}
\newcommand\Lub{\bigsqcup}
\newcommand\Glb{\bigsqcap}

\newcommand\similar[2]{\symb{sim}(#1,#2)}
\newcommand\similarr[2]{\abssim{\truncr{}{}}(#1,#2)}

\newcommand\abssim[1]{\symb{sim}_{#1}}

\newcommand\dd{\mathbf{d}}

\newcommand\ddr{\dd_{\truncr{}{}}}

\newcommand\trunc[2]{#1|#2}%|

\newcommand\truncr[2]{#1\mathclose{\ddagger}#2}%|
\newcommand\truncs[2]{#1\mathclose{\dagger}#2}%|

\newcommand\fNodes[2]{N^{#1}_{=#2}}
\newcommand\tNodes[2]{N^{#1}_{<#2}}

\newcommand\nats{{\mathbb N}}

\newcommand\reals{{\mathbb R}}
\newcommand\realsp{\reals^+}
\newcommand\realsnn{\reals^+_0}
\renewcommand\epsilon{\varepsilon}

\newcommand\oh\widehat
\newcommand\ol\overline
\newcommand\ot\widetilde

\newcommand\canon[1]{\calC(#1)}

\newcommand\unrav[1]{\calU\!\left(#1\right)}
\newcommand\quotient[2]{#1/\!\raisebox{-.65ex}{\ensuremath{#2}}}
\newcommand\eqc[2]{[#1]_{#2}}
\newcommand\idrel[1]{\calI_{#1}}

\newcommand\isom{\cong}
\newcommand\nisom{\not\cong}

\newcommand\calC{\mathcal{C}}

\newcommand\calG{\mathcal{G}}

\newcommand\calI{\mathcal{I}}

\newcommand\calP{\mathcal{P}}

\newcommand\calR{\mathcal{R}}

\newcommand\calT{\mathcal{T}}
\newcommand\calU{\mathcal{U}}
\newcommand\calV{\mathcal{V}}

\newcommand\calPos{\calP}

\mathchardef\mhyphen="2D

\newcommand\fcolon{\colon\,}
\newcommand\funto{\rightarrow}
\newcommand\limto{\rightarrow}
\newcommand\len[1]{\left\lvert #1 \right\rvert}
\newcommand\prefix[2]{#1|_{#2}}

\newcommand\seq[1]{\langle #1 \rangle}
\newcommand\emptyseq{\seq{}}
\newcommand\srank[1]{\symb{ar}(#1)}

\newcommand\rank[2]{\symb{ar}_{#1}(#2)}
\newcommand\nin{\not\in}

\newcommand\prs{p}
\newcommand\mrs{m}

\newcommand\isoto{\isom}

\newcommand\homto{\rightarrow}

\newcommand{\setcom}[2]{\set{#1\left\vert\vphantom{#1}\,#2\right.}}

\newcommand{\set}[1]{\left\{#1\right\}}

\newcommand\sdepth[2]{{#2}\mhyphen\symb{depth}(#1)}

\newcommand\pos[1]{\calPos(#1)}

\newcommand\nodePos[2]{\calPos_{#1}(#2)}
\newcommand\nodePosAcy[2]{\calPos^a_{#1}(#2)}

\newcommand\depth[2]{\symb{depth}_{#1}(#2)}
\newcommand\gdepth[1]{\symb{depth}(#1)}

\newcommand\iptgraphs[1][\Sigma]{\calG^\infty(#1_\bot)}
\newcommand\tgraphs[1][\Sigma]{\calG(#1)}
\newcommand\itgraphs[1][\Sigma]{\calG^\infty(#1)}
\newcommand\ctgraphs[1][\Sigma]{\calG_\calC(#1)}
\newcommand\ictgraphs[1][\Sigma]{\calG^\infty_\calC(#1)}

\newcommand\ipctgraphs[1][\Sigma]{\calG^\infty_\calC(#1_\bot)}

\newcommand\ipterms[1][\Sigma]{\calT^\infty(#1_\bot)}
\newcommand\terms[1][\Sigma]{\calT(#1)}
\newcommand\iterms[1][\Sigma]{\calT^\infty(#1)}

\newcommand\ivterms[1][\Sigma]{\calT^\infty(#1,\calV)}

\newcommand\symb[1]{\mathsf{#1}}
\newcommand\glab{\symb{lab}}
\newcommand\gsuc{\symb{suc}}

\newcommand\atPos[2]{#1|_{#2}}
\newcommand\substAtPos[3]{#1[#3]_{#2}}

\newcommand\lhs[1]{#1_{l}}
\newcommand\rhs[1]{#1_{r}}

\newcommand\subgraph[2]{#1|_{#2}}

\def\nothing{}
\let\oldTo\to

\newcommand\preright{\mapsto}
\newcommand\finleftright{\leftrightarrow}
\newcommand\finright{\oldTo}
\newcommand\finleft{\leftarrow}
\newcommand\weakright{\hookrightarrow}
\newcommand\weakleft{\hookleftarrow}
\newcommand\strongright{\twoheadrightarrow}
\newcommand\strongleft{\twoheadleftarrow}
\newcommand\mrsright{\mathrel{\hookrightarrow^{\hspace{-10pt}m}\hspace{2pt}}}
\newcommand\mrsleft{\mathrel{\hookleftarrow^{\hspace{-7pt}m}\hspace{0pt}}}
\newcommand\prsright{\mathrel{\hookrightarrow^{\hspace{-8pt}p}\hspace{4pt}}}
\newcommand\prsleft{\mathrel{\hookleftarrow^{\hspace{-5pt}p}\hspace{0pt}}}

\newcommand{\RewArr}[2] {
  \RewStmt{#1}{\nothing}{#2}
}

\newcommand{\RewStmt}[3] {
  \def\RewArrArr{#1}
  \def\RewArrRhs{#2}
  \def\RewArrIter{#3}
  \RewArrI
}

\makeatletter \newcommand{\RewArrI}[1][\nothing] { \def\RewArrCxt{#1}
  \ifthenelse{\equal{\RewArrArr}{\oldTo} \OR
    \equal{\RewArrArr}{\preright} \OR
    \equal{\RewArrArr}{\finright} \OR
    \equal{\RewArrArr}{\finleftright} \OR
    \equal{\RewArrArr}{\prsright} \OR
    \equal{\RewArrArr}{\mrsright} \OR
    \equal{\RewArrArr}{\strongright} \OR
    \equal{\RewArrArr}{\weakright}} {
    \RewArrArr\ifthenelse{\equal{\RewArrIter}{\nothing}}{}{^{\RewArrIter}}\ifthenelse{\equal{\RewArrCxt}{\nothing}}{}{_{\RewArrCxt}}
  } { \ifthenelse{\equal{\RewArrArr}{\finleft} \OR
      \equal{\RewArrArr}{\prsleft} \OR
      \equal{\RewArrArr}{\mrsleft} \OR
      \equal{\RewArrArr}{\strongleft} \OR
      \equal{\RewArrArr}{\weakleft}}{
      \RewArrArr\ifthenelse{\equal{\RewArrIter}{\nothing}}{}{^{\RewArrIter}}\ifthenelse{\equal{\RewArrCxt}{\nothing}}{}{_{\RewArrCxt}}
    } { \latex@error{Rewrite arrow not defined}\@ehc } } \RewArrRhs }
\makeatother

\newcommand{\twoheadleftrightarrow}{\twoheadleftarrow\hspace{-7pt}\twoheadrightarrow}
\newcommand{\preto}{\RewArr{\preright}{\nothing}}
\renewcommand{\to}{\RewArr{\finright}{\nothing}}
\newcommand{\fto}{\RewArr{\finright}}
\newcommand{\wpato}{\RewArr{\prsright}{\nothing}}
\newcommand{\wpacont}{\RewStmt{\prsright}{\dots}{\nothing}}

\newcommand{\wmacont}{\RewStmt{\mrsright}{\dots}{\nothing}}

\newcommand{\wmato}{\RewArr{\mrsright}{\nothing}}
\newcommand{\join}[1][\nothing]
{
  \ifthenelse{\equal{#1}{\nothing}}{\downarrow}{\downarrow_{#1}}
}

\newcommand{\wconv}[1][\nothing]
{
  \ifthenelse{\equal{#1}{\nothing}}{\leftrightarrow^{w}}{\leftrightarrow^{w}_{#1}}
}

\newcommand{\sconv}[1][\nothing]
{
  \ifthenelse{\equal{#1}{\nothing}}{\twoheadleftrightarrow}{\twoheadleftrightarrow_{#1}}
}

\newcommand\cons{{\,::\,}}

\usepackage{subfig}

\theoremstyle{definition}\newtheorem{definition}[thm]{Definition}
\theoremstyle{plain}\newtheorem{theorem}[thm]{Theorem}
\theoremstyle{plain}\newtheorem{corollary}[thm]{Corollary}
\theoremstyle{plain}\newtheorem{lemma}[thm]{Lemma}
\theoremstyle{plain}\newtheorem{proposition}[thm]{Proposition}
\theoremstyle{definition}\newtheorem{example}[thm]{Example}
\theoremstyle{definition}\newtheorem{remark}[thm]{Remark}

\usepackage[obeyDraft,colorinlistoftodos]{todonotes}
\newcommand\continue[1][]{\todo{{\bf continue here}}}%

\begin{document}

\title[Modes of Convergence for Term Graph Rewriting]{Modes of
  Convergence for Term Graph Rewriting\rsuper*}

\author[Patrick Bahr]{Patrick Bahr}
\address{
  Department of Computer Science, University of Copenhagen (DIKU)
  \newline
  Universitetsparken 1, 2100 Copenhagen, Denmark}
\urladdr{http://www.diku.dk/\~{}paba}
\email{paba@diku.dk}

\thanks{} %optional

\titlecomment{{\lsuper*}Parts of this paper have appeared in the proceedings of
  RTA 2011 \cite{bahr11rta}.}%
%%%%%%%%%%%%%%%%%%%%%%%%%%%%%%%%%%%%%%%%%%%%%%%%%%%%%%%%%%%%%%%%%%%%%%%%%%%

%% the abstract has to PRECEED the command \maketitle:
%% be sure not to issue the \maketitle command twice!

\begin{abstract}
  \noindent 
  Term graph rewriting provides a simple mechanism to finitely
  represent restricted forms of infinitary term rewriting. The
  correspondence between infinitary term rewriting and term graph
  rewriting has been studied to some extent. However, this endeavour
  is impaired by the lack of an appropriate counterpart of infinitary
  rewriting on the side of term graphs. We aim to fill this gap by
  devising two modes of convergence based on a partial order respectively a
  metric on term graphs. The thus obtained structures generalise
  corresponding modes of convergence that are usually studied in
  infinitary term rewriting. 

  We argue that this yields a common framework in which both term
  rewriting and term graph rewriting can be studied. In order to
  substantiate our claim, we compare convergence on term graphs and on
  terms. In particular, we show that the modes of convergence on term
  graphs are conservative extensions of the corresponding modes of
  convergence on terms and are preserved under unravelling term graphs
  to terms. Moreover, we show that many of the properties known from
  infinitary term rewriting are preserved. This includes the intrinsic
  completeness of both modes of convergence and the fact that
  convergence via the partial order is a conservative extension of the
  metric convergence.
\end{abstract}

\maketitle

\section{Introduction}
\label{S:one}

Non-terminating computations are not necessarily undesirable. For
instance, the termination of a reactive system would be usually
considered a critical failure. Even computations that, given an input
$x$, should produce an output $y$ are not necessarily terminating in
nature either. For example, the various iterative approximation
algorithms for $\pi$ produce approximations of increasing accuracy
without ever terminating with the exact value of $\pi$. While such
iterative approximation computations might not reach the exact target
value, they are able to come arbitrary close to the correct value
within finite time.

It is this kind of non-terminating computations which is the subject
of infinitary term rewriting \cite{kennaway03book}. It extends the
theory of term rewriting by giving a meaning to transfinite reductions
instead of dismissing them as undesired and meaningless
artifacts. Following the paradigm of iterative approximations, the
result of a transfinite reduction is simply the term that is
approximated by the reduction. In general, such a result term can be
infinite. For example, starting from the term $rep(0)$, the rewrite
rule $rep(x) \to x \cons rep(x)$ produces a reduction
\[
rep(0) \to 0 \cons rep(0) \to 0 \cons 0 \cons rep(0) \to 0 \cons 0
\cons 0 \cons rep(0) \to \dots
\]
that approximates the infinite term $0 \cons 0 \cons 0 \cons
\dots$. Here, we use $\cons$ as a binary symbol that we write infix
and assume to associate to the right. That is, the term $0 \cons 0
\cons rep(0)$ is parenthesised as $0 \cons (0 \cons rep(0))$. Think of
the $\cons$ symbol as the list constructor \emph{cons}.

\emph{Term graphs}, on the other hand, allow us to explicitly
represent and reason about sharing and recursion \cite{ariola97ic} by
dropping the restriction to a tree structure, which we have for
terms. Apart from that, term graphs also provide a finite
representation of certain infinite terms, viz.\ \emph{rational
  terms}. As Kennaway et
al.~\cite{kennaway95segragra,kennaway94toplas} have shown, this can be
leveraged in order to finitely represent restricted forms of
infinitary term rewriting using \emph{term graph rewriting}.

In this paper, we extend the theory of infinitary term rewriting to
the setting of term graphs. To this end, we devise modes of
convergence that constrain reductions of transfinite length in a
meaningful way. Our approach to convergence is twofold: we generalise
the metric on terms that is used to define convergence for infinitary
term rewriting \cite{dershowitz91tcs} to term graphs. In a similar
way, we generalise the partial order on terms that has been recently
used to define a closely related notion of convergence for infinitary
term rewriting \cite{bahr10rta2}. The use of two different -- but on
terms closely related -- approaches to convergence will allow us both
to assess the appropriateness of the resulting infinitary calculi and
to compare them against the corresponding infinitary calculi of term
rewriting.

\subsection{Motivation}
\label{sec:motivation}

\subsubsection{Lazy Evaluation}
\label{sec:lazy-evaluation}

Term rewriting is a useful formalism for studying declarative
programs, in particular, functional programs. A functional program
essentially consists of functions defined by a set of equations and an
expression that is supposed to be evaluated according to these
equations. The conceptual process of evaluating an expression is
nothing else than term rewriting.

A particularly interesting feature of modern functional programming
languages, such as Haskell~\cite{marlow10haskell}, is the ability to
use \emph{conceptually} infinite computations and data structures. For
example, the following definition of a function \texttt{from}
constructs for each number $n$ the infinite list of consecutive
numbers starting from $n$:
\begin{center}
  \verb|from(n) = n :: from(s(n))|
\end{center}
Here, we use the binary infix symbol \texttt{::} to denote the list
constructor \emph{cons} and \texttt{s} for the successor
function. While we cannot use the infinite list generated by
\texttt{from} directly -- the evaluation of an expression of the form
\texttt{from n} does not terminate -- we can use it in a setting in
which we only read a finite prefix of the infinite list conceptually
defined by \texttt{from}. Functional languages such as Haskell allow
this use of semantically infinite data structures through a
\emph{non-strict evaluation} strategy, which delays the evaluation of
a subexpression until its result is actually required for further
evaluation of the expression. This non-strict semantics is not only a
conceptual neatness but in fact one of the major features that make
functional programs highly modular~\cite{hughes89cj}.

The above definition of the function \texttt{from} can be represented
as a term rewriting system with the following rule:
\[
from(x) \to x \cons from(s(x))
\]
Starting with the term $from(0)$, we then obtain the following
infinite reduction:
\[
from(0) \to 0 \cons from(s(0)) \to 0 \cons s(0) \cons from(s(s(0))) \to
\dots
\]
Infinitary term rewriting~\cite{kennaway03book} provides a notion of
convergence that may assign a meaningful result term to such an
infinite reduction provided there exists one. In this sense, the above
reduction converges to the infinite term $0 \cons s(0) \cons s(s(0))
\cons \dots$, which represents the infinite list of numbers $0, 1, 2,
\dots$. Due to this extension of term rewriting with explicit limit
constructions for non-terminating reductions, infinitary term
rewriting allows us to directly reason about non-terminating functions
and infinite data structures.

Non-strict evaluation is rarely found unescorted, though. Usually, it
is implemented as \emph{lazy evaluation}~\cite{henderson76popl}, which
complements a non-strict evaluation strategy with \emph{sharing}. The
latter avoids duplication of subexpressions by using pointers instead
of copying. For example, the function \texttt{from} above duplicates
its argument \texttt{n} -- it occurs twice on the right-hand side of
the defining equation. A lazy evaluator simulates this duplication by
inserting two pointers pointing to the actual argument. Sharing is a
natural companion for non-strict evaluation as it avoids re-evaluation
of expressions that are duplicated before they are evaluated.

The underlying formalism that is typically used to obtain sharing for
functional programming languages is term graph
rewriting~\cite{jones87book,plasmeijer93book}. Term graph
rewriting~\cite{barendregt87parle,plump99hggcbgt} uses graphs to
represent terms thus allowing multiple arcs to point to the same
node. For example, term graphs allows us to change the representation
of the term rewrite rule defining the function $from$ by replacing
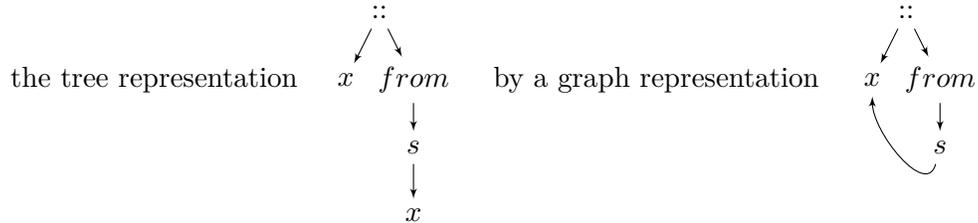
\begin{center}
  \begin{tikzpicture}[
  level distance=9mm,%
  sibling distance = 9mm]
    \node (c) {$\cons$}%
    child{%
      node (x2) {$x$}%
    } child {%
      node (from) {$from$}%
      child {%
        node (s) {$s$}%
        child {%
          node {$x$}%
        }%
      }%
    };%
    \node [left=5mm] at (x2) {the tree representation};
    \node[node distance=7cm,right of=c] {$\cons$}%
    child{%
      node (x) {$x$}%
    } child {%
      node {$from$}%
      child {%
        node (s) {$s$}%
      }%
    };%
    \draw (s) edge[->,out=-105,in=-90] (x); \path (from) -- (x)
    node[midway] {by a graph representation};
  \end{tikzpicture}
\end{center}
which shares the variable $x$ by having two arcs pointing to it.

While infinitary term rewriting is used to model the non-strictness of
lazy evaluation, term graph rewriting models the sharing part of it.
By endowing term graph rewriting with a notion of convergence, we aim
to unify the two formalisms into one calculus, thus allowing us to
model both aspects withing the same calculus.

\subsubsection{Rational Terms}
\label{sec:rational-terms}

Term graphs can do more than only share common subexpressions. Through
cycles term graphs may also provide a finite representation of certain
infinite terms -- so-called \emph{rational terms}. For example, the
infinite term $0 \cons 0 \cons 0 \cons \dots$ can be represented as
the finite term graph
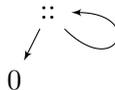
\begin{center}
  \begin{tikzpicture}[
  level distance=9mm,%
  sibling distance = 9mm]
    \node (c) {$\cons$}%
    child {%
      node {$0$}%
    } child [missing];%
    \draw (c) edge[->,out=-45,in=0,min distance=1cm] (c);%
  \end{tikzpicture}
\end{center}
Since a single node on a cycle in a term graph represents infinitely
many corresponding subterms, the contraction of a single term graph
redex may correspond to a transfinite term reduction that contracts
infinitely many term redexes. For example, if we apply the rewrite
rule $0 \to s(0)$ to the above term graph, we obtain a term graph that
represents the term $s(0)\cons s(0) \cons s(0) \cons \dots$, which can
only be obtained from the term $0 \cons 0 \cons 0 \cons \dots$ via a
\emph{transfinite} term reduction with the rule $0 \to s(0)$. Kennaway
et al.~\cite{kennaway94toplas} investigated this correspondence
between cyclic term graph rewriting and infinitary term
rewriting. Among other results they characterise a subset of
transfinite term reductions -- called \emph{rational reductions} --
that can be simulated by a corresponding finite term graph
reduction. The above reduction from the term $0 \cons 0 \cons 0 \cons
\dots$ is an example of such a rational reduction.

With the help of a unified formalism for infinitary and term graph
rewriting, it should be easier to study the correspondence between
infinitary term rewriting and finitary term graph rewriting
further. The move from an infinitary term rewriting system to a term
graph rewriting system only amounts to a change in the degree of
sharing if we use infinitary term graph rewriting as a common
framework.

Reconsider the term rewrite rule $rep(x) \to x \cons rep(x)$, which
defines a function $rep$ that repeats its argument infinitely often:
\[
rep(0) \to 0 \cons rep(0) \to 0 \cons 0 \cons rep(0) \to  0 \cons 0
\cons 0 \cons rep(0) \to \quad \dots \quad 0 \cons 0 \cons 0 \cons \dots
\]
This reduction happens to be not a rational reduction in the sense of
Kennaway et al.~\cite{kennaway94toplas}.

The move from the term rule $rep(x) \to x \cons rep(x)$ to a term
graph rule is a simple matter of introducing sharing of common
subexpressions:
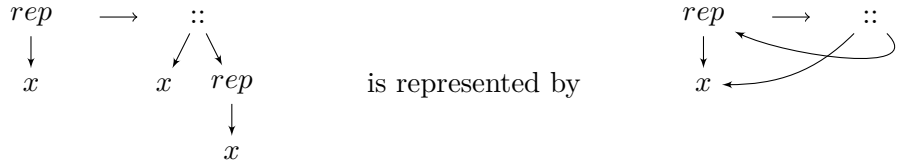
\begin{center}
  \begin{tikzpicture}[
  level distance=9mm,%
  sibling distance = 9mm]
    \node (l) {$rep$}%
    child {%
      node (x) {$x$}%
    };%

    \node[node distance=1.5cm,right=of l]
    (r) {$\cons$}%
    child {%
      node (b) {$x$}%
    } child {%
      node (c) {$rep$}%
      child {%
        node (a') {$x$}%
      }%
    };%
    
    \draw (l) edge[single step,shorten=5mm] (r);%
    
    \node [node distance=6cm,right=of r]
    (l) {$rep$}%
    child {%
      node (x) {$x$}%
    };%
    
    \path (c) -- (x) node[midway] {is represented by};

    \node[node distance=1.5cm,right=of l]
    (r) {$\cons$};%
    \draw%
    (r) edge[->,out=-45,in=-25] (l)%
    (r) edge[->,out=-135,in=0] (x);%
    \draw (l) edge[single step,shorten=5mm] (r);%
  \end{tikzpicture}
\end{center}
Instead of creating a fresh copy of the redex on the right-hand side,
the redex is reused by placing an edge from the right-hand side of the
rule to its left-hand side. This allows us to represent the infinite
reduction approximating the infinite term $0 \cons 0 \cons 0 \cons
\dots$ with the following single step term graph reduction induced by
the above term graph rule:
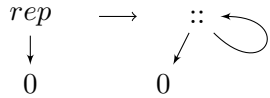
\begin{center}
  \begin{tikzpicture}[
  level distance=9mm,%
  sibling distance = 9mm]
    \node [node distance=4cm,right=of r]
    (l) {$rep$}%
    child {%
      node (x) {$0$}%
    };%

    \node[node distance=1.5cm,right=of l]%
    (r) {$\cons$}%
    child {%
      node (b) {$0$}%
    } child [missing];%
    \draw%
    (r) edge[min distance=1cm,out=-45,in=0,->] (r);%
    \draw (l) edge[single step,shorten=5mm] (r);%
  \end{tikzpicture}
\end{center}
Via its cyclic structure the resulting term graph represents the
infinite term $0 \cons 0 \cons 0 \cons \dots$.

Since both transfinite term reductions and the corresponding finite
term graph reductions can be treated within the same formalism, we
hope to provide a tool for studying the ability of cyclic term graph
rewriting to finitely represent transfinite term reductions.

\subsection{Contributions \& Related Work}
\label{sec:contr-relat}

\subsubsection{Contributions}
\label{sec:contributions}

The main contributions of this paper are the following:
\begin{enumerate}[(i)]
\item We devise a partial order on term graphs based on a restricted
  class of graph homomorphisms. We show that this partial order forms
  a complete semilattice and thus is technically suitable for defining
  a notion of convergence (Theorem~\ref{thm:graphBcpo}). Moreover, we
  illustrate alternative partial orders and show why they are not
  suitable for formalising convergence on term graphs.
\item Independently, we devise a metric on term graphs and show that
  it forms a complete ultrametric space on term graphs
  (Theorem~\ref{thr:metricGraphCompl}).
\item Based on the partial order respectively the metric we define a
  notion of \emph{weak convergence} for infinitary term graph
  rewriting. We show that -- similar to the term rewriting case
  \cite{bahr10rta2} -- the metric calculus of infinitary term graph
  rewriting is the \emph{total fragment} of the partial order calculus
  of infinitary term graph rewriting
  (Theorem~\ref{thr:graphTotalConv}).
\item We confirm that the partial order and the metric on term graphs
  generalise the partial order respectively the metric that is used
  for infinitary term rewriting
  (Proposition~\ref{prop:liminfGeneralise} and
  \ref{prop:limGeneralise}). Moreover, we show that the
  corresponding notions of convergence are preserved by unravelling
  term graphs to terms thus establishing the soundness of our notions
  of convergence on term graphs w.r.t.\ the convergence on terms
  (Theorems~\ref{thr:liminfUnrav} and \ref{thr:limUnrav}).
\item We substantiate the appropriateness of our calculi by a number
  of examples that illustrate how increasing the sharing gradually
  reduces the number of steps necessary to reach the final result --
  eventually, from an infinite number of steps to a finite number
  (Sections~\ref{sec:infin-term-graph} and \ref{sec:soundn--compl}).
\end{enumerate}

\subsubsection{Related Work}
\label{sec:related-work}

Calculi with explicit sharing and/or recursion, e.g.\ via
\emph{letrec}, can also be considered as a form of term graph
rewriting. Ariola and Klop~\cite{ariola97ic} recognised that adding
such an explicit recursion mechanism to the lambda calculus may break
confluence. In order to reconcile this, Ariola and
Blom~\cite{ariola02apal,ariola05ptc} developed a notion of skew
confluence that allows them to define an infinite normal form in the
vein of Böhm trees.

Recently, we have investigated other notions of convergence for term
graph rewriting~\cite{bahr12rta,bahr12mscs} that use simpler variants
of the partial order and the metric that we use in this paper. Both of
them have theoretically pleasing properties, e.g.\ the ideal
completion and the metric completion of the set of finite term graphs
both yield the set of all term graphs. However, the resulting notions
of weak convergence are not fully satisfying and in fact
counterintuitive for some cases. We will discuss this alternative
approach and compare it to the present approach in more detail in
Sections~\ref{sec:partial-order-lebot1} and
\ref{sec:alternative-metric}.

\subsection{Overview}
\label{sec:overview}
The structure of this paper is as follows: in
Section~\ref{sec:preliminaries}, we provide the necessary background
for metric spaces, partially ordered sets and term rewriting. In
Section~\ref{sec:infin-term-rewr}, we give an overview of infinitary
term rewriting.  Section~\ref{sec:term-graphs} provides the necessary
theory for graphs and term
graphs. Sections~\ref{sec:partial-order-lebot1} and
\ref{sec:alternative-metric} form the core of this paper. In these
sections we study the partial order and the metric on term graphs that
are the basis for the modes of convergence we propose in this
paper. In Section~\ref{sec:metric-vs.-partial}, we then compare the
two resulting modes of convergence. Moreover, in
Section~\ref{sec:infin-term-graph}, we use these two modes of
convergence to study two corresponding infinitary term graph rewriting
calculi. In Section~\ref{sec:soundn--compl}, we study the
correspondences between infinitary term graph rewriting and infinitary
term rewriting.

Some proofs have been omitted from the main body of the text. These
proofs can be found in the appendix of this paper.
\tableofcontents

\section{Preliminaries}
\label{sec:preliminaries}

We assume the reader to be familiar with the basic theory of ordinal
numbers, orders and topological spaces \cite{kelley55book}, as well as
term rewriting \cite{terese03book}. In order to make this paper
self-contained, we briefly recall all necessary preliminaries.

\subsection{Sequences}
We use the von Neumann definition of ordinal numbers. That is, an
\emph{ordinal number} (or simply \emph{ordinal}) $\alpha$ is the set
of all ordinal numbers strictly smaller than $\alpha$. In particular,
each natural number $n\in \nats$ is an ordinal number with $n =
\set{0,1,\dots,n - 1}$. The least infinite ordinal number is denoted
by $\omega$ and is the set of all natural numbers. Ordinal numbers
will be denoted by lower case Greek letters $\alpha, \beta, \gamma,
\lambda, \iota$.

A \emph{sequence} $S$ of \emph{length} $\alpha$ in a set $A$, written
$(a_\iota)_{\iota < \alpha}$, is a function from $\alpha$ to $A$ with
$\iota \mapsto a_\iota$ for all $\iota \in \alpha$. We use $\len{S}$
to denote the length $\alpha$ of $S$. If $\alpha$ is a limit ordinal,
then $S$ is called \emph{open}. Otherwise, it is called
\emph{closed}. If $\alpha$ is a finite ordinal, then $S$ is called
\emph{finite}. Otherwise, it is called \emph{infinite}. For a finite
sequence $(a_i)_{i < n}$ we also use the notation
$\seq{a_0,a_1,\dots,a_{n-1}}$. In particular, $\emptyseq$ denotes the
empty sequence. We write $A^*$ for the set of all finite sequences in
$A$.

The \emph{concatenation} $(a_\iota)_{\iota<\alpha}\concat
(b_\iota)_{\iota<\beta}$ of two sequences is the sequence
$(c_\iota)_{\iota<\alpha+\beta}$ with $c_\iota = a_\iota$ for $\iota <
\alpha$ and $c_{\alpha+\iota} = b_\iota$ for $\iota < \beta$. A
sequence $S$ is a (proper) \emph{prefix} of a sequence $T$, denoted $S
\le T$ (respectively $S < T$), if there is a (non-empty) sequence $S'$
with $S\concat S' = T$. The prefix of $T$ of length $\beta\le\len{T}$
is denoted $\prefix{T}{\beta}$. Similarly, a sequence $S$ is a
(proper) \emph{suffix} of a sequence $T$ if there is a (non-empty)
sequence $S'$ with $S'\concat S = T$.

\subsection{Metric Spaces}
\label{sec:metric-spaces}

A pair $(M,\dd)$ is called a \emph{metric space} if $\dd$ is a
\emph{metric} on the set $M$. That is, $\dd \fcolon M \times M \to
\realsnn$ is a function satisfying $\dd(x,y) = 0$ iff $x=y$
(identity), $\dd(x, y) = \dd(y, x)$ (symmetry), and $\dd(x, z) \le
\dd(x, y) + \dd(y, z)$ (triangle inequality), for all $x,y,z\in M$.
If $\dd$ instead of the triangle inequality, satisfies the stronger
property $\dd(x, z) \le \max \set{ \dd(x, y),\dd(y, z)}$ (strong
triangle), then $(M,\dd)$ is called an \emph{ultrametric space}.

Let $(a_\iota)_{\iota<\alpha}$ be a sequence in a metric space
$(M,\dd)$. The sequence $(a_\iota)_{\iota<\alpha}$ \emph{converges} to
an element $a\in M$, written $\lim_{\iota\limto\alpha} a_\iota$, if,
for each $\varepsilon \in \realsp$, there is a $\beta < \alpha$ such
that $\dd(a,a_\iota) < \varepsilon$ for every $\beta < \iota <
\alpha$; $(a_\iota)_{\iota<\alpha}$ is \emph{continuous} if
$\lim_{\iota\limto\lambda} a_\iota = a_\lambda$ for each limit ordinal
$\lambda < \alpha$. Intuitively speaking, $(a_\iota)_{\iota<\alpha}$
converges to $a$ if the metric distance between the elements $a_\iota$
of the sequence and $a$ tends to $0$ as the index $\iota$ approaches
$\alpha$, i.e.\ they approximate $a$ arbitrarily well. Accordingly,
$(a_\iota)_{\iota<\alpha}$ is continuous if it does not leap to a
distant object at limit ordinal indices.

The sequence $(a_\iota)_{\iota<\alpha}$ is called \emph{Cauchy} if,
for any $\varepsilon \in \realsp$, there is a $\beta<\alpha$ such
that, for all $\beta < \iota < \iota' < \alpha$, we have that
$\dd(m_\iota,m_{\iota'}) < \varepsilon$. That is, the elements
$a_\iota$ of the sequence move closer and closer to each other as the
index $\iota$ approaches $\alpha$.

A metric space is called \emph{complete} if each of its non-empty
Cauchy sequences converges. That is, whenever the elements $a_\iota$
of a sequence move closer and closer together, they in fact
approximate an existing object of the metric space, viz.\
$\lim_{\iota\limto\alpha} a_\iota$.

\subsection{Partial Orders}
\label{sec:partial-orders}
A \emph{partial order} $\le$ on a set $A$ is a binary relation on $A$
such that $x \le y, y \le z$ implies $x\le z$ (transitivity); $x \le
x$ (reflexivity); and $x \le y, y \le x$ implies $x = y$
(antisymmetry) for all $x,y,z \in A$. The pair $(A,\le)$ is then
called a \emph{partially ordered set}. A subset $D$ of the underlying
set $A$ is called \emph{directed} if it is non-empty and each pair of
elements in $D$ has an upper bound in $D$. A partially ordered set
$(A, \le)$ is called a \emph{complete partial order} (\emph{cpo}) if
it has a least element and each directed set $D$ has a \emph{least
  upper bound} (\emph{lub}) $\Lub D$. A cpo $(A, \le)$ is called a
\emph{complete semilattice} if every \emph{non-empty} set $B$ has
\emph{greatest lower bound} (\emph{glb}) $\Glb B$. In particular, this
means that, in a complete semilattice, the \emph{limit inferior} of
any sequence $(a_\iota)_{\iota<\alpha}$, defined by $\liminf_{\iota
  \limto \alpha}a_\iota = \Lub_{\beta<\alpha} \left(\Glb_{\beta \le
    \iota < \alpha} a_\iota\right)$, always exists.

There is also an alternative characterisation of complete
semilattices: a partially ordered set $(A, \le)$ is called
\emph{bounded complete} if each set $B \subseteq A$ that has an upper
bound in $A$ also has a least upper bound in $A$. Two elements $a,b\in
A$ are called \emph{compatible} if they have a common upper bound,
i.e.\ there is some $c \in A$ with $a,b \le c$.
\begin{proposition}[bounded complete cpo = complete semilattice,
  \cite{kahn93tcs}]
  \label{prop:bcpoCompSemi}
  Given a cpo $(A,\le)$ the following are equivalent:
  \begin{enumerate}[\em(i)]
  \item $(A,\le)$ is a complete semilattice.
  \item $(A,\le)$ is bounded complete.
  \item Each two compatible elements in $A$ have a least upper bound.\qed
  \end{enumerate}
\end{proposition}

Given two partially ordered sets $(A,\le_A)$ and $(B,\le_B)$, a
function $\phi\fcolon A \to B$ is called \emph{monotonic} iff $a_1
\le_A a_2$ implies $\phi(a_1) \le_B \phi(a_2)$. In particular, we have
that a sequence $(b_\iota)_{\iota<\alpha}$ in $(B,\le_B)$ is monotonic
if $b_\iota \le_B b_\gamma$ for all $\iota\le \gamma<\alpha$.
\subsection{Terms}
\label{sec:terms}

Since we are interested in the infinitary calculus of term rewriting,
we consider the set $\iterms$ of \emph{infinitary terms} (or simply
\emph{terms}) over some \emph{signature} $\Sigma$. A \emph{signature}
$\Sigma$ is a countable set of symbols such that each symbol
$f\in\Sigma$ is associated with an arity $\srank{f}\in \nats$, and we
write $\Sigma^{(n)}$ for the set of symbols in $\Sigma$ that have
arity $n$. The set $\iterms$ is defined as the \emph{greatest} set $T$
such that $t \in T$ implies $t = f(t_1,\dots, t_k)$ for some $f \in
\Sigma^{(k)}$ and $t_1,\dots,t_k\in T$. For each constant symbol $c\in
\Sigma^{(0)}$, we write $c$ for the term $c()$. For a term $t \in
\iterms$ we use the notation $\pos{t}$ to denote the \emph{set of
  positions} in $t$. $\pos{t}$ is the least subset of $\nats^{*}$ such
that $\emptyseq \in \pos{t}$ and $\seq{i}\concat\pi \in \pos{t}$ if $t
= f(t_0,\dots,t_{k-1})$ with $0 \le i < k$ and $\pi \in
\pos{t_i}$. For terms $s,t \in \iterms$ and a position $\pi \in
\pos{t}$, we write $\atPos{t}{\pi}$ for the \emph{subterm} of $t$ at
$\pi$, $t(\pi)$ for the function symbol in $t$ at $\pi$, and
$\substAtPos{t}{\pi}{s}$ for the term $t$ with the subterm at $\pi$
replaced by $s$. As positions are sequences, we use the prefix order
$\le$ defined on them. A position is also called an \emph{occurrence}
if the focus lies on the subterm at that position rather than the
position itself. The set $\terms$ of \emph{finite terms} is the subset
of $\iterms$ that contains all terms with a finite set of positions.

On $\iterms$ a similarity measure $\similar{\cdot}{\cdot}\fcolon
\iterms \times \iterms \funto \omega+1$ is defined as follows
\[
\similar{s}{t} = \min \setcom{\len{\pi}}{\pi \in \pos{s}\cap\pos{t}, s(\pi) \neq
  t(\pi)} \cup \set{\omega} \qquad \text {for } s,t\in \iterms
\]
That is, $\similar{s}{t}$ is the minimal depth at which $s$ and $t$
differ, respectively $\omega$ if $s = t$. Based on this similarity
measure, a distance function $\dd$ is defined by $\dd(s,t) =
2^{-\similar{s}{t}}$, where we interpret $2^{-\omega}$ as $0$. The
pair $(\iterms, \dd)$ is known to form a complete ultrametric space
\cite{arnold80fi}.

\emph{Partial terms}, i.e.\ terms over signature $\Sigma_\bot = \Sigma
\uplus \set{\bot}$ with $\bot$ a fresh nullary symbol, can be endowed
with a binary relation $\lebot$ by defining $s \lebot t$ iff $s$ can
be obtained from $t$ by replacing some subterm occurrences in $t$ by
$\bot$. Interpreting the term $\bot$ as denoting ``undefined'',
$\lebot$ can be read as ``is less defined than''. The pair
$(\ipterms,\lebot)$ is known to form a complete semilattice
\cite{goguen77jacm}. To explicitly distinguish them from partial
terms, we call terms in $\iterms$ \emph{total}.

\subsection{Term Rewriting Systems}
\label{sec:abstr-reduct-syst}

For term rewriting systems, we have to consider terms with
variables. To this end, we assume a countably infinite set $\calV$ of
variables and extend a signature $\Sigma$ to a signature $\Sigma_\calV
= \Sigma \uplus \calV$ with variables in $\calV$ as nullary
symbols. Instead of $\iterms[\Sigma_\calV]$ we also write $\ivterms$.
A \emph{term rewriting system} (TRS) $\calR$ is a pair $(\Sigma, R)$
consisting of a signature $\Sigma$ and a set $R$ of \emph{term rewrite
  rules} of the form $l \to r$ with $l \in \ivterms \setminus \calV$
and $r \in \ivterms$ such that all variables occurring in $r$ also
occur in $l$. Note that both the left- and the right-hand side may be
infinite. We usually use $x,y,z$ and primed respectively indexed variants
thereof to denote variables in $\calV$. A \emph{substitution} $\sigma$
is a mapping from $\calV$ to $\ivterms$. Such a substitution $\sigma$
can be uniquely lifted to a homomorphism from $\ivterms$ to $\ivterms$
mapping a term $t\in\ivterms$ to $t\sigma$ by setting $x\sigma =
\sigma(x)$ if $x\in\calV$ and $f(t_1,\dots,t_n)\sigma =
f(t_1\sigma,\dots,t_n\sigma)$ if $f \in \Sigma^{(n)}$.

As in the finitary setting, every TRS $\calR$ defines a \emph{rewrite
  relation} $\to[\calR]$ that indicates \emph{rewrite steps}:
\[
s \to[\calR] t \iff \exists \pi \in \pos{s}, l\to r \in
R, \sigma\colon\; \atPos{s}{\pi} = l\sigma, t = \substAtPos{s}{\pi}{r\sigma}
\]
Instead of $s \to[\calR] t$, we sometimes write $s \to[\pi,\rho] t$ in
order to indicate the applied rule $\rho$ and the position $\pi$, or
simply $s \to t$. The subterm $\atPos{s}{\pi}$ is called a
\emph{$\rho$-redex} or simply \emph{redex}, $r\sigma$ its
\emph{contractum}, and $\atPos{s}{\pi}$ is said to be
\emph{contracted} to $r\sigma$.

\section{Infinitary Term Rewriting}
\label{sec:infin-term-rewr}

Before pondering over the right approach to an infinitary calculus of
term graph rewriting, we want to provide a brief overview of
infinitary term
rewriting~\cite{kennaway03book,bahr10rta2,blom04rta}. This should give
an insight into the different approaches to dealing with infinite
reductions. However, in contrast to the majority of the literature on
infinitary term rewriting, which is concerned with strong convergence~
\cite{kennaway03book,kennaway95ic}, we will only consider weak notions
of convergence in this paper; cf.~\cite{dershowitz91tcs, kahrs07ai,
  simonsen10rta}. This weak form of convergence, also called
\emph{Cauchy convergence}, is entirely based on the sequence of
objects produced by rewriting without considering \emph{how} the
rewrite rules are applied.

A \emph{(transfinite) reduction} in a term rewriting system $\calR$,
is a sequence $S = (t_\iota \to_\calR t_{\iota +1})_{\iota < \alpha}$
of rewrite steps in $\calR$. Note that the underlying sequence of
terms $(t_\iota)_{\iota<\wsuc\alpha}$ has length $\wsuc\alpha$, where
$\wsuc\alpha = \alpha$ if $S$ is open, and $\wsuc\alpha = \alpha + 1$
if $S$ is closed. The reduction $S$ is called \emph{$\mrs$-continuous}
in $\calR$, written $S\fcolon t_0 \wmacont[\calR]$, if the sequence of
terms $(t_\iota)_{\iota < \wsuc\alpha}$ is continuous in
$(\iterms,\dd)$, i.e.\ $\lim_{\iota\limto\lambda} t_\iota = t_\lambda$
for each limit ordinal $\lambda < \alpha$. The reduction $S$ is said
to \emph{$\mrs$-converge} to a term $t$ in $\calR$, written $S\fcolon
t_0 \wmato[\calR] t$, if it is $\mrs$-continuous and
$\lim_{\iota\limto\wsuc\alpha} t_\iota = t$.

\begin{example}
  \label{ex:termRewr}
  Consider the term rewriting system $\calR$ containing the rule
  $\rho_1\fcolon a \cons x \to b \cons a \cons x$. By repeatedly
  applying $\rho_1$, we obtain the infinite reduction
  \[
  S\fcolon a \cons c \to b \cons a \cons c \to b \cons b \cons a \cons c \to \dots
  \]
  The position at which two consecutive terms differ moves deeper and
  deeper during the reduction $S$, i.e.\ the $\dd$-distance between
  them tends to $0$. Hence, $S$ $\mrs$-converges to the infinite term
  $s = b \cons b \cons b \cons \dots$, i.e.\ $S\fcolon a \cons c
  \wmato s$.

  Now consider a TRS with the slightly different rule $\rho_2\fcolon a
  \cons x \to a \cons b \cons x$. This TRS yields a reduction
  \[
  S'\fcolon a \cons c \to a \cons b \cons c \to a \cons b \cons b \cons c \to \dots
  \]
  Even though the rule $\rho_2$ is applied at the root of the term in
  each step of $S'$, the $\dd$-distance between two consecutive terms
  tends to $0$ again. The reduction $S'$ $\mrs$-converges to the
  infinite term $s' = a \cons b \cons b \cons \dots$, i.e.\ $S'\fcolon
  a \cons c \wmato s'$.
\end{example}

In contrast to the weak $\mrs$-convergence that we consider here,
strong $\mrs$-convergence~\cite{kennaway03book,kennaway95ic}
additionally requires that the depth of the contracted redexes tends
to infinity as the reduction approaches a limit ordinal. Concerning
Example~\ref{ex:termRewr} above, we have for instance that $S$ also
strongly $\mrs$-converges -- the rule is applied at increasingly deep
redexes -- whereas $S'$ does not strongly $\mrs$-converge -- each step
in $S'$ results from a contraction at the root.

In the partial order model of infinitary rewriting~\cite{bahr10rta2},
convergence is defined via the limit inferior in the complete
semilattice $(\ipterms,\lebot)$. Given a TRS $\calR = (\Sigma,R)$, we
extend it to $\calR_\bot = (\Sigma_\bot,R)$ by adding the fresh
constant symbol $\bot$ such that it admits all terms in $\ipterms$. A
reduction $S = (t_\iota \to_{\calR_\bot} t_{\iota +1})_{\iota <
  \alpha}$ in this system $\calR_\bot$ is called
\emph{$\prs$-continuous} in $\calR$, written $S\fcolon t_0
\wpacont[\calR]$, if $\liminf_{\iota \limto \lambda} t_\iota =
t_\lambda$ for each limit ordinal $\lambda < \alpha$. The reduction
$S$ is said to \emph{$\prs$-converge} to a term $t$ in $\calR$,
written $S\fcolon t_0 \wpato[\calR] t$, if it is $\prs$-continuous and
$\liminf_{\iota\limto\wsuc\alpha} t_\iota = t$.

The distinguishing feature of the partial order approach is that each
continuous reduction also converges due to the semilattice structure
of partial terms.  Moreover, $\prs$-convergence provides a
conservative extension to $\mrs$-convergence that allows rewriting
modulo \emph{meaningless terms}~\cite{bahr10rta2} by essentially
mapping those parts of the reduction to $\bot$ that are divergent
according to the metric mode of convergence.

Intuitively, the limit inferior in $(\ipterms,\lebot)$ -- and thus
$\prs$-convergence -- describes an approximation process that
accumulates each piece of information that remains \emph{stable} from
some point onwards. This is based on the ability of the partial order
$\lebot$ to capture a notion of \emph{information preservation}, i.e.\
$s \lebot t$ iff $t$ contains at least the same information as $s$
does but potentially more. A monotonic sequence of terms $t_0 \lebot
t_1 \lebot \dots$ thus approximates the information contained in
$\Lub_{i<\omega} t_i$. Given this reading of $\lebot$, the glb $\Glb
T$ of a set of terms $T$ captures the common (non-contradicting)
information of the terms in $T$. Leveraging this observation, a
sequence that is not necessarily monotonic can be turned into a
monotonic sequence $t_j = \Glb_{j \le i < \omega} s_i$ such that each
$t_j$ contains exactly the information that remains stable in
$(s_i)_{i<\omega}$ from $j$ onwards. Hence, the limit inferior
$\liminf_{i \limto \omega} s_i = \Lub_{j < \omega}\Glb_{j \le i <
  \omega} s_i$ is the term that contains the accumulated information
that eventually remains stable in $(s_i)_{i<\omega}$. This is
expressed as an approximation of the monotonically increasing
information that remains stable from some point on.

\begin{example}
  \label{ex:termRewr2}
  Reconsider the system from Example~\ref{ex:termRewr}. The reduction
  $S$ also $\prs$-converges to $s$. This can be seen by forming the
  sequence $(\Glb_{j \le i < \omega} s_i)_{i < \omega}$ of stable
  information of the underlying sequence $(s_i)_{i<\omega}$ of terms
  in $S$:
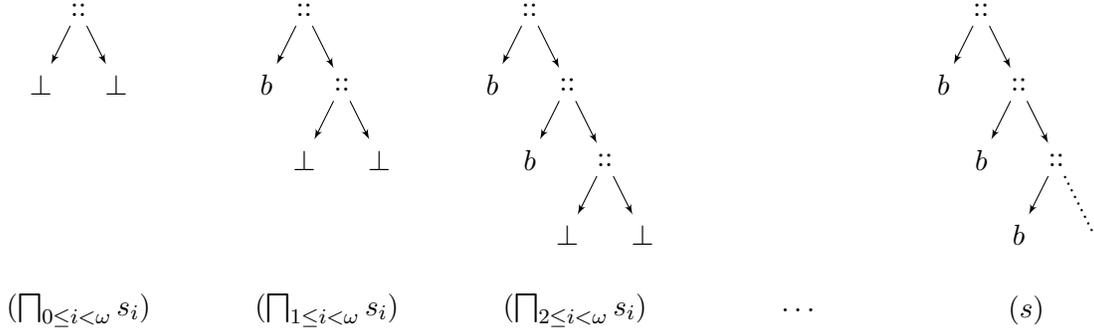
  \begin{center}
    \begin{tikzpicture}[node distance=3cm]
      \node (s0) {$\cons$}%
      child {%
        node {$\bot$}%
      } child {%
        node {$\bot$}%
      };%
      \node[right of=s0] (s1) {$\cons$}%
      child {%
        node {$b$}%
      } child {%
        node {$\cons$}% 
        child {%
          node {$\bot$}%
        } child {%
          node {$\bot$}%
        }
      };%
      \node[right of=s1] (s2) {$\cons$}%
      child {%
        node {$b$}%
      } child {%
        node {$\cons$}% 
        child {%
          node {$b$}%
        } child {% 
          node {$\cons$}% 
          child {%
            node {$\bot$}%
          } child {%
            node {$\bot$}%
          }
        }
      };%

      \node[node distance=6cm, right of=s2] (s) {$\cons$}%
      child {%
        node {$b$}%
      } child {%
        node {$\cons$}% 
        child {%
          node {$b$}%
        } child {% 
          node {$\cons$}% 
          child {%
            node {$b$}%
          } child [etc]
        }
      };%
      \foreach \i/\d in {0/0,1/3,2/6}{%
        \node (i\i) at ($(s\i)+(\d mm,-4)$) {$(\Glb_{\i\le i<\omega} s_i)$};%
      }%
      \node[right of=i2] {$\dots$};
      
      \node (i) at ($(s)+(6 mm,-4)$) {$(s)$};%

    \end{tikzpicture}
  \end{center}
  This sequence approximates the term $s = b\cons b\cons b \cons
  \dots$.
  
  Now consider the rule $\rho_1$ together with the rule
  $\rho_3\fcolon b \cons x \to a \cons b \cons x$. Starting with the
  same term, but applying the two rules alternately at the root, we
  obtain the reduction sequence
  \[
  T\fcolon a \cons c \to b \cons a \cons c \to a \cons b \cons a \cons c \to b \cons a \cons b \cons a \cons c \to \dots
  \]
  Now the differences between two consecutive terms occur right below
  the root symbol ``$\cons$''. Hence, $T$ does not
  $\mrs$-converge. This, however, only affects the left argument of
  each ``$\cons$''. Following the right argument position, the bare
  list structure becomes eventually stable. The sequence $(\Glb_{j
    \le i < \omega} s_i)_{i < \omega}$ of stable information
  
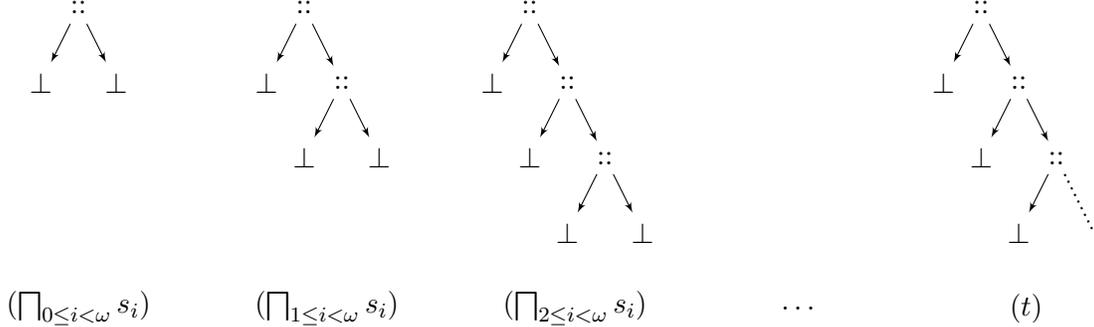
\begin{center}
    \begin{tikzpicture}[node distance=3cm]
      \node (s0) {$\cons$}%
      child {%
        node {$\bot$}%
      } child {%
        node {$\bot$}%
      };%
      \node[right of=s0] (s1) {$\cons$}%
      child {%
        node {$\bot$}%
      } child {%
        node {$\cons$}% 
        child {%
          node {$\bot$}%
        } child {%
          node {$\bot$}%
        }
      };%
      \node[right of=s1] (s2) {$\cons$}%
      child {%
        node {$\bot$}%
      } child {%
        node {$\cons$}% 
        child {%
          node {$\bot$}%
        } child {% 
          node {$\cons$}% 
          child {%
            node {$\bot$}%
          } child {%
            node {$\bot$}%
          }
        }
      };%
      \node[node distance=6cm,right of=s2] (s) {$\cons$}%
      child {%
        node {$\bot$}%
      } child {%
        node {$\cons$}% 
        child {%
          node {$\bot$}%
        } child {% 
          node {$\cons$}% 
          child {%
            node {$\bot$}%
          } child[etc]
        }
      };%
      \foreach \i/\d in {0/0,1/3,2/6}{%
        \node (i\i) at ($(s\i)+(\d mm,-4)$) {$(\Glb_{\i\le i<\omega} s_i)$};%
      }%
      \node[right of=i2] {$\dots$};
      \node (i) at ($(s)+(6 mm,-4)$) {$(t)$};%

    \end{tikzpicture}
  \end{center}
  approximates the term $t = \bot \cons \bot \cons \bot \dots$. Hence,
  $T$ $\prs$-converges to $t$.
\end{example}

Note that in both the metric and the partial order setting continuity
is simply the convergence of every proper prefix: a reduction $S =
(t_\iota \to t_{\iota +1})_{\iota < \alpha}$ is $\mrs$-continuous
(respectively $\prs$-continuous) iff every proper prefix
$\prefix{S}{\beta}$ $\mrs$-converges (respectively $\prs$-converges)
to $t_\beta$.

In order to define $\prs$-convergence, we had to extend terms with
partiality. However, apart from this extension, both $\mrs$- and
$\prs$-convergence coincide. To describe this more precisely we use
the following terms: a reduction $S\fcolon s\wpacont$ is
\emph{$\prs$-continuous in $\iterms$} iff each term in $S$ is total,
i.e.\ in $\iterms$; a reduction $S\fcolon s\wpato t$ is called
\emph{$\prs$-convergent in $\iterms$} iff $t$ and each term in $S$ is
total. We then have the following theorem:
\begin{theorem}[$\prs$-convergence in $\iterms$ = $\mrs$-convergence,
  \cite{bahr09master}]
  \label{thr:strongExt}
  \todo{replace with reference to the LMCS paper once it's accepted}%
  For every reduction $S$ in a TRS the following equivalences hold:
  \\[.5em]
    \begin{inparaenum}[\em(i)]
    \begin{tabular}{@{\quad}r@{\;\,}l@{\;}l@{\qquad iff \qquad}l}
    \item& $S\fcolon s \wpato t$ &in $\iterms$ & $S\fcolon
      s \wmato t$\\%
    \item& $S\fcolon s \wpacont$ &in $\iterms$ & $S\fcolon s
      \wmacont$\rlap{\hbox to 100 pt{\hfill\qEd}}%
  \end{tabular}
    \end{inparaenum}
\end{theorem}\smallskip

\noindent
Example~\ref{ex:termRewr2} illustrates the correspondence between
$\prs$- and $\mrs$-convergence: the reduction $S$ $\prs$-converges in
$\iterms$ and $\mrs$-converges whereas the reduction $T$
$\prs$-converges but not in $\iterms$ and thus does not
$\mrs$-converge.

Kennaway \cite{kennaway92rep} and Bahr \cite{bahr10rta} investigated
abstract models of infinitary rewriting based on metric spaces
respectively partially ordered sets. We shall take these abstract
models as a basis to formulate a theory of infinitary term graph
reductions. The key question that we have to address is what an
appropriate metric space respectively partial order on term graphs
looks like.

\section{Graphs \& Term Graphs}
\label{sec:term-graphs}

This section provides the basic notions for term graphs and more
generally for graphs.  Terms over a signature, say $\Sigma$, can be
thought of as rooted trees whose nodes are labelled with symbols from
$\Sigma$. Moreover, in these trees a node labelled with a $k$-ary
symbol is restricted to have out-degree $k$ and the outgoing edges are
ordered. In this way the $i$-th successor of a node labelled with a
symbol $f$ is interpreted as the root node of the subtree that
represents the $i$-th argument of $f$. For example, consider the term
$f(a,h(a,b))$. The corresponding representation as a tree is shown in
Figure~\ref{fig:exTermTree}.

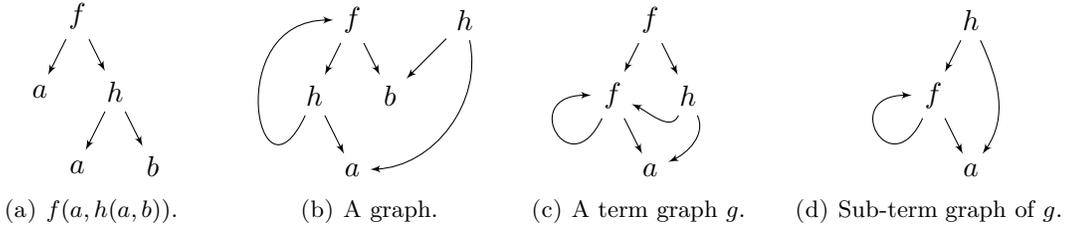
\begin{figure}
  \centering
  \hspace{-1cm}
  \subfloat[$f(a,h(a,b))$.]{
    \label{fig:exTermTree}
  \hspace{1cm}
    \begin{tikzpicture}[->,baseline=(b.base)]
      \node (r1) at (0,0)  {$f$}
      child {
        node  {$a$}
      } child {
        node {$h$}
        child {
          node {$a$}
        } child {
          node (b) {$b$}
        }
      };
    \end{tikzpicture}
  \hspace{1cm}
  }
  \hspace{-10mm}
  \subfloat[A graph.]{
    \label{fig:exGraph}
  \hspace{.5cm}
    \begin{tikzpicture}[->,baseline=(a.base)]
      \node (f) {$f$}
      child {
        node (g) {$h$}
        child [missing]
        child {
          node (a) {$a$}
        }
      } child {
        node (b) {$b$}
      };
      \node [node distance=1cm,right=of f] {$h$}
      edge (b)
      edge [bend left=50] (a);
      \path[use as bounding box] (-1.3,0);
      \draw (g) edge[out=-115,in=180, min distance=1.5cm] (f); 
    \end{tikzpicture}
  \hspace{.5cm}
  }
  \hspace{-10mm}
  \subfloat[A term graph $g$.]{%
    \label{fig:exTermGraph}%
  \hspace{7mm}
    \begin{tikzpicture}[baseline=(n4.base)]
      \node (n1) {$f$}
      child {
        node (n2) {$f$}
        child [missing]
        child {
          node (n4) {$a$}
        }
      }
      child {
        node (n3) {$h$}
      };
      \draw[->] (n3) edge[out=245, in=-25] (n2);
      \draw[->] (n3) edge[out=295, in=25] (n4);
      \path[use as bounding box] (-1.5,0);
      \draw[->] (n2) edge[out=245, in=180, loop] (n2);
    \end{tikzpicture}
  \hspace{10mm}
    }%
  \hspace{-10mm}
  \subfloat[Sub-term graph of $g$.]{%
    \label{fig:exSubTermGraph}%
  \hspace{1.5cm}
    \begin{tikzpicture}[baseline=(n4.base)]
      \node (n1) {$h$}
      child {
        node (n2) {$f$}
        child [missing]
        child {
          node (n4) {$a$}
        }
      }
      child [missing];
      \path[use as bounding box] (-1.5,0);
      \draw[->] (n1) edge[out=295, in=55] (n4);
      \draw[->] (n2) edge[out=245, in=180, loop] (n2);
    \end{tikzpicture}
    \hspace{1.5cm}
    }%
  \hspace{-10mm}
    \caption{Tree representation of a term and generalisation to
      (term) graphs.}
\end{figure}

In term graphs, the restriction to a tree structure is abolished. The
corresponding notion of term graphs we are using is taken from
Barendregt et al.~\cite{barendregt87parle}. We begin by defining the
underlying notion of graphs.

\begin{definition}[graphs]
  \label{def:graph}
  % graph %
  Let $\Sigma$ be a signature. A \emph{graph} over $\Sigma$ is a tuple
  $g = (N,\glab,\gsuc)$ consisting of a set $N$ (of \emph{nodes}), a
  \emph{labelling function} $\glab\fcolon N \funto \Sigma$, and a
  \emph{successor function} $\gsuc\fcolon N \funto N^*$ such that
  $\len{\gsuc(n)} = \srank{\glab(n)}$ for each node $n\in N$, i.e.\ a
  node labelled with a $k$-ary symbol has precisely $k$
  successors. The graph $g$ is called \emph{finite} whenever the
  underlying set $N$ of nodes is finite. If $\gsuc(n) =
  \seq{n_0,\dots,n_{k-1}}$, then we write $\gsuc_{i}(n)$ for
  $n_i$. Moreover, we use the abbreviation $\rank{g}{n}$ for the arity
  $\srank{\glab(n)}$ of $n$.
\end{definition}

\begin{example}
  Let $\Sigma = \set{f/2,h/2,a/0,b/0}$ be a signature. The graph over
  $\Sigma$, depicted in Figure~\ref{fig:exGraph}, is given by the
  triple $(N, \glab, \gsuc)$ with $N = \set{n_0,n_1,n_2,n_3,n_4}$,
  $\glab(n_0) = f, \glab(n_1) = \glab(n_4) = h, \glab(n_2) = b,
  \glab(n_3)=a$ and $\gsuc(n_0) = \seq{n_1,n_2},
  \gsuc(n_1)=\seq{n_0,n_3}, \gsuc(n_2) = \gsuc(n_3) =\seq{},
  \gsuc(n_4) = \seq{n_2,n_3}$.
\end{example}

\begin{definition}[paths, reachability]
  \label{def:graphPath}
  % path %
  Let $g = (N,\glab,\gsuc)$ be a graph and $n,m \in N$.
  \begin{enumerate}[(i)]
  \item A \emph{path} in $g$ from $n$ to $m$ is a finite sequence
    $\pi \in\nats^*$ such that either
    \begin{enumerate}[$-$]
    \item $\pi$ is empty and $n = m$, or
    \item $\pi = \seq{i}\concat\pi'$ with $0 \le i < \rank{g}{n}$ and
      the suffix $\pi'$ is a path in $g$ from $\gsuc_i(n)$ to $m$.
    \end{enumerate}
  \item If there exists a path from $n$ to $m$ in $g$, we say that
    $m$ is \emph{reachable} from $n$ in $g$.
  \end{enumerate}
\end{definition}

\noindent
Since paths are sequences, we may use the prefix order on sequences
for paths as well. That is, we write $\pi_1\le\pi_2$ (respectively
$\pi_1<\pi_2$) if there is a (non-empty) path $\pi_3$ with $\pi_1
\concat \pi_3 = \pi_2$.

\begin{definition}[term graphs]
  \label{def:tgraph}
  % term graph %
  Given a signature $\Sigma$, a \emph{term graph} $g$ over $\Sigma$ is
  a tuple $(N,\glab,\gsuc,r)$ consisting of an \emph{underlying} graph
  $(N,\glab,\gsuc)$ over $\Sigma$ whose nodes are all reachable from
  the \emph{root node} $r\in N$. The term graph $g$ is called
  \emph{finite} if the underlying graph is finite, i.e.\ the set $N$
  of nodes is finite. The class of all term graphs over $\Sigma$ is
  denoted $\itgraphs$; the class of all finite term graphs over
  $\Sigma$ is denoted $\tgraphs$. We use the notation $N^{g}$,
  $\glab^{g}$, $\gsuc^{g}$ and $r^{g}$ to refer to the respective
  components $N$,$\glab$, $\gsuc$ and $r$ of $g$. In analogy to
  subterms, term graphs have \emph{sub-term graphs}. Given a graph or
  a term graph $h$ and a node $n$ in $h$, we write $\subgraph{h}{n}$
  to denote the sub-term graph of $h$ rooted in $n$.
\end{definition}

\begin{example}
  Let $\Sigma = \set{f/2,h/2,c/0}$ be a signature. The term graph $g$
  over $\Sigma$, depicted in Figure~\ref{fig:exTermGraph}, is given by
  the quadruple $(N, \glab, \gsuc,r)$, where $N = \set{r, n_1, n_2,
    n_3}$, $\gsuc(r) = \seq{n_1, n_2}$, $\gsuc(n_1) = \seq{n_1, n_3}$,
  $\gsuc(n_2) = \seq{n_1, n_3}$, $\gsuc(n_3) = \emptyseq$ and
  $\glab(r) = \glab(n_1) = f$, $\glab(n_2) = h$, $\glab(n_3) =
  c$. Figure~\ref{fig:exSubTermGraph} depicts the sub-term graph
  $\subgraph{g}{n_2}$ of $g$.
\end{example}

Paths in a graph are not absolute but relative to a starting node. In
term graphs, however, we have a distinguished root node from which
each node is reachable. Paths relative to the root node are central
for dealing with term graphs:
\begin{definition}[positions, depth, cyclicity, trees]
  \label{def:tgraphOcc}
  % position, depth, tree %
  Let $g \in \itgraphs$ and $n \in N^g$.
  \begin{enumerate}[(i)]
  \item A \emph{position} of $n$ in $g$ is a path in the underlying
    graph of $g$ from $r^g$ to $n$. The set of all positions in $g$ is
    denoted $\pos{g}$; the set of all positions of $n$ in $g$ is
    denoted $\nodePos{g}{n}$.\footnote{The notion/notation of
      positions is borrowed from terms: Every position $\pi$ of a node
      $n$ corresponds to the subterm represented by $n$ occurring at
      position $\pi$ in the unravelling of the term graph to a term.}
  \item The \emph{depth} of $n$ in $g$, denoted $\depth{g}{n}$, is the
    minimum of the lengths of the positions of $n$ in $g$, i.e.\
    $\depth{g}{n} = \min \setcom{\len{\pi}}{\pi \in \nodePos{g}{n}}$.
  \item For a position $\pi \in \pos{g}$, we write
    $\nodeAtPos{g}{\pi}$ for the unique node $n\in N^g$ with $\pi \in
    \nodePos{g}{n}$ and $g(\pi)$ for its symbol $\glab^g(n)$.
  \item A position $\pi\in\pos{g}$ is called \emph{cyclic} if there
    are paths $\pi_1 <\pi_2 \le \pi$ with $\nodeAtPos{g}{\pi_1} =
    \nodeAtPos{g}{\pi_2}$, i.e.\ $\pi$ passes a node twice. The
    non-empty path $\pi'$ with $\pi_1\concat \pi' = \pi_2$ is then
    called a \emph{cycle} of $\nodeAtPos{g}{\pi_1}$. A position that
    is not cyclic is called \emph{acyclic}. If $g$ has a cyclic
    position, $g$ is called cyclic; otherwise $g$ is called acyclic.
  \item The term graph $g$ is called a \emph{term tree} if each node
    in $g$ has exactly one position.
  \end{enumerate}
\end{definition}

\noindent
Note that the labelling function of graphs -- and thus term graphs --
is \emph{total}. In contrast, Barendregt et al.\
\cite{barendregt87parle} considered \emph{open} (term) graphs with a
\emph{partial} labelling function such that unlabelled nodes denote
holes or variables. This is reflected in their notion of homomorphisms
in which the homomorphism condition is suspended for unlabelled nodes.

\subsection{Homomorphisms}
\label{sec:homomorphisms}

Instead of a partial node labelling function for term graphs, we chose
a \emph{syntactic} approach that is closer to the representation in
terms: variables, holes and ``bottoms'' are represented as
distinguished syntactic entities. We achieve this on term graphs by
making the notion of homomorphisms dependent on a set of constant
symbols $\Delta$ for which the homomorphism condition is suspended:

\begin{definition}[$\Delta$-homomorphisms]
  \label{def:D-hom}
  % Delta-homomorphism, Delta-isomorphism %
  Let $\Sigma$ be a signature, $\Delta\subseteq \Sigma^{(0)}$, and
  $g,h \in \itgraphs$.
  \begin{enumerate}[(i)]
  \item A function $\phi\fcolon N^g \funto N^h$ is called
    \emph{homomorphic}\ in $n \in N^g$ if the following holds:
    \begin{align*}
      \glab^g(n) &= \glab^h(\phi(n))
      \tag{labelling}\\
      \phi(\gsuc^g_i(n)) &= \gsuc^h_i(\phi(n)) \quad \text{ for all } 0 \le i <
      \rank{g}{n} \tag{successor}
    \end{align*}
    % For a subset $N' \subseteq N$, we also say that $\phi$ is
    % homomorphic in $N'$ if $\phi$ is homomorphic in $n$ for all $n \in
    % N'$.
  \item A \emph{$\Delta$-homomorphism} $\phi$ from $g$ to $h$, denoted
    $\phi\fcolon g \homto_\Delta h$, is a function $\phi\fcolon N^g
    \funto N^h$ that is homomorphic in $n$ for all $n \in N^g$ with
    $\glab^g(n) \nin \Delta$ and satisfies 
    \begin{gather}
      \phi(r^g) = r^h \tag{root}
    \end{gather}
  \end{enumerate}
\end{definition}

\noindent
Note that, for $\Delta = \emptyset$, we get the usual notion of
homomorphisms on term graphs (e.g.\
Barendsen~\cite{barendsen03book}). The $\Delta$-nodes can be thought
of as holes in the term graphs that can be filled with other term
graphs. For example, if we have a distinguished set of variable
symbols $\calV \subseteq \Sigma^{(0)}$, we can use
$\calV$-homomorphisms to formalise the matching step of term graph
rewriting, which requires the instantiation of variables.

\begin{figure}
  \centering
  \subfloat[A homomorphism.]{
  \begin{tikzpicture}
      \node (g) {$f$}%
      child {%
        node (g1) {$h$}%
        child {%
          node (g11) {$a$}%
        }%
      } child {%
        node (g2) {$a$}%
      };%
      \node[node distance=3cm, right=of g] (h) {$f$}%
      child {%
        node (h1) {$h$}%
        child {%
          node (h11) {$a$}%
        }%
      } child [missing];%
      \draw%
      (h) edge[->,bend left=20] (h11);%
      \begin{scope}[|-to,black!50,dashed, thin]%|
        \path%
        (g) edge (h)%
        (g1) edge[bend left=20] (h1)%
        (g2) edge (h11)%
        (g11) edge (h11)%
        ;
        \path(g) -- (h) node[midway,above,black!70] {$\phi$};
      \end{scope}

      \begin{scope}[black!90, node distance=2.2cm]
        \node[below=of g] (g) {$g_1$};
        \node[node distance=1mm,left=of g] {$\phi\fcolon$};
        \node[below=of h] (h) {$g_2$};
        \draw[fun,shorten=5mm] (g) -- (h) node[pos=.85,below] {\phantom{\tiny$\set{a,b}$}};
      \end{scope}
    \end{tikzpicture}
    \label{fig:deltaHom1}
  }
  \hspace{1.5cm}
  \subfloat[A $\set{a,b}$-homomorphism.]{
  \begin{tikzpicture}
      \node (g) {$f$}%
      child {%
        node (g1) {$a$}%
      } child {%
        node (g2) {$b$}%
      };%
      \node[node distance=3cm, right=of g] (h) {$f$}%
      child {%
        node (h1) {$h$}%
        child {%
          node {$a$}%
        }%
      } child [missing];%
      \draw (h) edge[->,out=-55,in=0,loop] (h);%
      \begin{scope}[|-to,black!50,dashed, thin]%|
        \path%
        (g) edge (h)%
        (g1) edge[bend right=20] (h1)%
        (g2) edge[bend right=20] (h)%
        ;
        \path(g) -- (h) node[midway,above,black!70] {$\psi$};
      \end{scope}

      \begin{scope}[black!90, node distance=2.2cm]
        \node[below=of g] (g) {$g_3$};
        \node[node distance=1mm,left=of g] {$\psi\fcolon$};
        \node[below=of h] (h) {$g_4$};
        \draw[fun,shorten=5mm] (g) -- (h) node[pos=.85,below] {\tiny$\set{a,b}$};
      \end{scope}
    \end{tikzpicture}
    \label{fig:deltaHom2}
    }
  \caption{$\Delta$-homomorphisms.}
  \label{fig:deltaHom}
\end{figure}
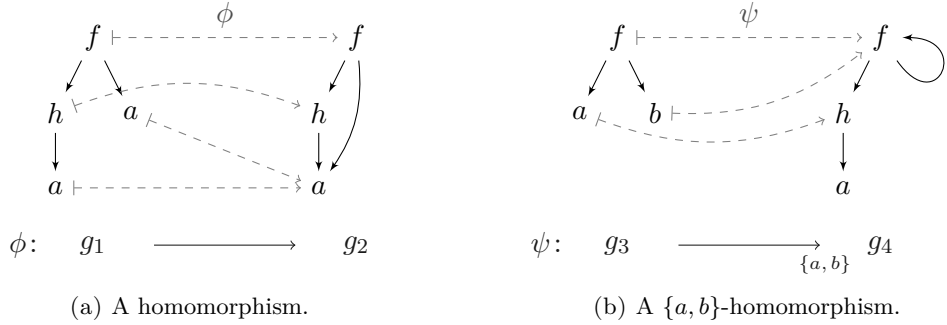

\begin{example}
  \label{ex:deltaHom}
  Figure~\ref{fig:deltaHom} depicts two functions $\phi$ and
  $\psi$. Whereas $\phi$ is a homomorphism, the function $\psi$ is not
  a homomorphism since, for example, the node labelled $a$ in $g_3$ is
  mapped to a node labelled $h$ in $g_3$. Nevertheless, $\psi$ is a
  $\set{a,b}$-homomorphism. Note that $\Delta$-homomorphisms may
  introduce additional sharing in the target term graph by mapping
  several nodes in the source to the same node in the target.
\end{example}

\begin{proposition}[$\Delta$-homomorphism preorder]
  \label{prop:catgraph}
  The $\Delta$-homomorphisms on $\itgraphs$ form a category that is a
  preorder, i.e.\ there is at most one $\Delta$-homomorphism from one
  term graph to another.
\end{proposition}
\begin{proof}
  The identity $\Delta$-homomorphism is obviously the identity mapping
  on the set of nodes. Moreover, an easy equational reasoning reveals
  that the composition of two $\Delta$-homomorphisms is again a
  $\Delta$-homomorphism. Associativity of this composition is obvious
  as $\Delta$-homomorphisms are functions.

  To show that the category is a preorder, assume that there are two
  $\Delta$-homomorphisms $\phi_1,\phi_2\fcolon g \homto_\Delta h$. We
  prove that $\phi_1 = \phi_2$ by showing that $\phi_1(n) = \phi_2(n)$
  for all $n \in N^g$ by induction on the depth of $n$ in $g$.

  Let $\depth{g}{n} = 0$, i.e.\ $n = r^g$. By the root condition for
  $\phi$, we have that $\phi_1(r^g) = r^h = \phi_2(r^g)$. Let
  $\depth{g}{n} = d > 0$. Then $n$ has a position $\pi \concat \seq i$
  in $g$ such that $\depth{g}{n'}<d$ for $n' =
  \nodeAtPos{g}{\pi}$. Hence, we can employ the induction hypothesis
  for $n'$. Moreover, since $n'$ has at least one successor node,
  viz.\ $n$, it cannot be labelled with a nullary symbol and a
  fortiori not with a symbol in $\Delta$. Therefore, the
  $\Delta$-homomorphisms $\phi_1$ and $\phi_2$ are homomorphic in $n'$
  and we can thus reason as follows:
  \begin{align*}
    \phi_1(n) &= \gsuc^h_i(\phi_1(n'))
    \tag{successor condition for $\phi_1$}\\
    &= \gsuc^h_i(\phi_2(n'))
    \tag{ind. hyp.} \\
    &= \phi_2(n)
    \tag{successor condition for $\phi_2$}
  \end{align*}
\end{proof}
As a consequence, whenever there are two $\Delta$-homomorphisms
$\phi\fcolon g \homto_\Delta h$ and $\psi\fcolon h \homto_\Delta g$,
they are inverses of each other, i.e.\
\emph{$\Delta$-isomorphisms}. If two term graphs are
\emph{$\Delta$-isomorphic}, we write $g \isom_\Delta h$.

For the two special cases $\Delta = \emptyset$ and $\Delta =
\set{\sigma}$, we write $\phi\fcolon g \homto h$ respectively $\phi\fcolon g
\homto_\sigma h$ instead of $\phi\fcolon g \homto_\Delta h$ and call
$\phi$ a \emph{homomorphism} respectively a \emph{$\sigma$-homomorphism}. The same
convention applies to $\Delta$-isomorphisms.

The structure of positions permits a convenient characterisation of
$\Delta$-homomorphisms:
\begin{lemma}[characterisation of $\Delta$-homomorphisms]
  \label{lem:canhom}
  % characterisation of Delta-homomorphisms %
  For $g,h\in \itgraphs$, a function $\phi\fcolon N^g \funto N^h$ is a
  $\Delta$-homomorphism $\phi\fcolon g \homto_\Delta h$ iff the
  following holds for all $n \in N^g$:
  \begin{center}
    \begin{inparaenum}[(a)]
    \item $\nodePos{g}{n} \subseteq \nodePos{h}{\phi(n)}$,\quad and
      \label{item:canhom1}
      \qquad
    \item $\glab^g(n)\nin\Delta \quad\implies\quad \glab^g(n) =
      \glab^h(\phi(n))$.
      \label{item:canhom2}
    \end{inparaenum}
  \end{center}
\end{lemma}
\begin{proof}
  \def\itema{(\ref{item:canhom1})}%
  \def\itemb{(\ref{item:canhom2})}%
  For the ``only if'' direction, assume that $\phi\fcolon g
  \homto_\Delta h$. \itemb{} is the labelling condition and is
  therefore satisfied by $\phi$. To establish \itema{}, we show the
  equivalent statement
  \[
  \forall \pi \in \pos{g}.\; \forall n \in N^g. \; \pi \in \nodePos{g}{n} \implies
  \pi \in \nodePos{h}{\phi(n)}
  \]
  We do so by induction on the length of $\pi$: if $\pi = \emptyseq$,
  then $\pi \in \nodePos{g}{n}$ implies $n = r^g$. By the root
  condition, we have $\phi(r^g)=r^h$ and, therefore, $\pi = \emptyseq
  \in \phi(r^g)$. If $\pi = \pi' \concat \seq i$, then let $n' =
  \nodeAtPos{g}{\pi'}$. Consequently, $\pi' \in \nodePos{g}{n'}$ and,
  by induction hypothesis, $\pi' \in \nodePos{h}{\phi(n')}$. Since
  $\pi = \pi'\concat \seq i$, we have $\gsuc^g_i(n') = n$. By the
  successor condition we can conclude $\phi(n) =
  \gsuc^h_i(\phi(n'))$. This and $\pi' \in \nodePos{h}{\phi(n')}$
  yields that $\pi'\concat \seq i \in \nodePos{h}{\phi(n)}$.

  For the ``if'' direction, we assume \itema{} and \itemb{}. The
  labelling condition follows immediately from \itemb{}. For the root
  condition, observe that since $\emptyseq \in \nodePos{g}{r^g}$, we
  also have $\emptyseq \in \nodePos{h}{\phi(r^g)}$. Hence, $\phi(r^g)
  = r^h$. In order to show the successor condition, let $n, n' \in
  N^g$ and $0 \le i < \rank{g}{n}$ such that $\gsuc^g_i(n) = n'$. Then
  there is a position $\pi \in \nodePos{g}{n}$ with $\pi \concat \seq
  i \in \nodePos{g}{n'}$. By \itema{}, we can conclude that $\pi \in
  \nodePos{h}{\phi(n)}$ and $\pi \concat \seq i \in
  \nodePos{h}{\phi(n')}$ which implies that $\gsuc^h_i(\phi(n)) =
  \phi(n')$.
\end{proof}

By Proposition~\ref{prop:catgraph}, there is at most one
$\Delta$-homomorphism between two term graphs. The lemma above
uniquely defines this $\Delta$-homomorphism: if there is a
$\Delta$-homomorphism from $g$ to $h$, it is defined by $\phi(n) =
n'$, where $n'$ is the unique node $n' \in N^h$ with $\nodePos{g}{n}
\subseteq \nodePos{h}{n'}$. Moreover, while it is not true for
arbitrary $\Delta$-homomorphisms, we have that homomorphisms are
surjective.

\begin{lemma}[homomorphisms are surjective]
  \label{lem:homSurj}
  % homomorphisms are surjective %
  Every homomorphism $\phi\fcolon g \to h$, with $g, h \in \itgraphs$,
  is surjective.
\end{lemma}
\begin{proof}
  Follows from an easy induction on the depth of the nodes in $h$.
\end{proof}

The $\set{a,b}$-homomorphism illustrated in
Figure~\ref{fig:deltaHom2}, shows that the above lemma does not hold
for $\Delta$-homomorphisms in general.

\subsection{Isomorphisms \& Isomorphism Classes}
\label{sec:isom--canon}

When dealing with term graphs, in particular, when studying term graph
transformations, we do not want to distinguish between isomorphic term
graphs. Distinct but isomorphic term graphs do only differ in the
naming of nodes and are thus an unwanted artifact of the definition of
term graphs. In this way, equality up to isomorphism is similar to
$\alpha$-equivalence of $\lambda$-terms and has to be dealt with.

In this section, we characterise isomorphisms and more generally
$\Delta$-isomorphisms. From this we derive two canonical
representations of isomorphism classes of term graphs. One is simply a
subclass of the class of term graphs while the other one is based on
the structure provided by the positions of term graphs. The relevance
of the former representation is derived from the fact that we still
have term graphs that can be easily manipulated whereas the latter is
more technical and will be helpful for constructing term graphs up to
isomorphism.

Note that a bijective $\Delta$-homomorphism is not necessarily a
$\Delta$-isomorphism. To realise this, consider two term graphs $g,h$,
each with one node only. Let the node in $g$ be labelled with $a$ and
the node in $h$ with $b$ then the only possible $a$-homomorphism from
$g$ to $h$ is clearly a bijection but not an $a$-isomorphism. On the
other hand, bijective homomorphisms indeed are isomorphisms.
\begin{lemma}[bijective homomorphisms are isomorphisms]
  \label{lem:isomBij}
  % bijective homomorphisms are isomorphisms %
  Let $g,h \in \itgraphs$ and $\phi\fcolon g \homto h$. Then the following are
  equivalent
  \begin{enumerate}[(a)]
  \item $\phi$ is an isomorphism. \label{item:isomBij-a}
  \item $\phi$ is bijective. \label{item:isomBij-b}
  \item $\phi$ is injective. \label{item:isomBij-c}
  \end{enumerate}
\end{lemma}
\begin{proof}
  The implication (\ref{item:isomBij-a}) $\Rightarrow$
  (\ref{item:isomBij-b}) is trivial. The equivalence
  (\ref{item:isomBij-b}) $\Leftrightarrow$ (\ref{item:isomBij-c})
  follows from Lemma~\ref{lem:homSurj}. For the implication
  (\ref{item:isomBij-b}) $\Rightarrow$ (\ref{item:isomBij-a}),
  consider the inverse $\phi^{-1}$ of $\phi$. We need to show that
  $\phi^{-1}$ is a homomorphism from $h$ to $g$. The root condition
  follows immediately from the root condition for $\phi$. Similarly,
  an easy equational reasoning reveals that $\phi^{-1}$ is homomorphic
  in $N^h$ since $\phi$ is homomorphic in all $n\in N^g$.
\end{proof}

From the characterisation of $\Delta$-homomorphisms in
Lemma~\ref{lem:canhom}, we immediately obtain a characterisation of
$\Delta$-isomorphisms as follows:
\begin{lemma}[characterisation of $\Delta$-isomorphisms]
  \label{lem:isomCan}
  % characterisation of Delta-isomorphisms %
  For all $g,h \in \itgraphs$, a function $\phi\fcolon N^g \funto N^h$
  is a $\Delta$-isomorphism iff for all $n \in N^g$
    \begin{enumerate}[\em(a)]
    \item $\nodePos{h}{\phi(n)} = \nodePos{g}{n}$, and
    \item $\glab^g(n) = \glab^h(\phi(n))$ or
      $\glab^g(n),\glab^h(\phi(n))\in\Delta$.
    \end{enumerate}
\end{lemma}
\begin{proof}
  Immediate consequence of Lemma~\ref{lem:canhom} and
  Proposition~\ref{prop:catgraph}.
\end{proof}

Note that whenever $\Delta$ is a singleton set, the condition
$\glab^g(n),\glab^h(\phi(n)) \in \Delta$ in the above lemma implies
$\glab^g(n) = \glab^h(\phi(n))$. Therefore, we obtain the following
corollary:
\begin{corollary}[$\sigma$-isomorphism = isomorphism]
  \label{cor:sig-isom-isom}
  Given $g,h \in \itgraphs$ and $\sigma \in \Sigma^{(0)}$, we have $g
  \isom h$ iff $g \isom_\sigma h$.
\end{corollary}

Note that the above equivalence does not hold for
$\Delta$-homomorphisms with more than one symbol in $\Delta$: consider
the term graphs $g = a$ and $h = b$ consisting of a single node
labelled $a$ respectively $b$. While $g$ and $h$ are
$\Delta$-isomorphic for $\Delta = \set{a,b}$, they are not isomorphic.

\subsubsection{Canonical Term Graphs}
\label{sec:canon-term-graphs}

From the Lemmas~\ref{lem:isomBij} and \ref{lem:isomCan} we learned
that isomorphisms between term graphs are bijections that preserve and
reflect the positions as well as the labelling of each node. These
findings motivate the following definition of canonical term graphs as
candidates for representatives of isomorphism classes:
\begin{definition}[canonical term graphs]
  \label{def:canTgraph}
  % canonical term graphs %
  A term graph $g$ is called \emph{canonical} if $n = \nodePos{g}{n}$
  holds for each $n \in N^g$. That is, each node is the set of its
  positions in the term graph. The set of all (finite) canonical term
  graphs over $\Sigma$ is denoted $\ictgraphs$ (respectively $\ctgraphs$).
  Given a term graph $h \in \ictgraphs$, its \emph{canonical
    representative} $\canon{h}$ is the canonical term graph given by
  \begin{align*}
    &N^{\canon{h}} = \setcom{\nodePos{h}{n}}{n \in N} \qquad r^{\canon{h}}
    = \nodePos{h}{r}  \qquad \glab^{\canon{h}}(\nodePos{h}{n}) =
    \glab^h(n) \quad \text{for all } n \in N \\
    &\gsuc^{\canon{h}}_i(\nodePos{h}{n}) = \nodePos{h}{\gsuc^{h}_i(n)}
    \quad \text{for all } n \in N, 0 \le i < \rank{h}{n}
\end{align*}
\end{definition}
The above definition follows a well-known approach to obtain, for each
term graph $g$, a canonical representative $\canon g$
\cite{plump99hggcbgt}. One can easily see that $\canon{g}$ is a
well-defined canonical term graph. With this definition we indeed
capture a notion of canonical representatives of isomorphism classes:
\begin{proposition}[canonical term graphs are isomorphism class
  representatives]
  \label{prop:canon}
  Given $g \in \itgraphs$, the term graph $\canon{g}$ canonically
  represents the equivalence class $\eqc{g}{\isom}$. More precisely,
  it holds that
  \begin{center}
    \begin{inparaenum}[\em(i)]
    \item $\eqc{g}{\isom} = \eqc{\canon{g}}{\isom}$\,, \quad and\qquad
    \item $\eqc{g}{\isom} = \eqc{h}{\isom}$ \quad iff \quad $\canon{g}
      = \canon{h}$.
    \end{inparaenum}
  \end{center}
  In particular, we have, for all canonical term graphs $g,h$, that $g =
  h$ iff $g \isom h$.
\end{proposition}
\begin{proof}
  Straightforward consequence of Lemma~\ref{lem:isomCan}.
\end{proof}

\subsubsection{Labelled Quotient Trees}
\label{sec:labell-quot-tree}

Intuitively, term graphs can be thought of as ``terms with sharing'',
i.e.\ terms in which occurrences of the same subterm may be
identified. The representation of isomorphic term graphs as
\emph{labelled quotient trees}, which we shall study in this section,
makes use of and formalises this intuition. To this end, we introduce
an equivalence relation on the positions of a term graph that captures
the sharing in a term graph:
\begin{definition}[aliasing positions]
  Given a term graph $g$ and two positions $\pi_1,\pi_2 \in \pos{g}$,
  we say that $\pi_1$ and $\pi_2$ \emph{alias each other} in $g$,
  denoted $\pi_1 \sim_g \pi_2$, if $\nodeAtPos{g}{\pi_1} =
  \nodeAtPos{g}{\pi_2}$.
\end{definition}
One can easily see that the thus defined relation $\sim_g$ on
$\pos{g}$ is an equivalence relation. Moreover, the partition on
$\pos{g}$ induced by $\sim_g$ is simply the set
$\setcom{\nodePos{g}{n}}{n\in N^g}$ that contains the sets of
positions of nodes in $g$.

\begin{example}
  For the term graph $g_2$ illustrated in Figure~\ref{fig:deltaHom1},
  we have that $\seq{0,0} \sim_{g_2} \seq{1}$ as both $\seq{0,0}$ and
  $\seq{1}$ are positions of the $a$-node in $g_2$. For the term graph
  $g_4$ in Figure~\ref{fig:deltaHom2}, $\emptyseq \sim_{g_4} \seq{1}
  \sim_{g_4} \seq {1,1} \sim_{g_4} \dots$ as all finite sequences over
  $1$ are positions of the $f$-node in $g_4$.
\end{example}

The characterisation of $\Delta$-homomorphisms of
Lemma~\ref{lem:canhom} can be recast in terms of aliasing positions,
which then yields the following characterisation of the
\emph{existence} of $\Delta$-homomorphisms:
\begin{lemma}[characterisation of $\Delta$-homomorphisms]
  \label{lem:occrephom}
  % characterisation of Delta-homomorphisms %
  Given $g,h\in \itgraphs$, there is a $\Delta$-homomorphism
  $\phi\fcolon g \homto_\Delta h$ iff, for all $\pi,\pi'\in\pos{g}$,
  we have
  \begin{center}
    \begin{inparaenum}[\em(a)]
    \item $\pi \sim_g \pi' \quad\implies\quad \pi \sim_h \pi'$\,,
      \quad and\qquad
      \label{item:occrephom1}
    \item $g(\pi) \nin \Delta \quad\implies\quad g(\pi) = h(\pi)$.
      \label{item:occrephom2}
    \end{inparaenum}
  \end{center}

\end{lemma}
\begin{proof}
  \def\itema{(\ref{item:occrephom1})}%
  \def\itemb{(\ref{item:occrephom2})}%
  \def\itemca{(\ref{item:canhom1})}%
  \def\itemcb{(\ref{item:canhom2})}%
  \def\itemap{(\ref{item:canhom1}')}%
  \def\itembp{(\ref{item:canhom2}')}%
  For the ``only if'' direction, assume that $\phi$ is a
  $\Delta$-homomorphism from $g$ to $h$. Then we can use the
  properties \itemca{} and \itemcb{} of Lemma~\ref{lem:canhom}, which
  we will refer to as \itemap{} and \itembp{} to avoid confusion. In
  order to show \itema{}, assume $\pi \sim_g \pi'$. Then there is some
  node $n \in N^g$ with $\pi,\pi' \in \nodePos{g}{n}$. \itemap{}
  yields $\pi,\pi' \in \phi(n)$ and, therefore, $\pi \sim_h \pi'$. To
  show \itemb{}, we assume some $\pi \in \pos{g}$ with $g(\pi) \nin
  \Delta$. Then we can reason as follows:
  \[
  g(\pi) = \glab^g(\nodeAtPos{g}{\pi}) \stackrel{\text{\itembp}}{=}
  \glab^h(\phi(\nodeAtPos{g}{\pi}))  \stackrel{\text{\itemap}}{=}
  \glab^h(\nodeAtPos{h}{\pi}) = h(\pi)
  \]
  For the converse direction, assume that both \itema{} and \itemb{}
  hold. Define the function $\phi\fcolon N^g \funto N^h$ by $\phi(n) =
  m$ iff $\nodePos{g}{n} \subseteq \nodePos{h}{m}$ for all $n \in N^g$
  and $m \in N^h$. To see that this is well-defined, we show at first
  that, for each $n \in N^g$, there is at most one $m \in N^h$ with
  $\nodePos{g}{n} \subseteq \nodePos{h}{m}$. Suppose there is another
  node $m' \in N^h$ with $\nodePos{g}{n} \subseteq
  \nodePos{h}{m'}$. Since $\nodePos{g}{n} \neq \emptyset$, this
  implies $\nodePos{h}{m} \cap \nodePos{h}{m'} \neq \emptyset$. Hence,
  $m = m'$. Secondly, we show that there is at least one such node
  $m$. Choose some $\pi^* \in \nodePos{g}{n}$. Since then $\pi^*
  \sim_g \pi^*$ and, by \itema{}, also $\pi^* \sim_h \pi^*$ holds,
  there is some $m \in N^h$ with $\pi^* \in \nodePos{h}{m}$. For each
  $\pi \in \nodePos{g}{n}$, we have $\pi^* \sim_g \pi$ and, therefore,
  $\pi^* \sim_h \pi$ by \itema{}. Hence, $\pi \in \nodePos{h}{m}$. So
  we know that $\phi$ is well-defined. By construction, $\phi$
  satisfies \itemap{}. Moreover, because of \itemb{}, it is also
  easily seen to satisfy \itembp{}. Hence, $\phi$ is a homomorphism
  from $g$ to $h$.
\end{proof}
Intuitively, Clause (\ref{item:occrephom1}) states that $h$ has at
least as much sharing of nodes as $g$ has, whereas Clause
(\ref{item:occrephom2}) states that $h$ has at least the same
non-$\Delta$-labelling as $g$. In this sense, the above
characterisation confirms the intuition about $\Delta$-homomorphisms
that we mentioned in Example~\ref{ex:deltaHom}, viz.\
$\Delta$-homomorphisms may only introduce sharing and relabel
$\Delta$-nodes. This can be observed in the two $\Delta$-homomorphisms
illustrated in Figure~\ref{fig:deltaHom}.

From the above characterisations of the existence of
$\Delta$-homomorphisms, we can easily derive the following
characterisation of $\Delta$-isomorphisms using the uniqueness of
$\Delta$-homomorphisms between two term graphs:
\begin{lemma}[characterisation of $\Delta$-isomorphisms]
  \label{lem:isomOcc}
  % characterisation of Delta-isomorphisms %
  For all $g,h \in \itgraphs$, $g \isom_\Delta h$ iff
  \begin{center}
    \begin{inparaenum}[\em(a)]
    \item $\sim_g\ =\ \sim_h$\,, \quad and\qquad
    \item $g(\pi) = h(\pi)$ or $g(\pi),h(\pi) \in \Delta$ \quad for
      all $\pi \in \pos{g}$.
    \end{inparaenum}
  \end{center}
\end{lemma}
\begin{proof}
  Immediate consequence of Lemma~\ref{lem:occrephom} and
  Proposition~\ref{prop:catgraph}.
\end{proof}

\begin{rem}
  \label{rem:catTgraphs}
  $\Delta$-homomorphisms can be naturally lifted to the set of
  isomorphism classes $\quotient{\itgraphs}{\isom}$: we say that two
  $\Delta$-homomorphisms $\phi\fcolon g \homto_\Delta h$,
  $\phi'\fcolon g' \homto_\Delta h'$, are isomorphic, written $\phi
  \isom \phi'$ iff there are isomorphisms $\psi_1\fcolon g \isoto g'$
  and $\psi_2\fcolon h \isoto h'$ such that $\psi_2 \circ \phi = \phi'
  \circ \psi_1$. Given a $\Delta$-homomorphism $\phi\fcolon g
  \homto_\Delta h$ in $\itgraphs$, $\eqc{\phi}\isom\fcolon
  \eqc{g}\isom \homto_\Delta \eqc{h}\isom$ is a $\Delta$-homomorphism
  in $\quotient{\itgraphs}{\isom}$. These $\Delta$-homomorphisms then
  form a category which can easily be show to be isomorphic to the
  category of $\Delta$-homomorphisms on $\ictgraphs$ via the mapping
  $\eqc{\cdot}\isom$.
\end{rem}

Lemma~\ref{lem:isomOcc} has shown that term graphs can be
characterised up to isomorphism by only giving the equivalence
$\sim_g$ and the labelling $g(\cdot)\fcolon \pi \mapsto g(\pi)$ of the
involved term graphs. This observation gives rise to the following
definition:

\begin{definition}[labelled quotient trees]
  \label{def:occRep}
  A \emph{labelled quotient tree} over signature $\Sigma$ is a triple
  $(P,l,\sim)$ consisting of a non-empty set $P \subseteq \nats^*$, a
  function $l\fcolon P \funto \Sigma$, and an equivalence relation
  $\sim$ on $P$ that satisfies the following conditions for all
  $\pi,\pi' \in \nats^*$ and $i \in \nats$:
  \begin{align*}
    \pi\concat \seq i \in P \quad &\implies \quad \pi \in P \quad \text{
      and } \quad i < \srank{l(\pi)}
    \tag{reachability} \\
    \pi \sim \pi' \quad &\implies \quad
    \begin{cases}
      l(\pi) = l(\pi') &\text{ and }\\
      \pi\concat \seq i \sim \pi' \concat \seq i &\text{ for all }\; i <
      \srank{l(\pi)}
    \end{cases}
    \tag{congruence}
  \end{align*}
\end{definition}
\noindent In other words, a labelled quotient tree $(P,l,\sim)$ is a
ranked tree domain $P$ together with a congruence $\sim$ on it and a
labelling function $l\fcolon \quotient{P}{\sim} \funto \Sigma$ that
honours the rank. Also note that since $P$ must be non-empty, the
reachability condition implies that $\emptyseq \in P$.

\begin{example}
  The term graph $g_2$ depicted in Figure~\ref{fig:deltaHom1} is
  represented up to isomorphism by the labelled quotient tree
  $(P,l,\sim)$ with $P = \set{\emptyseq,\seq 0,\seq {0,0},\seq 1}$,
  $l(\emptyseq) = f$, $l(\seq 0) = h$, $l(\seq {0,0}) = l(\seq 1) = a$
  and $\sim$ the least equivalence relation on $P$ with $\seq{0,0}
  \sim \seq 1$.
\end{example}

The following lemma confirms that labelled quotient trees uniquely
characterise any term graph up to isomorphism:
\begin{lemma}[labelled quotient trees are canonical]
  \label{lem:occrep}
  Each term graph $g \in \itgraphs$ induces a \emph{canonical labelled
    quotient tree} $(\pos{g},g(\cdot),\sim_g)$ over $\Sigma$. Vice
  versa, for each labelled quotient tree $(P,l,\sim)$ over $\Sigma$
  there is a unique canonical term graph $g\in \ictgraphs$ whose
  canonical labelled quotient tree is $(P,l,\sim)$, i.e.\ $\pos{g} =
  P$, $g(\pi) = l(\pi)$ for all $\pi \in P$, and $\sim_g\ =\ \sim$.
\end{lemma}
\begin{proof}
  The first part is trivial: $(\pos{g},g(\cdot),\sim_g)$ satisfies the
  conditions from Definition~\ref{def:occRep}.
  
  For the second part, let $(P,l, \sim)$ be a labelled quotient
  tree. Define the term graph $g = (N,\glab,\gsuc,r)$ by
  \begin{align*}
    N &= \quotient{P}{\sim} &
    \glab(n) = f \quad &\text{ iff } \quad \exists \pi \in n.\;l(\pi)
    = f\\
    r &= \eqc{\emptyseq}{\sim} &
    \gsuc_i(n) = n' \quad &\text{ iff } \quad \exists \pi \in
    n.\; \pi\concat \seq i \in n'
  \end{align*}
  The functions $\glab$ and $\gsuc$ are well-defined due to the
  congruence condition satisfied by $(P,l, \sim)$. Since $P$ is
  non-empty and closed under prefixes, it contains $\emptyseq$. Hence,
  $r$ is well-defined. Moreover, by the reachability condition, each
  node in $N$ is reachable from the root node. An easy induction proof
  shows that $\nodePos{g}{n} = n$ for each node $n \in N$. Thus, $g$
  is a well-defined canonical term graph. The canonical labelled
  quotient tree of $g$ is obviously $(P,l, \sim)$. Whenever there are
  two canonical term graphs with the same canonical labelled quotient
  tree $(P,l, \sim)$, they are isomorphic due to
  Lemma~\ref{lem:isomOcc} and, therefore, have to be identical by
  Proposition~\ref{prop:canon}.
\end{proof}

Labelled quotient trees provide a valuable tool for constructing
canonical term graphs as we shall see. Nevertheless, the original
graph representation remains convenient for practical purposes as it
allows a straightforward formalisation of term graph rewriting and
provides a finite representation of finite cyclic term graphs, which
induce an infinite labelled quotient tree.

\subsubsection{Terms, Term Trees \& Unravelling}
\label{sec:terms-term-trees}

Before we continue, it is instructive to make the correspondence
between terms and term graphs clear. First, note that, for each term
tree $t$, the equivalence $\sim_t$ is the identity relation
$\idrel{\pos{t}}$ on $\pos{t}$, i.e.\ $\pi_1 \sim_t \pi_2$ iff $\pi_1
= \pi_2$. Consequently, we have the following one-to-one
correspondence between canonical term \emph{trees} and terms: each
term $t \in \iterms$ induces the canonical term tree given by the
labelled quotient tree $(\pos{t},t(\cdot),\idrel{\pos{t}})$. For
example, the term tree depicted in Figure~\ref{fig:exTermTree}
corresponds to the term $f(a,h(a,b))$. We thus consider the set of
terms $\iterms$ as the subset of canonical term trees of $\ictgraphs$.

With this correspondence in mind, we can define the \emph{unravelling}
of a term graph $g$ as the unique term $t$ such that there is a
homomorphism $\phi\fcolon t \homto g$. The unravelling of cyclic term
graphs yields infinite terms, e.g.\ in Figure~\ref{fig:gredEx} on
page~\pageref{fig:gredEx}, the term $h_\omega$ is the unravelling of
the term graph $g_2$. We use the notation $\unrav{g}$ for the
unravelling of $g$.

\section{A Rigid Partial Order on Term Graphs}
\label{sec:partial-order-lebot1}

In this section, we shall establish a partial order suitable for
formalising convergence of sequences of canonical term graphs
similarly to $\prs$-convergence on terms.

Recall that $\prs$-convergence in term rewriting systems is based on a
partial order $\lebot$ on the set $\ipterms$ of \emph{partial}
terms. The partial order $\lebot$ instantiates occurrences of $\bot$
from left to right, i.e.\ $s \lebot t$ iff $t$ is obtained by
replacing occurrences of $\bot$ in $s$ by arbitrary terms in
$\ipterms$.

Since we are considering term graph rewriting as a generalisation of
term rewriting, our aim is to generalise the partial order $\lebot$ on
terms to term graphs. That is, the partial order we are looking for
should coincide with $\lebot$ if restricted to term trees. Moreover,
we also want to maintain the characteristic properties of the partial
order $\lebot$ when generalising to term graphs. The most important
characteristic we are striving for is a complete semilattice structure
in order to define $\prs$-convergence in terms of the limit
inferior. Apart from that, we also want to maintain the intuition of
the partial order $\lebot$, viz.\ the intuition of information
preservation, which $\lebot$ captures on terms as we illustrated in
Section~\ref{sec:preliminaries}. We will make this last guiding
principle clearer as we go along.

Analogously to partial terms, we consider the class of \emph{partial
  term graphs} simply as term graphs over the signature $\Sigma_\bot =
\Sigma \uplus \set{\bot}$. In order to generalise the partial order
$\lebot$ to term graphs, we need to formalise the instantiation of
occurrences of $\bot$ in term graphs. $\Delta$-homomorphisms, for
$\Delta = \set{\bot}$ -- or $\bot$-homomorphisms for short -- provide
the right starting point for that. A homomorphism $\phi\colon g \homto
h$ maps each node in $g$ to a node in $h$ while preserving the
\emph{local} structure of each node, viz.\ its labelling and its
successors. In the case of a $\bot$-homomorphisms $\phi\colon g
\homto_\bot h$, the preservation of the labelling is suspended for
nodes labelled $\bot$ thus allowing $\phi$ to instantiate each
$\bot$-node in $g$ with an arbitrary node in $h$.

Therefore, we shall use $\bot$-homomorphisms as the basis for
generalising $\lebot$ to canonical partial term graphs. This approach
is based on the observation that $\bot$-homomorphisms characterise the
partial order $\lebot$ on terms. Considering terms as canonical term
trees, we obtain the following equivalence:
\[
s \lebot t \iff \text{ there is a $\bot$-homomorphism } \phi\fcolon s
\homto_\bot t.
\]
Thus, $\bot$-homomorphisms constitute the ideal tool for defining a
partial order on canonical partial term graphs that generalises
$\lebot$. In the following subsection, we shall explore different
partial orders on canonical partial term graphs based on
$\bot$-homomorphisms.

\subsection{Partial Orders on Term Graphs}
\label{sec:partial-orders-term}

Consider the \emph{simple partial order} $\lebots$ defined on term
graphs as follows: $g \lebots h$ iff there is a $\bot$-homomorphism
$\phi\fcolon g \homto_\bot h$. This is a straightforward
generalisation of the partial order $\lebot$ to term graphs. In fact,
this partial order forms a complete semilattice on $\ipctgraphs$
\cite{bahr12rta}.

As we have explained in Section~\ref{sec:preliminaries},
$\prs$-convergence on terms is based on the ability of the partial
order $\lebot$ to capture \emph{information preservation} between
terms -- $s \lebot t$ means that $t$ contains at least the same
information as $s$ does. The limit inferior -- and thus
$\prs$-convergence -- comprises the accumulated information that
eventually remains stable. Following the approach on terms, a partial
order suitable as a basis for convergence for term graph rewriting,
has to capture an appropriate notion of information preservation as
well.

One has to keep in mind, however, that term graphs encode an
additional dimension of information through \emph{sharing} of nodes,
i.e.\ the fact that nodes may have multiple positions. Since $\lebots$
specialises to $\lebot$ on terms, it does preserve the information on
the tree structure in the same way as $\lebot$ does. The difficult
part is to determine the right approach to the role of sharing.

Indeed, $\bot$-homomorphisms instantiate occurrences of $\bot$ and are
thereby able to introduce new information. But while
$\bot$-homomorphisms preserve the \emph{local} structure of each node,
they may change the \emph{global} structure of a term graph by
introducing sharing: for the term graphs $g_0$ and $g_1$ in
Figure~\ref{fig:convWeird}, we have an obvious $\bot$-homomorphism --
in fact a homomorphism -- $\phi\fcolon g_0 \homto_\bot g_1$ and thus
$g_0 \lebots g_1$.

\begin{figure}
  \centering
  \begin{tikzpicture}[node distance=15mm]%
    \node (r1) {$f$} %
    child{ node (n1) {$c$} }%
    child{ node (n2) {$c$} };%
    \node[right=of r1] (r2) {$f$}%
    child {%
      node (n2) {$c$}%
      edge from parent[transparent] %
    };%
    \draw[->] (r2)%
    edge [bend right=25] (n2)%
    edge [bend left=25] (n2);%
    
    \draw[single step,shorten=5mm] (r1) -- (r2);%
    
    \node[right=of r2] (r3) {$f$} %
    child{ node (n1) {$c$} }%
    child{ node (n2) {$c$} };%
    
    \draw[single step,shorten=5mm] (r2) -- (r3);%
    
    \node[right=of r3] (r4) {$f$}%
    child {%
      node (n2) {$c$}%
      edge from parent[transparent] %
    };%
    \draw[->] (r4)%
    edge [bend right=25] (n2)%
    edge [bend left=25] (n2);%
    \draw[single step,shorten=5mm] (r3) -- (r4);%

    \node[node distance=25mm,right=of r4] (r5) {$f$} %
    child{ node (n1) {$c$} }%
    child{ node (n2) {$c$} };%
    
    \draw[dotted,thick,shorten=10mm] (r4) -- (r5);%
    \begin{scope}[node distance=1cm]
      \node[below=of r1] {$(g_0)$};
      \node[below=of r2] {$(g_1)$};
      \node[below=of r3] {$(g_2)$};
      \node[below=of r4] {$(g_4)$};
      \node[below=of r5] {$(g_\omega)$};
    \end{scope}
  \end{tikzpicture}
  \caption{Limit inferior w.r.t.\ $\lebots$ in the presence of acyclic
    sharing.}
  \label{fig:convWeird}
\end{figure}
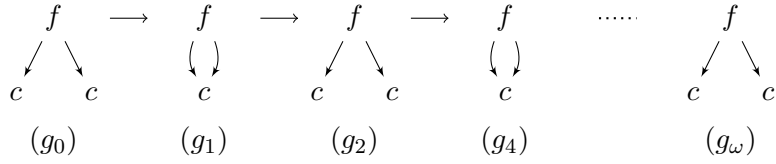

There are at least two different ways to interpret the differences in
$g_0$ and $g_1$. The first one dismisses $\lebots$ as a partial order
suitable for our purposes: the term graphs $g_0$ and $g_1$ contain
contradicting information. While in $g_0$ the two children of the
$f$-node are distinct, they are identical in $g_1$. We will indeed
follow this view in this paper and introduce a rigid partial order
$\lebotr$ that addresses this concern. There is, however, also a
second view that does not see $g_0$ and $g_1$ in contradiction: both
term graphs show the $f$-node with two successors, both of which are
labelled with $c$. The term graph $g_1$ merely contains the additional
piece of information that the two successor nodes of the $f$-node are
identical. The simple partial order $\lebots$, which follows this
view, is studied further in \cite{bahr12rta}.

One consequence of the above behaviour of $\lebots$ is that total term
graphs are not necessarily maximal w.r.t.\ $\lebots$, e.g.\ $g_0$ is
total but not maximal. The second -- more severe -- consequence is
that there can be no metric on total term graphs such that the limit
w.r.t.\ that metric coincides with the limit inferior on total term
graph. To see this consider the sequence $(g_i)_{i<\omega}$ of term
graphs illustrated in Figure~\ref{fig:convWeird}. Its limit inferior
w.r.t.\ $\lebots$ is the total term graph $g_\omega$. On the other
hand, there is no metric w.r.t.\ which $(g_i)_{i<\omega}$ converges
since the sequence alternates between two distinct term graphs. That
is, the correspondence between metric and partial order convergence
that we know from term rewriting, cf.\ Theorem~\ref{thr:strongExt}, is
impossible.

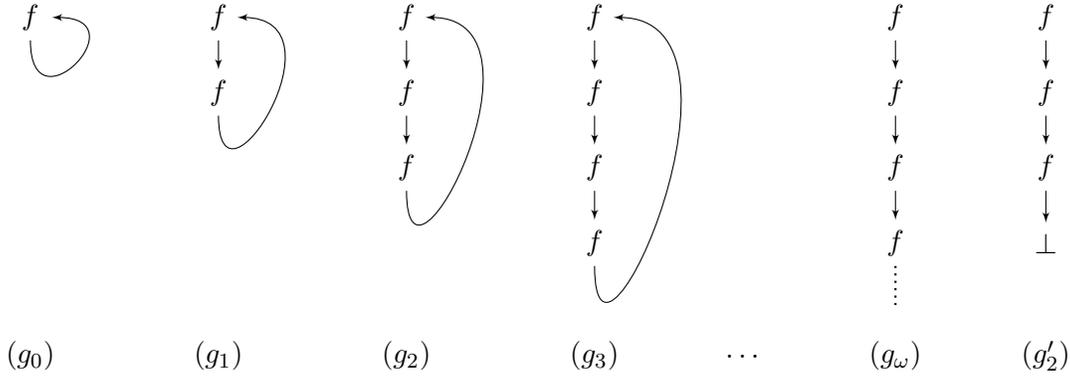
\begin{figure}
  \centering
  \begin{tikzpicture}[node distance=2.5cm]
    \node (r1) {$f$};%
    \draw (r1) edge[out=-90,in=0,loop,->](r1);%

    \node[right of=r1] (r2) {$f$}
    child {%
      node (b2) {$f$}%
    };%
    \draw (b2) edge[out=-90,in=0,->,min distance=15mm](r2);%

    \node[right of=r2] (r3) {$f$}
    child {%
      node {$f$}%
      child {%
        node (b3) {$f$}%
      }
    };%
    \draw (b3) edge[out=-90,in=0,->,min distance=18mm](r3);%

    \node[right of=r3] (r4) {$f$}
    child {%
      node {$f$}%
      child {%
        node {$f$}%
        child {%
          node (b4) {$f$}%
        }
      }
    };%
    \draw (b4) edge[out=-90,in=0,->,min distance=21mm](r4);%

    \node[node distance=4cm,right of=r4] (r5) {$f$}
    child {%
      node {$f$}%
      child {%
        node {$f$}%
        child {%
          node {$f$}%
          child[etc] {%
            node {}
          }
        }
      }
    };%

    \node[node distance=2cm,right of=r5] (r6) {$f$}
    child {%
      node {$f$}%
      child {%
        node {$f$}%
        child {%
          node {$\bot$}%
        }
      }
    };%

    \foreach \i/\l in {1/g_0,2/g_1,3/g_2,4/g_3,5/g_\omega,6/g'_2} {%
      \node[node distance=4.5cm,below of=r\i] (g\i) {$(\l)$};%
      }%
      \node[node distance=2cm,right of=g4] {$\dots$};
  \end{tikzpicture}
  \caption{Varying acyclic sharing.}
  \label{fig:acySharing}
\end{figure}

In order to avoid the introduction of sharing, we need to consider
$\bot$-homomorphisms that preserve the structure of term graphs more
rigidly, i.e.\ not only locally. Recall that by
Lemma~\ref{lem:occrep}, the structure of a term graph is essentially
given by the positions of nodes and their labelling. Labellings are
already taken into consideration by $\bot$-homomorphisms. Thus, we can
define a partial order $\lebotp$ that preserves the structure of term
graphs as follows: $g \lebotp h$ iff there is a $\bot$-homomorphism
$\phi\fcolon g \homto_\bot h$ with $\nodePos{h}{\phi(n)} =
\nodePos{g}{n}$ for all $n\in N^g$ with $\glab^g(n) \neq \bot$. While
this would again yield a complete semilattice, it is unfortunately too
restrictive. For example, consider the sequence of term graphs
$(g_i)_{i<\omega}$ depicted in Figure~\ref{fig:acySharing}. Due to the
cycle, we have for each term graph $g_i$ that $\bot$ is the only term
graph strictly smaller than $g_i$ w.r.t.\ $\lebotp$. The reason for
this is the fact that the only way to maintain the positions of the
root node of the term graph $g_i$ is to keep all nodes of the cycle in
$g_i$. Hence, in order to obtain a term graph $h$ with $h\lebotp g_i$,
we have to either keep the whole term graph $g_i$ or collapse it
completely, yielding $\bot$. For example, we neither have $g'_2
\lebotp g_2$ nor $g'_2 \lebotp g_3$ for the term graph $g'_2$
illustrated in Figure~\ref{fig:acySharing}. As a consequence, the
limit inferior of the sequence $(g_i)_{i<\omega}$ is $\bot$ and not
the expected term graph $g_\omega$.

The fact that the root nodes $g_2$ and $g'_2$ have different sets of
positions is solely caused by the edge to the root node of $g_2$ that
comes from \emph{below} and thus closes a \emph{cycle}. Even though
the edge occurs below the root node, it affects its positions. Cutting
off that edge, like in $g'_2$, changes the sharing. As a consequence,
in the complete semilattice $(\ipctgraphs,\lebotp)$, we do not obtain
the intuitively expected convergence behaviour depicted in
Figure~\ref{fig:gtransRed} on page~\pageref{fig:gtransRed}.

This observation suggests that we should only consider the
\emph{upward structure} of each node, ignoring the sharing that is
caused by edges occurring \emph{below} a node. We will see that by
restricting our attention to \emph{acyclic positions}, we indeed
obtain the desired properties for a partial order on term graphs.

Recall that a position $\pi$ in a term graph $g$ is called cyclic iff
there are positions $\pi_1, \pi_2$ with $\pi_1 < \pi_2 \le \pi$ such
that $\nodeAtPos{g}{\pi_1} = \nodeAtPos{g}{\pi_2}$, i.e.\ $\pi$ passes
a node twice. Otherwise it is called acyclic. We will use the notation
$\posAcy{g}$ for the set of all acyclic positions in $g$, and
$\nodePosAcy{g}{n}$ for the set of all acyclic positions of a node $n$
in $g$. That is, $\posAcy{g}$ is the set of positions in $g$ that pass
each node in $g$ at most once. Clearly, every node has at least one
acyclic position, i.e.\ $\nodePosAcy{g}{n}$ is a non-empty set.

\begin{definition}[rigidity]
  \label{def:strongD-hom}
  % rigidity %
  Let $\Sigma$ be a signature, $\Delta \subseteq \Sigma^{(0)}$ and
  $g,h \in \itgraphs$ such that $\phi\fcolon g \homto_\Delta h$.
  \begin{enumerate}[(i)]
  \item Given $n \in N^g$, $\phi$ is said to be \emph{rigid} in $n$ if
    it satisfies the equation
    \[
    \nodePosAcy{g}{n} = \nodePosAcy{h}{\phi(n)}
    \tag{rigid}
    \]
  \item $\phi$ is called a \emph{rigid} $\Delta$-homomorphism if it is
    rigid in all $n \in N^g$ with $\glab^g(n) \nin \Delta$.
  \end{enumerate}
\end{definition}
\begin{proposition}[category of rigid $\Delta$-homomorphisms]
  \label{prop:catStrongD-hom}
  The rigid $\Delta$-homomorphisms on $\itgraphs$ form a subcategory
  of the category of $\Delta$-homomorphisms on $\itgraphs$.
\end{proposition}
\begin{proof}
  Straightforward.
\end{proof}
Note that, for each node $n$ in a term graph $g$, the positions in
$\nodePosAcy{g}{n}$ are minimal positions of $n$ w.r.t.\ the prefix
order. Rigid $\bot$-homomorphisms thus preserve the upward structure
of each non-$\bot$-node and, therefore, provide the desired structure
for a partial order that captures information preservation on term
graphs:
\begin{definition}[rigid partial order $\lebotr$] For every $g,h \in
  \iptgraphs$, define $g \lebotr h$ iff there is a rigid
  $\bot$-homomorphism $\phi\fcolon g \homto_\bot h$.
\end{definition}

\begin{proposition}[partial order $\lebotr$]
  \label{prop:tgraphPo}
  % partial order on canonical term graphs %
  The relation $\lebotr$ is a partial order on $\ipctgraphs$.
\end{proposition}
\begin{proof}
  Reflexivity and transitivity of $\lebotr$ follow immediately from
  Proposition~\ref{prop:catStrongD-hom}. For antisymmetry, assume $g
  \lebotr h$ and $h \lebotr g$. By Proposition~\ref{prop:catgraph},
  this implies $g \isom_\bot h$. Corollary~\ref{cor:sig-isom-isom} then
  yields that $g \isom h$. Hence, according to
  Proposition~\ref{prop:canon}, $g = h$.
\end{proof}

\begin{example}
  \label{ex:liminf'}
  Figure~\ref{fig:gtransRed} on page \pageref{fig:gtransRed} shows a
  sequence $(h_\iota)_{\iota<\omega}$ of term graphs and its limit
  inferior $h_\omega$ in $(\ipctgraphs,\lebotr)$: a cyclic list
  structure is repeatedly rewritten by inserting an element $b$ in
  front of the $a$. We can see that in each step the newly inserted
  $b$ (including the additional $\cons$-node) remains unchanged
  afterwards. In terms of positions, however, each of the nodes
  changes in each step since the length of the cycle in the term graph
  grows with each step. Since this affects only cyclic positions, we
  still get the following sequence $(\Glb_{\beta\le\iota<\omega}
  h_\iota)_{\beta<\omega}$ of canonical term trees:
  \[
  \seq{\bot\cons\bot,b\cons\bot\cons\bot,b\cons b\cons
    \bot\cons\bot,\dots}
  \]
  The least upper bound of this sequence $(\Glb_{\beta\le\iota<\omega}
  h_\iota)_{\beta<\omega}$ and thus the limit inferior of
  $(h_\iota)_{\iota<\omega}$ is the infinite canonical term tree
  $h_\omega = b\cons b\cons b\cons \dots$. Since the cycle changes in
  each step and is thus cut through in each element of
  $(\Glb_{\beta\le\iota<\omega} h_\iota)_{\beta<\omega}$, the limit
  inferior has no cycles at all.

  Note that we do not have this intuitively expected convergence
  behaviour for the partial order $\lebotp$ based on positions: since
  the length of the cycle grows along the sequence
  $(h_\iota)_{\iota<\omega}$, we have that the set of positions of the
  root nodes changes constantly. Hence, the limit inferior of
  $(h_\iota)_{\iota<\omega}$ in $(\ipctgraphs,\lebotp)$ is $\bot$.
\end{example}

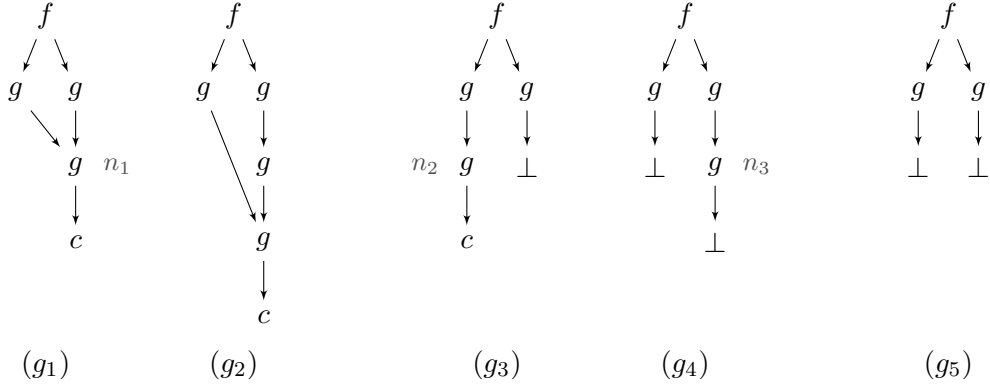
\begin{figure}
  \centering
  \begin{tikzpicture}[->,node distance=2cm,sibling distance=8mm]
    \node (r1) {$f$}
    child{
      node (n1) {$g$}
    }
    child {
      node {$g$}
      child {
        node[node name=0:n_1] (n2) {$g$}
        child { node {$c$}}
      }
    };
    \draw (n1) edge (n2);
    
    \node[right=of r1] (r2) {$f$}
    child{
      node (n1) {$g$}
    }
    child {
      node {$g$}
      child {
        node {$g$}
        child {
          node (n2) {$g$}
          child { node {$c$}}
        }
      }
    };
    \draw (n1) edge (n2);

    \node[node distance=3cm, right=of r2] (r5) {$f$}
    child{
      node {$g$}
      child {
        node [node name=180:n_2] {$g$}
        child {
          node {$c$}
        }
      }
    } child {
      node {$g$}
      child {
        node {$\bot$}
      }
    };
    
    \node[right=of r5] (r6) {$f$}
    child{
      node {$g$}
      child {
        node {$\bot$}
      }
    } child {
      node {$g$}
      child { 
        node [node name=0:n_3] {$g$}
        child {
          node {$\bot$}
        }
      }
    };

    \node[node distance=3cm,right=of r6] (r7) {$f$}%
    child{%
      node {$g$}%
      child {%
        node {$\bot$}%
      }%
    } child {%
      node {$g$}%
      child {%
        node {$\bot$}%
      }%
    };%

    \begin{scope}[node distance=4cm]
      \node[below=of r1] {$(g_1)$};
      \node[below=of r2] {$(g_2)$};
      \node[below=of r5] {$(g_3)$};
      \node[below=of r6] {$(g_4)$};
      \node[below=of r7] {$(g_5)$};
    \end{scope}
  \end{tikzpicture}
  \caption{Term graphs $g_1,g_2$ with maximal lower bounds $g_3,g_4$
    w.r.t.\ $\leboti$.}
  \label{fig:lebotCounter}
\end{figure}

The partial order $\lebotr$ based on rigid $\bot$-homomorphisms is
defined in a rather non-local fashion as the definition of rigidity
uses the set of \emph{all} acyclic positions. This poses the question
whether there is a more natural definition of a suitable partial
order. One such candidate is the partial order $\leboti$, which uses
injectivity in order to restrict the introduction of sharing: $g
\leboti h$ iff there is a $\bot$-homomorphism $\phi\fcolon
g\homto_\bot h$ that is injective on non-$\bot$-nodes, i.e.\ $\phi(n)
= \phi(m)$ and $\glab^g(n),\glab^g(m) \neq \bot$ implies $n =
m$. While this yields indeed a cpo on $\ipctgraphs$, we do not get a
complete semilattice. To see this, consider
Figure~\ref{fig:lebotCounter}. The two term graphs $g_3,g_4$ are two
distinct maximal lower bounds of the two term graphs $g_1,g_2$ w.r.t.\
the partial order $\leboti$. Hence, the set $\set{g_1,g_2}$ does not
have a greatest lower bound in $(\ipctgraphs,\leboti)$, which is
therefore not a complete semilattice. The same phenomenon occurs if we
consider a partial order derived from $\bot$-homomorphisms that are
injective on all nodes.

The rigid partial order $\lebotr$ resolves the issue of $\leboti$
illustrated in Figure~\ref{fig:lebotCounter}: $g_3$ and $g_4$ are not
lower bounds of $g_1$ and $g_2$ w.r.t.\ $\lebotr$. The (unique)
$\bot$-homomorphism from $g_3$ to $g_1$ is not rigid as it maps the
node $n_2$ to $n_1$ and $\nodePosAcy{g_3}{n_2} = \set{\seq{0,0}}$
whereas $\nodePosAcy{g_1}{n_1} = \set{\seq{0,0},\seq{1,0}}$. Hence,
$g_3 \not\lebotr g_1$. Likewise, $g_4 \not\lebotr g_1$ as the (unique)
$\bot$-homomorphism from $g_4$ to $g_1$ maps $n_3$ to $n_1$, which
again have different acyclic positions. We do find, however, a
greatest lower bound of $g_1$ and $g_2$ w.r.t.\ $\lebotr$, viz.\
$g_5$.

\subsection{The Rigid Partial Order}
\label{sec:rigid-partial-order}
In the remainder of this section, we will study the rigid partial
order $\lebotr$. In particular, we shall give a characterisation of
rigidity in terms of labelled quotient trees analogous to
Lemma~\ref{lem:occrephom}, show that $(\ipctgraphs,\lebotr)$ forms a
complete semilattice, illustrate the resulting mode of convergence,
and give a characterisation of term graphs that are maximal w.r.t.\
$\lebotr$.

The partial order $\leboti$, derived from injective
$\bot$-homomorphisms, failed to form a complete semilattice, which is
why we abandoned that approach. The following lemma shows that
rigidity is, in fact, a stronger property than injectivity on
non-$\Delta$-nodes. Hence, $\lebotr$ is a restriction of $\leboti$.
\begin{lemma}[rigid $\Delta$-homomorphisms are injective for
  non-$\Delta$-nodes]
  \label{lem:strongD-homInj}
  % rigid Delta-homomorphisms are injective on non-Delta nodes %
  Let $g, h \in \itgraphs$ and $\phi\fcolon g \homto_\Delta h$ rigid. Then
  $\phi$ is injective for all non-$\Delta$-nodes in $g$. That is, for
  two nodes $n, m \in N^g$ with $\glab^g(n), \glab^g(m) \nin \Delta$
  we have that $\phi(n) = \phi(m)$ implies $n = m$.
\end{lemma}
\begin{proof}
  Let $n, m \in N^g$ with $\glab^g(n), \glab^g(m) \nin \Delta$ and
  $\phi(n) = \phi(m)$. Since $\phi$ is rigid, it is rigid in $n$ and
  $m$. That is, in particular we have $\nodePosAcy{h}{\phi(n)}
  \subseteq \nodePos{g}{n}$ and $\nodePosAcy{h}{\phi(m)} \subseteq
  \nodePos{g}{m}$. Moreover, because $\nodePosAcy{h}{\phi(n)} =
  \nodePosAcy{h}{\phi(m)} \neq \emptyset$, we can conclude that
  $\nodePos{g}{n} \cap \nodePos{g}{m} \neq \emptyset$ and, therefore,
  $m = n$.
\end{proof}

\subsubsection{Characterising Rigidity}
\label{sec:char-rigid}

The goal of this subsection is to give a characterisation of rigidity
in terms of labelled quotient trees. We will then combine this
characterisation with Lemma~\ref{lem:occrephom} to obtain a
characterisation of the partial order $\lebotr$.

The following lemma provides a characterisation of rigid
$\Delta$-homomorphisms that reduces the proof obligations necessary to
show that a $\Delta$-homomorphism is rigid.
\begin{lemma}[rigidity]
  \label{lem:presShar}
  % preservation of sharing %
  Let $g,h \in \itgraphs$, $\phi\fcolon g \homto_\Delta h$.  Then $\phi$ is
  rigid iff $\nodePosAcy{h}{\phi(n)} \subseteq \nodePos{g}{n}$ for all $n
  \in N^g$ with $\glab^g(n) \nin \Delta$.
\end{lemma}
\begin{proof}
  The ``only if'' direction is trivial. For the ``if'' direction,
  suppose that $\phi$ satisfies $\nodePosAcy{h}{\phi(n)} \subseteq
  \nodePos{g}{n}$ for all $n \in N^g$ with $\glab^g(n) \nin
  \Delta$. In order to prove that $\phi$ is rigid, we will show that
  $\nodePosAcy{h}{\phi(n)} = \nodePosAcy{g}{n}$ holds for each $n \in
  N^g$ with $\glab^g(n) \nin \Delta$.

  We first show the inclusion $\nodePosAcy{h}{\phi(n)} \subseteq
  \nodePosAcy{g}{n}$. For this purpose, let $\pi \in
  \nodePosAcy{h}{\phi(n)}$. Due to the hypothesis, this implies that
  $\pi \in \nodePos{g}{n}$. Now suppose that $\pi$ is cyclic in $g$,
  i.e.\ there are two positions $\pi_1, \pi_2$ of a node $m \in N^g$
  with $\pi_1 < \pi_2 \le \pi$. By Lemma~\ref{lem:canhom}, we can
  conclude that $\pi_1, \pi_2 \in \nodePos{h}{\phi(m)}$. This is a
  contradiction to the assumption that $\pi$ is acyclic in $h$. Hence,
  $\pi \in \nodePosAcy{g}{n}$.

  For the other inclusion, assume some $\pi \in
  \nodePosAcy{g}{n}$. Using Lemma~\ref{lem:canhom} we obtain that $\pi
  \in \nodePos{h}{\phi(n)}$. It remains to be shown that $\pi$ is
  acyclic in $h$. Suppose that this is not true, i.e.\ there are two
  positions $\pi_1, \pi_2$ of a node $m \in N^h$ with $\pi_1 < \pi_2
  \le \pi$. Note that since $\pi \in \pos{g}$, also $\pi_1, \pi_2 \in
  \pos{g}$. Let $m_i = \nodeAtPos{g}{\pi_i}$, $i=1,2$. According to
  Lemma~\ref{lem:canhom}, we have that $\phi(m_1) = m =
  \phi(m_2)$. Moreover, observe that $g(\pi_1), g(\pi_2) \nin \Delta$:
  $g(\pi_1)$ cannot be a nullary symbol because $\pi_1 < \pi \in
  \pos{g}$. The same argument applies for the case that $\pi_2 <
  \pi$. If this is not the case, then $\pi_2 = \pi$ and $g(\pi) \nin
  \Delta$ follows from the assumption that $\glab^g(n) \nin
  \Delta$. Thus, we can apply Lemma~\ref{lem:strongD-homInj} to
  conclude that $m_1 = m_2$. Consequently, $\pi$ is cyclic in $g$, which
  contradicts the assumption. Hence, $\pi \in \nodePosAcy{h}{\phi(n)}$.
\end{proof}

From the above lemma we learn that $\Delta$-isomorphisms are also
rigid $\Delta$-homomorphisms.
\begin{corollary}[$\Delta$-isomorphisms are rigid]
  \label{cor:isomStrong}
  % $\Delta$-isomorphisms are rigid %
  Let $g, h \in \itgraphs$. If $\phi\fcolon g \isoto_\Delta h$, then $\phi$
  is a rigid $\Delta$-homomorphism.
\end{corollary}
\begin{proof}
  This follows from Lemma~\ref{lem:isomCan} and
  Lemma~\ref{lem:presShar}.
\end{proof}

Similarly to Lemma~\ref{lem:occrephom}, we provide a characterisation
of rigid $\Delta$-homomorphisms in terms of labelled quotient trees:
\begin{lemma}[characterisation of rigid $\Delta$-homomorphisms]
  \label{lem:chaStrongD-hom}
  % characterisation of rigid Delta-homomorphisms %
  Given $g, h \in \itgraphs$, a $\Delta$-homomorphism $\phi\fcolon g
  \homto_\Delta h$ is rigid iff
  \[
  \pi \sim_h \pi' \quad \implies \quad  \pi \sim_g \pi' \quad \text{ for all
  } \pi \in \pos{g} \text{ with } g(\pi) \nin \Delta \text{ and } \pi' \in \posAcy{h}.
  \]
\end{lemma}
\begin{proof}
  For the ``only if'' direction, assume that $\phi$ is
  rigid. Moreover, let $\pi \in \pos{g}$ with $g(\pi) \nin \Delta$ and
  $\pi' \in \posAcy{h}$ such that $\pi \sim_h \pi'$, and let $n =
  \nodeAtPos{g}{\pi}$. By applying Lemma~\ref{lem:canhom}, we get that
  $\pi \in \nodePos{h}{\phi(n)}$. Because of $\pi \sim_h \pi'$, also
  $\pi' \in \nodePos{h}{\phi(n)}$. Since, according to the assumption,
  $\pi'$ is acyclic in $h$, we know in particular that $\pi' \in
  \nodePosAcy{h}{\phi(n)}$. Since $\phi$ is rigid and $\glab^g(n) \nin
  \Delta$, we know that $\phi$ is rigid in $n$ which yields that $\pi'
  \in \nodePos{g}{n}$. Hence, $\pi \sim_g \pi'$.

  For the converse direction, let $n \in N^g$ with $\glab^g(n) \nin
  \Delta$. We need to show that $\phi$ is rigid in $n$. Due to
  Lemma~\ref{lem:presShar}, it suffices to show that
  $\nodePosAcy{h}{\phi(n)} \subseteq \nodePos{g}{n}$. Since
  $\nodePos{g}{n} \neq \emptyset$, we can choose some $\pi^*\in
  \nodePos{g}{n}$. Then, according to Lemma~\ref{lem:canhom}, also
  $\pi^*\in\nodePos{h}{\phi(n)}$. Let $\pi \in
  \nodePosAcy{h}{\phi(n)}$. Then $\pi^* \sim_h \pi$ holds. Since $\pi$
  is acyclic in $h$ and $g(\pi^*)\nin \Delta$, we can use the
  hypothesis to obtain that $\pi^* \sim_g \pi$ holds which shows that
  $\pi \in \nodePos{g}{n}$.
\end{proof}

Note that the above characterisation of rigidity is independent of the
$\Delta$-homomorphism at hand. This is expected since
$\Delta$-homomorphisms between a given pair of term graphs are unique.

By combining the above characterisation of rigidity with the
corresponding characterisation of $\Delta$-homomorphisms, we obtain
the following compact characterisation of $\lebotr$:
\begin{corollary}[characterisation of $\lebotr$]
  \label{cor:chaTgraphPo}
  % characterisation of $\lebotr$ %
  Let $g,h \in \iptgraphs$. Then $g \lebotr h$ iff the following
  conditions are met:
  \begin{enumerate}[\em(a)]
  \item $\pi \sim_g \pi' \quad \implies \quad  \pi \sim_h \pi'$ \quad for all
    $\pi,\pi' \in \pos{g}$
    \label{item:chaTGraphPo1}
  \item $\pi \sim_h \pi' \quad \implies \quad \pi \sim_g \pi'$ \quad
    for all $\pi \in \pos{g}$ with $g(\pi) \in \Sigma$ and $\pi' \in
    \posAcy{h}$
    \label{item:chaTGraphPo2}
 \item $g(\pi) = h(\pi)$ \quad for all $\pi\in \pos{g}$ with $g(\pi) \in
   \Sigma$.
    \label{item:chaTGraphPo3}
  \end{enumerate}
\end{corollary}
\begin{proof}
  This follows immediately from Lemma~\ref{lem:occrephom} and
  Lemma~\ref{lem:chaStrongD-hom}.
\end{proof}%
\def\lebota{(\ref{item:chaTGraphPo1})}%
\def\lebotb{(\ref{item:chaTGraphPo2})}%
\def\lebotc{(\ref{item:chaTGraphPo3})}%
Note that for term trees (\ref{item:chaTGraphPo2}) is always true and
(\ref{item:chaTGraphPo1}) follows from
(\ref{item:chaTGraphPo3}). Hence, on term trees, $\lebotr$ is
characterised by (\ref{item:chaTGraphPo3}) alone. This observation
shows that $\lebotr$ is indeed a generalisation of $\lebot$.
\begin{corollary}
  \label{cor:lebotrGeneralise}
  For all $s,t \in \ipterms$, we have that $s \lebotr t$ iff $s \lebot
  t$.
\end{corollary}
\begin{proof}
  Follows from Corollary~\ref{cor:chaTgraphPo}.
\end{proof}

\subsubsection{Convergence}
\label{sec:convergence}

In the following, we shall show that $\lebotr$ indeed forms a complete
semilattice on $\ipctgraphs$. We begin by showing that it constitutes
a complete partial order.
\begin{theorem}[$\lebotr$ is a cpo]
  \label{thm:graphCpo}
  % $\lebotr$ is a cpo %
  The pair $(\ipctgraphs,\lebotr)$ forms a cpo. In particular, it has
  the least element $\bot$, and the least upper bound of a directed
  set $G$ is given by the following labelled quotient tree
  $(P,l,\sim)$:
  \begin{gather*}
    P = \bigcup\limits_{g\in G} \pos{g} \hspace{30pt}%
    \sim\ = \bigcup\limits_{g\in G} \sim_g \hspace{30pt}%
    l(\pi) =
    \begin{cases}
      f & \text{ if } f \in \Sigma \text{ and } \exists g \in G. \;
      g(\pi) = f \\
      \bot &\text{ otherwise }
    \end{cases}
  \end{gather*}
\end{theorem}
\begin{proof}
  The least element of $\lebotr$ is obviously $\bot$. Hence, it
  remains to be shown that each directed subset $G$ of $\ipctgraphs$
  has a least upper bound w.r.t.\ $\lebotr$. To this end, we show that
  the canonical term graph $\ol g$ given by the labelled quotient tree
  $(P,l,\sim)$ described above is indeed the lub of $G$. We will make
  extensive use of Corollary~\ref{cor:chaTgraphPo} to do
  so. Therefore, we write \lebota, \lebotb, \lebotc\ to refer to
  corresponding conditions of Corollary~\ref{cor:chaTgraphPo}.

  At first we need to show that $l$ is indeed well-defined. For this
  purpose, let $g_1,g_2 \in G$ and $\pi\in \pos{g_1}\cap\pos{g_2}$
  with $g_1(\pi), g_2(\pi) \in\Sigma$. Since $G$ is directed, there is
  some $g \in G$ such that $g_1, g_2 \lebotr g$. By \lebotc{}, we can
  conclude $g_1(\pi) = g(\pi) = g_2(\pi)$.
  
  Next we show that $(P,l,\sim)$ is indeed a labelled quotient
  tree. Recall that $\sim$ needs to be an equivalence relation. For
  the reflexivity, assume that $\pi \in P$. Then there is some $g \in
  G$ with $\pi \in \pos{g}$. Since $\sim_g$ is an equivalence
  relation, $\pi \sim_g \pi$ must hold and, therefore, $\pi \sim
  \pi$. For the symmetry, assume that $\pi_1 \sim \pi_2$. Then there
  is some $g\in G$ such that $\pi_1 \sim_g \pi_2$. Hence, we get
  $\pi_2 \sim_g \pi_1$ and, consequently, $\pi_2 \sim \pi_1$. In order
  to show transitivity, assume that $\pi_1 \sim \pi_2, \pi_2 \sim
  \pi_3$. That is, there are $g_1, g_2\in G$ with $\pi_1 \sim_{g_1}
  \pi_2$ and $\pi_2 \sim_{g_2} \pi_3$. Since $G$ is directed, we find
  some $g\in G$ such that $g_1,g_2\lebotr g$. By \lebota{}, this
  implies that also $\pi_1 \sim_{g} \pi_2$ and $\pi_2 \sim_{g}
  \pi_3$. Hence, $\pi_1 \sim_{g} \pi_3$ and, therefore, $\pi_1 \sim
  \pi_3$.

  For the reachability condition, let $\pi\concat \seq i \in P$. That is,
  there is a $g \in G$ with $\pi \concat \seq i \in \pos{g}$. Hence, $\pi
  \in \pos{g}$, which in turn implies $\pi \in P$. Moreover, $\pi
  \concat \seq i \in \pos{g}$ implies that $i < \srank{g(\pi)}$. Since
  $g(\pi)$ cannot be a nullary symbol and in particular not $\bot$, we
  obtain that $l(\pi) = g(\pi)$. Hence, $i < \srank{l(\pi)}$.

  For the congruence condition, assume that $\pi_1 \sim \pi_2$ and that
  $l(\pi_1) = f$. If $f \in \Sigma$, then there are $g_1,g_2 \in G$
  with $\pi_1 \sim_{g_1} \pi_2$ and $g_2(\pi_1) = f$. Since $G$ is
  directed, there is some $g\in G$ such that $g_1, g_2 \lebotr
  g$. Hence, by \lebota{} respectively \lebotc{}, we have $\pi_1 \sim_{g}
  \pi_2$ and $g(\pi_1) = f$. Using Lemma~\ref{lem:occrep} we can
  conclude that $g(\pi_2) = g(\pi_1) = f$ and that $\pi_1 \concat i
  \sim_g \pi_2 \concat i$ for all $i < \srank{g(\pi_1)}$. Because $g \in
  G$, it holds that $l(\pi_2) = f$ and that $\pi_1 \concat i \sim \pi
  \concat i$ for all $i < \srank{l(\pi_1)}$. If $f = \bot$, then also
  $l(\pi_2) = \bot$, for if $l(\pi_2) = f'$ for some $f' \in\Sigma$,
  then, by the symmetry of $\sim$ and the above argument (for the case
  $f\in\Sigma$), we would obtain $f = f'$ and, therefore, a
  contradiction. Since $\bot$ is a nullary symbol, the remainder of
  the condition is vacuously satisfied.

  This shows that $(P,l,\sim)$ is a labelled quotient tree which, by
  Lemma~\ref{lem:occrep}, uniquely defines a canonical term
  graph. Next we show that the thus obtained term graph $\overline{g}$
  is an upper bound for $G$. To this end, let $g\in G$. We will show
  that $g \lebotr \overline{g}$ by establishing \lebota,\lebotb\ and
  \lebotc. \lebota\ and \lebotc\ are an immediate consequence of the
  construction. For \lebotb{}, assume that $\pi_1 \in \pos{g}$,
  $g(\pi_1) \in \Sigma$, $\pi_2 \in \posAcy{\ol g}$ and $\pi_1 \sim
  \pi_2$. We will show that then also $\pi_1 \sim_g \pi_2$
  holds. Since $\pi_1 \sim \pi_2$, there is some $g' \in G$ with
  $\pi_1 \sim_{g'} \pi_2$. Because $G$ is directed, there is some $g^*
  \in G$ with $g,g' \lebotr g^*$. Using \lebota{}, we then get that
  $\pi_1 \sim_{g^*} \pi_2$. Note that since $\pi_2$ is acyclic in $\ol
  g$, it is also acyclic in $g^*$: Suppose that this is not the case,
  i.e.\ there are positions $\pi_3, \pi_4$ with $\pi_3 < \pi_4 \le
  \pi_2$ and $\pi_3 \sim_{g^*} \pi_4$. But then we also have $\pi_3
  \sim \pi_4$, which contradicts the assumption that $\pi_2$ is acyclic
  in $\ol g$. With this knowledge we are able to apply \lebotb\ to
  $\pi_1 \sim_{g^*} \pi_2$ in order to obtain $\pi_1 \sim_g \pi_2$.

  In the final part of this proof, we will show that $\ol g$ is the
  least upper bound of $G$. For this purpose, let $\oh{g}$ be an
  upper bound of $G$, i.e.\ $g \lebotr \oh{g}$ for all $g\in G$. We
  will show that $\overline{g} \lebotr \oh{g}$ by establishing
  \lebota, \lebotb\ and \lebotc. For \lebota{}, assume that $\pi_1
  \sim \pi_2$. Hence, there is some $g \in G$ with $\pi_1 \sim_g
  \pi_2$. Since, by assumption, $g \lebotr \oh g$, we can conclude
  $\pi_1 \sim_{\oh g} \pi_2$ using \lebota. For \lebotb{}, assume
  $\pi_1 \in P$, $l(\pi_1) \in \Sigma$, $\pi_2 \in \posAcy{\oh g}$ and
  $\pi_1 \sim_{\oh g} \pi_2$. That is, there is some $g \in G$ with
  $g(\pi_1) \in \Sigma$. Together with $g \lebotr \oh g$ this implies
  $\pi_1 \sim_g \pi_2$ by \lebotb. $\pi_1 \sim \pi_2$ follows
  immediately. For \lebotc{}, assume $\pi \in P$ and $l(\pi) = f
  \in\Sigma$. Then there is some $g\in G$ with $g(\pi) = f$. Applying
  \lebotc\ then yields $\oh g(\pi) = f$ since $g \lebotr \oh g$.
\end{proof}

\begin{rem}
  \label{rem:lebot}
  Following Remark~\ref{rem:catTgraphs}, we define an order $\lebotr$
  on $\quotient{\iptgraphs}{\isom}$ which is isomorphic to the order
  $\lebotr$ on $\ipctgraphs$. Define $\eqc{g}\isom \lebotr
  \eqc{h}\isom$ iff there is a rigid $\bot$-homomorphism $\phi\fcolon
  g \homto_\bot h$.

  The extension of $\lebotr$ to equivalence classes is easily seen to
  be well-defined: assume some rigid $\bot$-homomorphism $\phi\fcolon
  g \homto_\bot h$ and two isomorphisms $g' \isom g$ and $h' \isom
  h$. Since, by Corollary~\ref{cor:isomStrong}, isomorphisms are also
  rigid ($\bot$-)homomorphisms, we have two rigid $\bot$-homomorphisms
  $\phi_1\fcolon g' \homto_\bot g$ and $\phi_2\fcolon h \homto_\bot
  h'$. Hence, by Proposition~\ref{prop:catStrongD-hom}, $\phi_2 \circ
  \phi \circ \phi_1$ is a rigid $\bot$-homomorphism from $g'$ to $h'$.

  The isomorphism illustrated above allows us switch between the two
  partially ordered sets $(\ipctgraphs,\lebotr)$ and
  $(\quotient{\iptgraphs}{\isom},\lebotr)$ in order to use the
  structure that is more convenient for the given setting. In
  particular, the proof of Lemma~\ref{lem:graphCompLub} below is based
  on this isomorphism.
\end{rem}

By Proposition~\ref{prop:bcpoCompSemi}, a cpo is a complete
semilattice iff each two compatible elements have a least upper
bound. Recall that compatible elements in a partially ordered set are
elements that have a common upper bound. We make use of this
proposition in order to show that $(\ipctgraphs,\lebotr)$ is a
complete semilattice. However, showing that each two term graphs $g,h
\in \ipctgraphs$ with a common upper bound also have a least upper
bound is not easy. The issue that makes the construction of the lub of
compatible term graphs a bit more complicated than in the case of
directed sets is illustrated in Figure~\ref{fig:lubCompGraph}. Note
that the lub $g \lub h$ of the term graphs $g$ and $h$ has an
additional cycle. The fact that in $g \lub h$ the second successor of
$r$ has to be $r$ itself is enforced by $g$ saying that the first
successor of $r_1$ is $r_1$ itself and by $h$ saying that the first
and the second successor of $r_2$ must be identical. Because of the
additional cycle in $g \lub h$, we have that the set of positions in
$g \lub h$ is a proper superset of the union of the sets of positions
in $g$ and $h$. This makes the construction of $g \lub h$ using a
labelled quotient tree quite intricate.

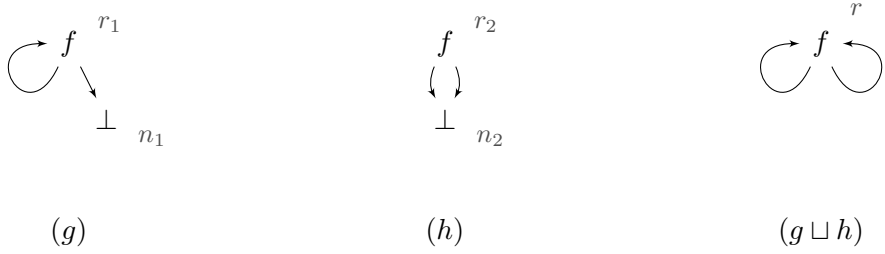
\begin{figure}
  \centering
  \begin{tikzpicture}
    \node (r1) [node name=5:r_1] at (0,0) {$f$} 
    edge [out=245, in=180,loop] ()
    child [missing]
    child{ node (n1) [node name=-5:n_1] {$\bot$} };
    \node at (0,-2.5) {($g$)};

    \node (r2) [node name=5:r_2]  at (5,0) {$f$}
    child {
      node (n2) [node name=-5:n_2] {$\bot$}
      edge from parent[transparent] 
    };
    \draw[->] (r2)
    edge [bend right=25] (n2)
    edge [bend left=25] (n2);
    \node at (5,-2.5) {($h$)};

    \node (r) [node name=45:r]  at (10,0) {$f$}
    edge [out=245, in=180,loop] ()
    edge [out=295, in=0,loop] ();
    \node at (10,-2.5) {($g\lub h$)};
  \end{tikzpicture}
  \caption{Least upper bound $g\lub h$ of compatible term graphs
    $g$ and $h$.}
\label{fig:lubCompGraph}
\end{figure}

Our strategy to construct the lub is to form the disjoint union of the
two term graphs in question and then identify nodes that have a common
position w.r.t.\ the term graph they originate from. In our example,
we have four nodes $r_1, n_1, r_2$ and $n_2$. At first $r_1$ and $r_2$
have to be identified as both have the position $\emptyseq$. Next,
$r_1$ and $n_2$ are identified as they share the position $\seq
0$. And eventually, also $n_2$ and $n_1$ are identified since they
share the position $\seq 1$. Hence, all four nodes have to be
identified. The result is, therefore, a term graph with a single node
$r$. The following lemma and its proof, given in
Appendix~\ref{sec:proof-lemma-refl-1}, show that, for any two
compatible term graphs, this construction always yields their lub.
\begin{lemma}[compatible elements have lub]
  \label{lem:graphCompLub}
  % compatible elements have lub %
  Each pair $g_1,g_2$ of compatible term graphs in
  $(\ipctgraphs,\lebotr)$ has a least upper bound.
\end{lemma}

\begin{theorem}%[complete semilattice $\lebotr$]
  \label{thm:graphBcpo}
  The pair $(\ipctgraphs,\lebotr)$ forms a complete semilattice.
\end{theorem}
\begin{proof}
  This is, by Proposition~\ref{prop:bcpoCompSemi}, a consequence of
  Theorem~\ref{thm:graphCpo} and Lemma~\ref{lem:graphCompLub}.
\end{proof}

In particular, this means that the limit inferior is defined for every
sequence of term graphs.
\begin{corollary}[limit inferior of $\lebotr$]
  \label{cor:lebotrLiminf}
  % limit inferior %
  Each sequence in $(\ipctgraphs,\lebotr)$ has a limit inferior.
\end{corollary}

Recall that the intuition of the limit inferior on terms is that it
contains the accumulated information that eventually remains stable in
the sequence. This interpretation is, of course, based on the partial
order $\lebot$ on terms, which embodies the underlying notion of
``information encoded in a term''.

The same interpretation can be given for the limit inferior based on
the rigid partial order $\lebotr$ on term graphs. Given a sequence
$(g_\iota)_{\iota<\alpha}$ of term graphs, its limit inferior
$\liminf_{\iota\limto\alpha} g_\iota$ is the term graph that contains
the accumulation of all pieces of information that from some point
onwards remain unchanged in $(g_\iota)_{\iota<\alpha}$.

\begin{example}
  \label{ex:liminf}
  \ref{fig:doubleTransRed} and \ref{fig:doubleTransRedDiv} on page
  \pageref{fig:doubleTransRed} each show a sequence of term graphs and
  its limit inferior in $(\ipctgraphs,\lebotr)$.

  \begin{enumerate}[(i)]
  \item%
    \label{item:liminf2}%
    Figure~\ref{fig:doubleTransRed} shows a simple example of how
    acyclic sharing is preserved by the limit inferior. The
    corresponding sequence $(\Glb_{\beta\le\iota<\omega}
    g_\iota)_{\beta<\omega}$ of greatest lower bounds is given as
    follows:
    
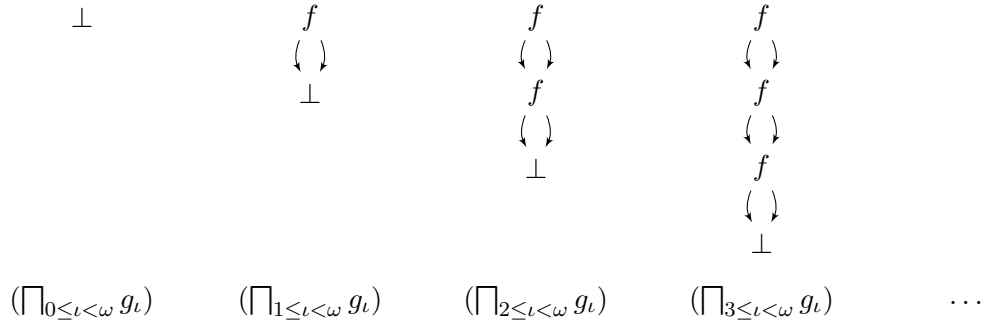
\begin{center}
      \begin{tikzpicture}
            
        \node (g1) {$\bot$};%

        \node () at ($(g1) + (0,-3.8)$)%
        {$(\Glb_{0\le\iota<\omega} g_\iota)$};
        
        \node [node distance=2.5cm,right=of g1] (g2) {$f$};%
        \node [node distance=10mm,below of=g2] (b) {$\bot$};%

        \draw[->] (g2)%
        edge [bend right=25] (b)%
        edge [bend left=25] (b);%

        \node () at ($(g2) + (0,-3.8)$)%
        {$(\Glb_{1\le\iota<\omega} g_\iota)$};
        
        \node [node distance=2.5cm,right=of g2] (g3) {$f$};%
        \node [node distance=10mm,below of=g3] (b) {$f$};%
        \node [node distance=10mm,below of=b] (b2) {$\bot$};%
        \draw[->] (g3)%
        edge [bend right=25] (b)%
        edge [bend left=25] (b);%
        \draw[->] (b)%
        edge [bend right=25] (b2)%
        edge [bend left=25] (b2);%

        \node () at ($(g3) + (0,-3.8)$)%
        {$(\Glb_{2\le\iota<\omega} g_\iota)$};
        
        \node [node distance=2.5cm,right=of g3] (g4) {$f$};%
        \node [node distance=10mm,below of=g4] (b) {$f$};%
        \node [node distance=10mm,below of=b] (b2) {$f$};%
        \node [node distance=10mm,below of=b2] (b3) {$\bot$};%
        \draw[->] (g4)%
        edge [bend right=25] (b)%
        edge [bend left=25] (b);%
        \draw[->] (b)%
        edge [bend right=25] (b2)%
        edge [bend left=25] (b2);%
        \draw[->] (b2)%
        edge [bend right=25] (b3)%
        edge [bend left=25] (b3);%

        \node () at ($(g4) + (0,-3.8)$)%
        {$(\Glb_{3\le\iota<\omega} g_\iota)$};

        \coordinate [node distance=2.5cm,right=of g4] (g5);%
        \node () at ($(g5) + (0,-3.8)$)%
        {$\dots$};
      \end{tikzpicture}
    \end{center}
    The least upper bound of this sequence of term graphs and thus the
    limit inferior of $(g_\iota)_{\iota<\omega}$ is the term graph
    $g_\omega$ depicted in Figure~\ref{fig:doubleTransRed}.

  \item%
    \label{item:liminf3}%
    The situation is slightly different in the sequence
    $(g_\iota)_{\iota<\omega}$ from
    Figure~\ref{fig:doubleTransRedDiv}. Here we also have acyclic
    sharing, viz.\ in the $c$-node. However, unlike in the previous
    example from Figure~\ref{fig:doubleTransRed}, the acyclic sharing
    changes in each step. Hence, a lower bound of two distinct term
    graphs in $(g_\iota)_{\iota<\omega}$ cannot contain a $c$-node
    because a rigid $\bot$-homomorphism must map such a $c$-node to a
    $c$-node with the same acyclic sharing, i.e.\ the same acyclic
    positions. Consequently, the sequence of greatest lower bounds
    $(\Glb_{\beta\le\iota<\omega} g_\iota)_{\beta<\omega}$ looks as
    follows:
    
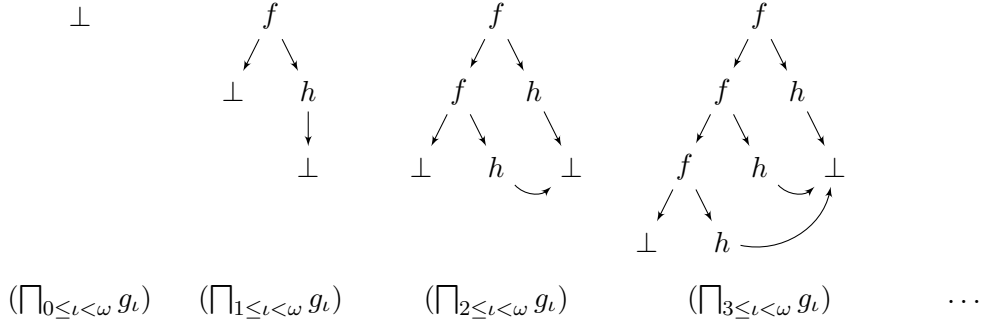
\begin{center}
      \begin{tikzpicture}
        
        \node (g1) {$\bot$};%

        \node () at ($(g1) + (0,-3.8)$)%
        {$(\Glb_{0\le\iota<\omega} g_\iota)$};
        
        \node [node distance=2cm,right=of g1] (g2) {$f$}%
        child {%
          node (h1) {$\bot$}%
        } child {%
          node (h2) {$h$}%
          child {%
            node (c) {$\bot$}%
          } };%

        \node () at ($(g2) + (0,-3.8)$)%
        {$(\Glb_{1\le\iota<\omega} g_\iota)$};
         
        \node [node distance=2.5cm,right=of g2] (g3) {$f$}%
        child {%
          node (f) {$f$}%
          child {%
            node (h1) {$\bot$}%
          } child {%
            node (h2) {$h$}%
          } } child {%
          node (h3) {$h$}%
          child [missing]%
          child {%
            node (c) {$\bot$}%
          }%
        };%

        \draw[->]%
        (h2) edge[bend right=45] (c);%

        \node () at ($(g3) + (0,-3.8)$)%
        {$(\Glb_{2\le\iota<\omega} g_\iota)$};
         
        \node [node distance=3cm,right=of g3] (go) {$f$}%
        child {%
          node (f1) {$f$}%
          child {%
            node (f2) {$f$}%
            child {%
              node {$\bot$}%
            } child {%
              node (h1) {$h$}%
            }%
          } child {%
            node (h2) {$h$}%
          } } child {%
          node (h3) {$h$}%
          child [missing]%
          child {%
            node (c) {$\bot$}%
          }%
        };%

        \draw[->]%
        (h1) edge[bend right=45] (c)%
        (h2) edge[bend right=45] (c);%

        \node () at ($(go) + (0,-3.8)$)%
        {$(\Glb_{3\le\iota<\omega} g_\iota)$};

        \coordinate [node distance=2.5cm,right=of g4] (g5);%
        \node () at ($(g5) + (0,-3.8)$)%
        {$\dots$};
      \end{tikzpicture}
    \end{center}
    We thus get the term graph $g_\omega$, depicted in
    Figure~\ref{fig:doubleTransRedDiv}, as the limit inferior of
    $(g_\iota)_{\iota<\omega}$. The $\bot$ labelling is necessary
    because of the change in acyclic sharing throughout the sequence.
  \end{enumerate}
\end{example}

\noindent
While we have confirmed in Corollary~\ref{cor:lebotrGeneralise} that
the partial order $\lebotr$ generalises the partial order $\lebot$ on
terms, we still have to show that this also carries over to the limit
inferior. We can derive this property from the following simple
lemma:
\begin{lemma}
  \label{lem:lebotrTrees}
  If $g \in \ipctgraphs$ and $t \in \ipterms$ with $g \lebotr t$, then
  $g \in \ipterms$.
\end{lemma}
\begin{proof}
  Since $t$ is a term tree, $\sim_t$ is an identity
  relation. According to Corollary~\ref{cor:chaTgraphPo}, $g \lebotr
  t$ implies that ${\sim_g} \subseteq {\sim_t}$. Hence, also $\sim_g$
  is an identity relation, which means that $g$ is a term tree as
  well.
\end{proof}

\begin{proposition}
  \label{prop:liminfGeneralise}
  Given a sequence $(t_\iota)_{\iota<\alpha}$ over $\ipterms$, the
  limit inferior of $(t_\iota)_{\iota<\alpha}$ in $(\ipterms,\lebot)$
  coincides with the limit inferior of $(t_\iota)_{\iota<\alpha}$ in
  $(\ipctgraphs,\lebotr)$.
\end{proposition}
\begin{proof}
  Since both structures are complete semilattices, both limit
  inferiors exist. For each $\beta < \alpha$, let $s_\beta$ be the glb
  of $T_\beta = \setcom{t_\iota}{\beta \le \iota < \alpha}$ in
  $(\ipterms,\lebot)$ and $g_\beta$ the glb of $T_\beta$ in
  $(\ipctgraphs,\lebotr)$. Since then $g_\beta \lebotr t_\beta$, we
  know by Lemma~\ref{lem:lebotrTrees} that $g_\beta$ is a term
  tree. By Corollary~\ref{cor:lebotrGeneralise}, this implies that
  $g_\beta$ is the glb of $T_\beta$ in $(\ipterms,\lebot)$ as well,
  which means that $g_\beta = s_\beta$.

  Let $t$ and $g$ be the limit inferior of $(t_\iota)_{\iota<\alpha}$
  in $(\ipterms,\lebot)$ and $(\ipctgraphs,\lebotr)$, respectively. By
  the above argument, we know that $t$ and $g$ are the lub of the set
  $S = \setcom{s_\beta}{\beta<\alpha}$ in $(\ipterms,\lebot)$
  respectively $(\ipctgraphs,\lebotr)$. By
  Corollary~\ref{cor:lebotrGeneralise}, $t$ is an upper bound of $S$
  in $(\ipctgraphs,\lebotr)$. Since $g$ is the least such upper bound,
  we know that $g \lebotr t$. According to
  Lemma~\ref{lem:lebotrTrees}, this implies that $g$ is a term
  tree. Hence, by Corollary~\ref{cor:lebotrGeneralise}, $g$ is an
  upper bound of $S$ in $(\ipterms,\lebot)$ and $g \lebot t$. Since
  $t$ is the least upper bound of $S$ in $(\ipterms,\lebot)$, we can
  conclude that $t = g$.
\end{proof}

\subsubsection{Maximal Term Graphs}
\label{sec:maximality}

Intuitively, partial term graphs represent partial results of
computations where $\bot$-nodes act as placeholders denoting the
uncertainty or ignorance of the actual ``value'' at that position. On
the other hand, total term graphs do contain all the information of a
result of a computation -- they have the maximally possible
information content. In other words, they are the maximal elements
w.r.t.\ $\lebotr$. The following proposition confirms this intuition.

\begin{proposition}[total term graphs are maximal]
  \label{prop:nonPartMax}
  % total term graphs are the maximal elements %
  Let $\Sigma$ be a non-empty signature. Then $\ictgraphs$ is the set
  of maximal elements in $(\ipctgraphs,\lebotr)$.
\end{proposition}
\begin{proof}
  At first we need to show that each element in $\ictgraphs$ is
  maximal. For this purpose, let $g \in \ictgraphs$ and $h \in
  \ipctgraphs$ such that $g \lebotr h$. We have to show that then $g =
  h$. Since $g \lebotr h$, there is a rigid $\bot$-homomorphism
  $\phi\fcolon g \homto_\bot h$. As $g$ does not contain any
  $\bot$-node, $\phi$ is even a rigid homomorphism. By
  Lemma~\ref{lem:strongD-homInj}, $\phi$ is injective and, therefore,
  according to Lemma~\ref{lem:isomBij}, an isomorphism. Hence, we
  obtain that $g \isom h$ and, consequently, using
  Proposition~\ref{prop:canon}, that $g = h$.

  Secondly, we need to show that $\ipctgraphs$ does not contain any
  other maximal elements besides those in $\ictgraphs$. Suppose there
  is a term graph $g \in \ipctgraphs \setminus \ictgraphs$ which is
  maximal in $\ipctgraphs$. Hence, there is a node $n^* \in N^g$ with
  $\glab^g(n^*) = \bot$. If $\Sigma$ contains a nullary symbol $c$,
  construct a term graph $h$ from $g$ by relabelling the node $n^*$
  from $\bot$ to $c$. However, then $g \lbotr h$, which contradicts
  the assumption that $g$ is maximal w.r.t.\ $\lebotr$.  Otherwise, if
  $\Sigma^{(0)} = \emptyset$, let $\ol n$ be a fresh node (i.e.\ $\ol
  n\nin N^g$) and $f$ some $k$-ary symbol in $\Sigma$. Define the term
  graph $h$ by
  \begin{align*}
    N^h &= N^g \uplus \set{\ol n} & r^h &= r^g \\
    \glab^h(n) &=
    \begin{cases}
      f & \text{if } n = n^* \\
      \bot &\text{if } n= \ol n \\
      \glab^g(n) &\text{otherwise}
    \end{cases}
     &
     \gsuc^h(n) &= 
    \begin{cases}
      \seq{\ol n, \dots , \ol n} & \text{if } n = n^* \\
      \epsilon &\text{if } n= \ol n \\
      \gsuc^g(n) &\text{otherwise}
    \end{cases}
  \end{align*}
  That is, $h$ is obtained from $g$ by relabelling $n^*$ with $f$
  and setting the $\bot$-labelled node $\ol n$ as the target of all
  outgoing edges of $n^*$.  We assume that $\ol n$ was chosen such
  that $h$ is canonical (i.e.\ $\ol n = \nodePos{h}{\ol
    n}$). Obviously, $g$ and $h$ are distinct. Define $\phi\fcolon N^g
  \funto N^h$ by $n \mapsto n$ for all $n\in N^g$. Clearly, $\phi$
  defines a rigid $\bot$-homomorphism from $g$ to $h$. Hence, $g
  \lebotr h$. This contradicts the assumption of $g$ being
  maximal. Consequently, no element in $\ipctgraphs \setminus
  \ictgraphs$ is maximal.
\end{proof}

Note that this property does not hold for the simple partial order
$\lebots$ that we have considered briefly in the beginning of this
section. Figure~\ref{fig:convWeird} shows the total term graph $g_0$,
which is strictly smaller than $g_1$ w.r.t.\ $\lebots$.

\section{A Rigid Metric on Term Graphs}
\label{sec:alternative-metric}

In this section, we pursue the metric approach to convergence in
rewriting systems. To this end, we shall define a metric space on
canonical term graphs. We base our approach to defining a metric
distance on the definition of the metric distance $\dd$ on terms. In
particular, we shall define a \emph{truncation} operation on term
graphs, which cuts off certain nodes depending on their depth in the
term graph. Subsequently, we study the interplay of the truncation
with $\Delta$-homomorphisms and the depth of nodes within a term
graph. Finally, we use the truncation operation to derive a metric on
term graphs.

\subsection{Truncating Term Graphs}
\label{sec:trunc-term-graphs}

Originally, Arnold and Nivat~\cite{arnold80fi} used a truncation of
terms to define the metric on terms. The truncation of a term $t$ at
depth $d\le\omega$, denoted $\trunc{t}{d}$, replaces all subterms at
depth $d$ by $\bot$:
\begin{align*}
  \trunc{t}{0} = \bot, \quad \trunc{f(t_1,\dots,t_k)}{d+1} =
  f(\trunc{t_1}{d},\dots,\trunc{t_k}{d}),\quad \trunc{t}{\omega} = t
\end{align*}

Recall that the metric distance $\dd$ on terms is defined by $\dd(s,t)
= 2^{-\similar{s}{t}}$. The underlying notion of similarity
$\similar{\cdot}{\cdot}$ can be characterised via truncations as
follows:
\[
\similar{s}{t} = \max\setcom{d\le\omega}{\trunc{s}{d} = \trunc{t}{d}}
\]

We adopt this approach for term graphs as well. To this end, we shall
define a rigid truncation on term graphs. In
Section~\ref{sec:deriving-metric-term} we will then show that this
truncation indeed yields a complete metric space.

\begin{definition} [rigid truncation of term graphs]
  \label{def:truncGraph}
  % rigid truncation of term graphs %
  Let $g \in \iptgraphs$ and $d < \omega$.
  \begin{enumerate}[(i)]
  \item Given $n,m\in N^g$, $m$ is an \emph{acyclic predecessor} of
    $n$ in $g$ if there is an acyclic position $\pi \concat \seq i \in
    \nodePosAcy{g}{n}$ with $\pi \in \nodePos{g}{m}$. The set of
    acyclic predecessors of $n$ in $g$ is denoted $\predAcy{g}{n}$.
  \item The set of \emph{retained nodes} of $g$ at $d$, denoted
    $\tNodes{g}{d}$, is the least subset $M$ of $N^g$ satisfying the
    following two conditions for all $n\in N^g$:
    \begin{center}
      \begin{inparaenum}[(T1)]
        \def\theenumi{T}
      \item $\depth{g}{n} < d  \implies  n \in
        M$  \label{eq:truncNodes1} \qquad
      \item $n \in M  \implies  \predAcy{g}{n} \subseteq M$
        \label{eq:truncNodes2}
      \end{inparaenum}
    \end{center}
  \item For each $n\in N^g$ and $i\in \nats$, we use $n^i$ to denote a
    fresh node, i.e.\ $\setcom{n^i}{n \in N^g, i\in \nats}$ is a set
    of pairwise distinct nodes not occurring in $N^g$.  The set of
    \emph{fringe nodes} of $g$ at $d$, denoted $\fNodes{g}{d}$, is
    defined as the singleton set $\set{r^g}$ if $d = 0$, and otherwise
    as the set
    \begin{gather*}
      \setcom{n^i}{
        \begin{aligned}
          n \in \tNodes{g}{d}, 0\le i < \rank{g}{n} \text{ with} \quad
          &\gsuc^g_i(n)\nin \tNodes{g}{d}\\\text{or }\quad
          &\depth{g}{n} \ge d - 1, n\nin \predAcy{g}{\gsuc^g_i(n)}
        \end{aligned}
}
    \end{gather*}
  \item The \emph{rigid truncation} of $g$ at $d$, denoted $\truncr{g}{d}$,
    is the term graph defined by
    \begin{align*}
      N^{\truncr{g}{d}} &= \tNodes{g}{d} \uplus \fNodes{g}{d}
      & r^{\truncr{g}{d}} &= r^g
      \\
      \glab^{\truncr{g}{d}}(n) &= 
      \begin{cases}
        \glab^g(n) &\text{if } n \in \tNodes{g}{d} \\
        \bot &\text{if } n \in \fNodes{g}{d}
      \end{cases} &
      \gsuc_i^{\truncr{g}{d}}(n) &= 
      \begin{cases}
        \gsuc_i^g(n) &\text{if } n^i \nin \fNodes{g}{d} \\
        n^i &\text{if } n^i \in \fNodes{g}{d}
      \end{cases}
    \end{align*}
    Additionally, we define $\truncr{g}{\omega}$ to be the term graph
    $g$ itself.
  \end{enumerate}
\end{definition}
\def\trna{(\ref{eq:truncNodes1})}
\def\trnb{(\ref{eq:truncNodes2})}

\begin{figure}
  \centering
  \subfloat[Cyclic vs.\ acyclic sharing.]{
    \begin{tikzpicture}[node distance=1.5cm, allow upside down]
      \node[node name=175:r, alias=r1] (f1) {$h$}
      child {
        node[node name=185:n] (f2) {$h$}
      };
      \draw (f2) edge[min distance=1cm,out=-45,in=0,->] (f1);
      
      \node[node name=5:r,alias=r2,right=of r1] {$h$}
      child {
        node[node name=5:n] {$h$}
        child {
          node[node name=5:n^0] {$\bot$}
        }
      };
    \node[alias=r3,node distance=1.8cm,right=of r2] (g) {$f$}
    child {
      node (f1) {$h$}
      child {
        node {$\vdots$}
        child {
          node (f2) {$h$}
          child {
            node (a) {$a$}
          }
        }
      }
    };
    \draw[->] (g) edge[bend left=35] (a);
    \draw[decorate,decoration=brace] (f2.south west) -- (f1.north
    west) node[midway,above,sloped] {$n$ times};
      \begin{scope}[node distance=4cm]
        \node[below=of r1] {$(g)$};
        \node[below=of r2] {$(\truncr{g}{2})$};
      \node[below=of r3] {$(g_n = \truncr{g_n}{2})$};
      \end{scope}
    \end{tikzpicture}
  \label{fig:loopTrunc}}
  \hspace{1cm}
  \subfloat[Comparison with simple truncation.]{
  \begin{tikzpicture}[->,node distance=2cm]
    \node[alias=r1] (f) {$f$}
    child [missing]
    child {
      node (g1) {$h$}
      child {
        node (g2) {$h$}
        child {
          node (g) {$h$}
          child {
            node {$a$}
          }
        }
      }
    };
    \draw (f) edge[out=-115,in=115] (g);
    \begin{pgfonlayer}{background}
      \fill[lightgray,rounded corners=.5cm] ($(f.north)+(0,.6)$) --
      ($(g1.east)+(.3,0)$) -- ($(g2.west)-(.3,0)$) -- ($(g.south
      east)+(.3,0)$) --($(g.south west)+(-.2,-.4)$) -- ($(f)+(-1,-2)$) --
      cycle;

      \draw[dashed ,rounded corners=.5cm] ($(f.north)+(0,.6)$) --
      ($(g1.east)+(.3,0)$)  -- ($(g.south
      east)+(.3,0)$) --($(g.south west)+(-.2,-.4)$) -- ($(f)+(-1,-2)$) --
      cycle;
    \end{pgfonlayer}

    \node[alias=r2,right=of f] (f) {$f$}
    child [missing]
    child {
      node (g1) {$h$}
      child {
        node (g2) {$\bot$}
        child[edge from parent/.style={}] {
          node (g) {$h$}
          child[edge from parent/.style={draw}] {
            node {$\bot$}
          }
        }
      }
    };
    \draw (f) edge[out=-115,in=115] (g);
    \begin{pgfonlayer}{background}
      \fill[lightgray,rounded corners=.5cm] ($(f.north)+(0,.6)$) --
      ($(g1.east)+(.3,0)$) -- ($(g2.west)-(.3,0)$) -- ($(g.south
      east)+(.3,0)$) --($(g.south west)+(-.2,-.4)$) -- ($(f)+(-1,-2)$) --
      cycle;
    \end{pgfonlayer}

    \node[alias=r3,right=of f] (f) {$f$}
    child [missing]
    child {
      node (g1) {$h$}
      child {
        node (g2) {$h$}
        child {
          node (g) {$h$}
          child {
            node {$\bot$}
          }
        }
      }
    };
    \draw (f) edge[out=-115,in=115] (g);
    \begin{pgfonlayer}{background}
      \draw[dashed ,rounded corners=.5cm] ($(f.north)+(0,.6)$) --
      ($(g1.east)+(.3,0)$)  -- ($(g.south
      east)+(.3,0)$) --($(g.south west)+(-.2,-.4)$) -- ($(f)+(-1,-2)$) --
      cycle;
    \end{pgfonlayer}
    \begin{scope}[node distance=4cm]
      \node[below=of r1] {($g$)};
      \node[below=of r2] {($\truncs{g}{2}$)};
      \node[below=of r3] {($\truncr{g}{2}$)};
    \end{scope}
  \end{tikzpicture}
  \label{fig:exTrunc}
}
  \caption{Examples of truncations.}
\end{figure}
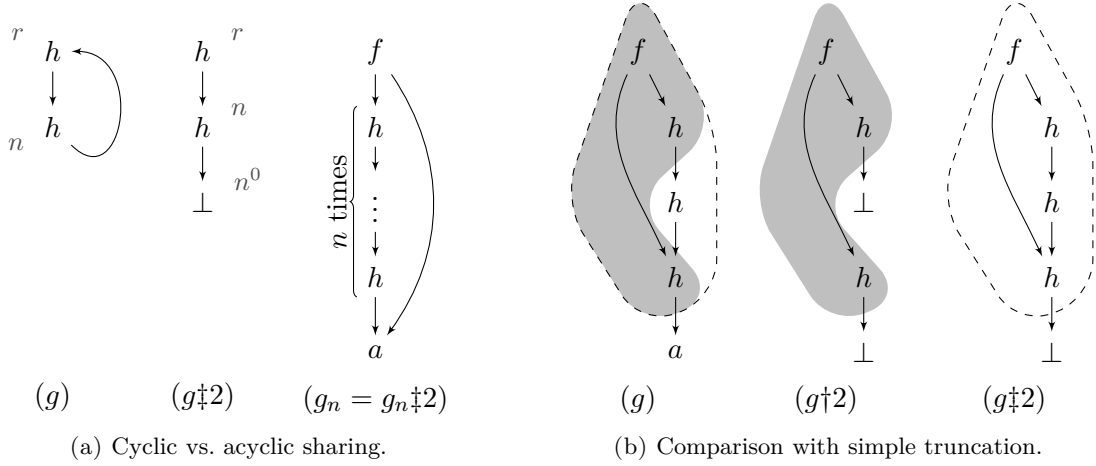

\noindent
Before discussing the intuition behind this definition of rigid
truncation, let us have a look at the r\^ole of retained and fringe
nodes: the set of retained nodes $\tNodes{g}{d}$ contains the nodes
that are \emph{preserved} by the rigid truncation. All other nodes in
$N^g\setminus \tNodes{g}{d}$ are cut off. The ``holes'' that are thus
created are filled by the fringe nodes in $\fNodes{g}{d}$. This is
expressed in the condition
$\gsuc^g_i(n)\nin \tNodes{g}{d}$
which, if satisfied, yields a fringe node $n^i$. That is, a fresh
fringe node is inserted for each successor of a retained node that is
not a retained node itself. As fringe nodes function as a replacement
for cut-off sub-term graphs, they are labelled with $\bot$ and have no
successors.

But there is another circumstance that can give rise to a fringe node:
if $\depth{g}{n} \ge d - 1$ and $n\nin \predAcy{g}{\gsuc^g_i(n)}$, we
also get a fringe node $n^i$. This condition is satisfied whenever an
outgoing edge from a retained node closes a cycle. The lower bound for
the depth is chosen such that a successor node of $n$ is not
necessarily a retained node. An example is depicted in
Figure~\ref{fig:loopTrunc}. For depth $d = 2$, the node $n$ in the
term graph $g$ is just above the fringe, i.e.\ satisfies $\depth{g}{n}
\ge d - 1$. Moreover, it has an edge to the node $r$ that closes a
cycle. Hence, the rigid truncation $\truncr{g}{2}$ contains the fringe
node $n^0$ which is now the $0$-th successor of $n$.

We chose this admittedly complicated notion of truncation in order to
make it compatible with the partial order $\lebotr$: first of all, the
rigid truncation of a term graph is supposed to yield a smaller term
graph w.r.t.\ the rigid partial order $\lebotr$, i.e.\ $\truncr{g}{d}
\lebotr g$. Hence, whenever a node is kept as a retained node, also
its acyclic positions have to be kept in order to preserve its upward
structure. To achieve this, with each node also its \emph{acyclic
  ancestors} have to be retained. The closure condition \trnb{} is
enforced exactly for this purpose.

To see what this means, consider Figure~\ref{fig:exTrunc}. It shows a
term graph $g$ and its truncation at depth $2$, once without the
closure condition \trnb{}, denoted $\truncs{g}{2}$, and once including
\trnb{}, denoted $\truncr{g}{2}$. The grey area highlights the nodes
that are at depth smaller than $2$, i.e.\ the nodes contained in
$\tNodes{g}{2}$ due to \trna{}. The nodes within the area surrounded
by a dashed line are all the nodes in $\tNodes{g}{2}$. One can observe
that with the \emph{simple truncation} $\truncs{g}{d}$ without
\trnb{}, we do not have $\truncs{g}{2} \lebotr g$.  The reason in this
particular example is the bottommost $h$-node whose acyclic sharing in
$g$ differs from that in the simple truncation $\truncs{g}{2}$ as one
of its predecessors was removed due to the truncation. This effect is
avoided in our definition of rigid truncation, which always includes
all acyclic predecessors of a node.

Nevertheless, the simple truncation $\truncs{g}{d}$ has its
benefits. It is much easier to work with and provides a natural
counterpart for the simple partial order $\lebots$ \cite{bahr12rta}.

The following lemma confirms that we were indeed successful in making
the truncation of term graphs compatible with the rigid partial order
$\lebotr$:
\begin{lemma}[rigid truncation is smaller]
  \label{lem:truncLebot}
  % truncation yields a smaller term graph %
  Given $g \in \iptgraphs$ and $d\le\omega$, we have that $\truncr{g}{d} \lebotr g$.
\end{lemma}
\begin{proof}
  The cases $d = \omega$ and $d = 0$ are trivial. Assume $0 <d <
  \omega$ and define the function $\phi$ as follows:
  \begin{align*}
    \phi\fcolon N^{\truncr{g}{d}} &\funto N^g \\
    \tNodes{g}{d} \ni n &\mapsto n \\
    \fNodes{g}{d} \ni n^i &\mapsto \gsuc^g_i(n)
  \end{align*}
  We will show that $\phi$ is a rigid $\bot$-homomorphism from
  $\truncr{g}{d}$ to $g$ and, thereby, $\truncr{g}{d} \lebotr g$.

  Since $r^{\truncr{g}{d}} = r^g$ and $r^{\truncr{g}{d}} \in
  \tNodes{g}{d}$, we have $\phi(r^{\truncr{g}{d}}) = r^g$ and,
  therefore, the root condition. Note that all nodes in
  $\fNodes{g}{d}$ are labelled with $\bot$ in $\truncr{g}{d}$, i.e.\
  all non-$\bot$-nodes are in $\tNodes{g}{d}$. Thus, the labelling
  condition is trivially satisfied as for all $n \in \tNodes{g}{d}$ we
  have
  \[
  \glab^{\truncr{g}{d}}(n) = \glab^g(n) = \glab^g(\phi(n)).
  \]
  For the successor condition, let $n \in \tNodes{g}{d}$. If $n^i \in
  \fNodes{g}{d}$, then $\gsuc^{\truncr{g}{d}}_i(n) = n^i$. Hence, we have
  \[
  \phi(\gsuc^{\truncr{g}{d}}_i(n)) = \phi(n^i) = \gsuc^g_i(n) =
  \gsuc^g_i(\phi(n)).
  \]
  If, on the other hand, $n^i \nin \fNodes{g}{d}$, then
  $\gsuc^{\truncr{g}{d}}_i(n) = \gsuc^g_i(n) \in \tNodes{g}{d}$. Hence,
  we have
  \[
  \phi(\gsuc^{\truncr{g}{d}}_i(n)) = \phi(\gsuc^g_i(n)) = \gsuc^g_i(n)
  = \gsuc^g_i(\phi(n)).
  \]

  This shows that $\phi$ is a $\bot$-homomorphism. In order to prove
  that $\phi$ is rigid, we will show that $\nodePosAcy{g}{\phi(n)}
  \subseteq \nodePos{\truncr{g}{d}}{n}$ for all $n \in \tNodes{g}{d}$,
  which is sufficient according to Lemma~\ref{lem:presShar}. Note that
  we can replace $\phi(n)$ by $n$ since $n \in
  \tNodes{g}{d}$. Therefore, we can show this statement by proving
  \[
  \forall \pi \in \nats^* \forall n \in \tNodes{g}{d}.\; (\pi
  \in \nodePosAcy{g}{n} \implies \pi \in \nodePos{\truncr{g}{d}}{n})
  \]
  by induction on the length of $\pi$. If $\pi = \emptyseq$, then $n =
  r^g$ and, therefore, $\pi \in \nodePos{\truncr{g}{d}}{n}$. If $\pi =
  \pi'\concat \seq i$, let $m = \nodeAtPos{g}{\pi'}$. Then we have $m
  \in \predAcy{g}{n}$ and, therefore, $m \in \tNodes{g}{d}$ by the
  closure property \trnb. And since $\pi' \in \nodePosAcy{g}{m}$, we
  can apply the induction hypothesis to obtain that $\pi' \in
  \nodePos{\truncr{g}{d}}{m}$. Moreover, because $\gsuc^g_i(m) = n$,
  this implies that $m^i \nin \fNodes{g}{d}$. Thus,
  $\gsuc^{\truncr{g}{d}}_i(m) = n$ and, therefore, $\pi' \concat \seq i \in
  \nodePos{\truncr{g}{d}}{n}$.
\end{proof}

Also note that the rigid truncation on term graphs generalises Arnold
and Nivat's~\cite{arnold80fi} truncation on terms.
\begin{proposition}
  \label{prop:truncTruncr}
  For each $t\in \ipterms$ and $d\le\omega$, we have that
  $\truncr{t}{d} \isom \trunc{t}{d}$.
\end{proposition}
\begin{proof}
  For the case that $d \in \set{0,\omega}$, the equation
  $\truncr{t}{d} = \trunc{t}{d}$ holds trivially. For the other cases,
  we can easily see that $\trunc{t}{d}$ is obtained from $t$ by
  replacing all subterms at depth $d$ by $\bot$. On the other hand,
  since in a term tree each node has at most one (acyclic)
  predecessor, which has a strictly smaller depth, we know that the
  set of retained nodes $\tNodes{t}{d}$ is the set of nodes of depth
  smaller than $d$ and the set of fringe nodes $\fNodes{t}{d}$ is the
  set $\setcom{n^i}{n \in N^t, \depth{t}{\gsuc^t_i(n)}=d}$. Hence,
  $\truncr{t}{d}$ is obtained from $t$ by replacing each node at depth
  $d$ with a fresh node labelled $\bot$. We can thus conclude that
  $\truncr{t}{d} \isom \trunc{t}{d}$.
\end{proof}
Consequently, if we use the rigid truncation to define a metric on
term graphs analogously to Arnold and Nivat, we obtain a metric on
term graphs that generalises the metric $\dd$ on terms.

\subsection{The Effect of Truncation}
\label{sec:effect-truncation}

In order to characterise the effect of a truncation to a term graph,
we need to associate an appropriate notion of depth to a whole term
graph:

\begin{definition}[symbol/graph depth]
  Let $g \in \itgraphs$ and $\Delta\subseteq \Sigma$.
  \begin{enumerate}[(i)]
  \item The \emph{depth} of $g$, denoted $\gdepth{g}$, is the least
    upper bound of the depth of nodes in $g$, i.e.\
    \[
    \gdepth{g} = \Lub\setcom{\depth{g}{n}}{n \in N^g}.
    \]
  \item The \emph{$\Delta$-depth} of $g$, denoted
    $\sdepth{g}{\Delta}$, is the minimum depth of nodes in $g$
    labelled in $\Delta$, i.e.\
    \[
    \sdepth{g}\Delta = \min\setcom{\depth{g}{n}}{n \in
      N^g,\glab^g(n)\in\Delta}\cup\set\omega.
    \]
    If $\Delta$ is a singleton set $\set{\sigma}$, we also write
    $\sdepth{g}\sigma$ instead of $\sdepth{g}{\set\sigma}$.
  \end{enumerate}
\end{definition}

\noindent
Notice the difference between depth and $\Delta$-depth. The former is
the least upper bound of the depth of nodes in a term graph whereas
the latter is the \emph{minimum} depth of nodes labelled by a symbol
in $\Delta$. Thus, we have that $\gdepth{g} = \omega$ iff $g$ is
infinite; and $\sdepth{g}{\Delta} = \omega$ iff $g$ does not contain a
$\Delta$-node.

In the following, we will prove a number of lemmas that show how
$\Delta$-homomorphisms preserve the depth of nodes in term
graphs. Understanding how $\Delta$-homomorphisms affect the depth of
nodes will become important for relating the rigid truncation to the
rigid partial order $\lebotr$.

\begin{lemma}[reverse depth preservation of $\Delta$-homomorphisms]
  \label{lem:homDepthRev}
  % reverse depth preservation of Delta-homomorphisms %
  Let $g,h \in \itgraphs$ and $\phi\fcolon g \homto_\Delta h$. Then,
  for all $n \in N^h$ with $\depth{h}{n} \le \sdepth{g}{\Delta}$,
  there is a node $m \in \phi^{-1}(n)$ with $\depth{g}{m} \le
  \depth{h}{n}$.
\end{lemma}
\proof
  We prove the statement by induction on $\depth{h}{n}$. If
  $\depth{h}{n} = 0$, then $n = r^h$. With $m = r^g$, we have $\phi(m)
  = n$ and $\depth{g}{m} = 0$. If $\depth{h}{n} > 0$, then there is
  some $n' \in N^h$ with $\gsuc^h_i(n') = n$ and $\depth{h}{n'} <
  \depth{h}{n}$. Hence, we can employ the induction hypothesis to
  obtain some $m' \in N^g$ with $\depth{g}{m'} \le \depth{h}{n'}$ and
  $\phi(m') = n'$. Since $\depth{g}{m'} \le \depth{h}{n'} <
  \depth{h}{n} \le \sdepth{g}{\Delta}$, we have $\glab^g(m') \nin
  \Delta$. Hence, $\phi$ is homomorphic in $m'$. For $m =
  \gsuc^g_i(m')$, we can then reason as follows:
  \begin{align*}
    \phi(m) &= \phi(\gsuc^g_i(m')) = \gsuc^h_i(\phi(m')) =
    \gsuc^h_i(n') = n,\quad \text{and} \\
    \depth{g}{m} &\le \depth{g}{m'} + 1 \le \depth{h}{n}.
    \rlap{\hbox to 167 pt{\hfill\qEd}}
  \end{align*}

\begin{lemma}[$\Delta$-depth preservation of $\Delta$-homomorphisms]
  \label{lem:DhomDeltaDepth}
  % Delta-depth preservation %
  Let $g,h \in \itgraphs$ and $\phi\fcolon g \homto_\Delta h$, then
  $\sdepth{g}{\Delta} \le \sdepth{h}{\Delta}$.
\end{lemma}
\begin{proof}
  Let $n \in N^h$ with $\depth{h}{n} < \sdepth{g}{\Delta}$. To prove
  the lemma, we have to show that $\glab^h(n) \nin \Delta$. According
  to Lemma~\ref{lem:homDepthRev}, we find a node $m \in N^g$ with
  $\depth{g}{m} \le \depth{h}{n} < \sdepth{g}{\Delta}$ and $\phi(m) =
  n$. Since then $\glab^g(m) \nin \Delta$, we also have $\glab^h(n)
  \nin \Delta$ by the labelling condition for $\phi$.
\end{proof}

For rigid $\Delta$-homomorphisms, we even have a stronger form of
depth preservation.
\begin{lemma}[depth preservation of rigid $\Delta$-homomorphisms]
  \label{lem:strongDepthPres}
  % depth preservation of rigid Delta-homomorphisms %
  Let $g, h \in \itgraphs$ and $\phi\fcolon g \homto_\Delta h$ a rigid
  $\Delta$-homomorphism. Then $\depth{g}{n} = \depth{h}{\phi(n)}$ for
  all $n \in N^g$ with $\glab^g(n) \nin \Delta$.
\end{lemma}
\begin{proof}
  If $\glab^g(n) \nin \Delta$, then $\nodePosAcy{g}{n} =
  \nodePosAcy{h}{\phi(n)}$. Hence, $\depth{g}{n} = \depth{h}{\phi(n)}$
  follows since a shortest position of a node must be acyclic.
\end{proof}

The gaps that are caused by a truncation due to the removal of nodes
are filled by fresh $\bot$-nodes. The following lemma provides a lower
bound for the depth of the introduced $\bot$-nodes.

\begin{lemma}[$\bot$-depth in rigid truncations]
  \label{lem:truncDepth}
  % $\bot$-depth in truncated term graphs %
  For all $g \in \itgraphs$ and $d < \omega$, we have that
  \begin{enumerate}[\em(i)]
  \item $\sdepth{\truncr{g}{d}}{\bot} \ge d$, and
    \label{item:truncDepth1}
  \item if $d > \gdepth{g}+1$, then $\truncr{g}{d} = g$, i.e.\
    $\sdepth{\truncr{g}{d}}{\bot} = \omega$.
    \label{item:truncDepth2}
  \end{enumerate}
\end{lemma}
\begin{proof}
  \def\itema{(\ref{item:truncDepth1})}%
  \def\itemb{(\ref{item:truncDepth2})}%
  \itema{} From the proof of Lemma~\ref{lem:truncLebot}, we obtain a
  rigid $\bot$-homomorphism $\phi\fcolon \truncr{g}{d} \homto_\bot
  g$. Note that the only $\bot$-nodes in $\truncr{g}{d}$ are those in
  $\fNodes{g}{d}$. Each of these nodes has only a single predecessor,
  a node $n \in \tNodes{g}{d}$ with $\depth{g}{n} \ge d - 1$. By
  Lemma~\ref{lem:strongDepthPres}, we also have
  $\depth{\truncr{g}{d}}{n} \ge d - 1$ for these nodes since $\phi$ is
  rigid, $n$ is not labelled with $\bot$ and $\phi(n) = n$. Hence, we
  have $\depth{\truncr{g}{d}}{m} \ge d$ for each node $m \in
  \fNodes{g}{d}$. Consequently, it holds that
  $\sdepth{\truncr{g}{d}}{\bot} \ge d$.

  \itemb{} Note that if $d > \gdepth{g}+1$, then $\tNodes{g}{d} = N^g$
  and $\fNodes{g}{d} = \emptyset$. Hence, $\truncr{g}{d} = g$.
\end{proof}

\begin{rem}
  Note that the precondition for the statement of
  clause~(\ref{item:truncDepth2}) in the lemma above reads $d >
  \gdepth{g}+1$ rather than $d > \gdepth{g}$ as one might expect. The
  reason for this is that a truncation might cut off an edge that
  emanates from a node at depth $d-1$ and closes a cycle. For an
  example of this phenomenon, take a look at
  Figure~\ref{fig:loopTrunc}. It shows a term graph $g$ of depth $1$
  and its rigid truncation at depth $2$. Even though there is no node
  at depth $2$ the truncation introduces a $\bot$-node.

  On the other hand, although a term graph has depth greater than $d$,
  the truncation at depth $d$ might still preserve the whole term
  graph. An example for this behaviour is the family of term graphs
  $(g_n)_{n<\omega}$ depicted in Figure~\ref{fig:loopTrunc}. Each of
  the term graphs $g_n$ has depth $n + 1$. Yet, the truncation at
  depth $2$ preserves the whole term graph $g_n$ for each $n >
  0$. Even though there might be $h$-nodes which are at depth $\ge 2$
  these nodes are directly or indirectly acyclic predecessors of the
  $a$-node and are, thus, included in $\tNodes{g_n}{2}$.
\end{rem}

Intuitively, the following lemma states that a rigid
$\bot$-homomorphism has the properties of an isomorphism up to the
depth of the shallowest $\bot$-node:
\begin{lemma}[$\lebotr$ and rigid truncation]
  \label{lem:lebotTrunc}
  % $\lebotr$ and truncation %
  Given $g, h \in \iptgraphs$ and $d < \omega$ with $g \lebotr h$ and
  $\sdepth{g}{\bot} \ge d$, we have that $\truncr{g}{d} \isom
  \truncr{h}{d}$.
\end{lemma}

The proof of the above lemma is based on a generalisation of
Lemma~\ref{lem:strongDepthPres}, which states that rigid
$\bot$-homomorphisms map non-$\bot$-nodes to nodes of the same
depth. However, since the rigid truncation of a term graph does not
only depend on the depth of nodes but also the acyclic sharing in the
term graph, we cannot rely on this statement on the depth of nodes
alone. The two key components of the proof of
Lemma~\ref{lem:lebotTrunc} are
\begin{inparaenum}
\item the property of rigid $\bot$-homomorphisms to map retained nodes
  of the source term graph exactly to the retained nodes of the target
  term graph and
\item that in the same way fringe nodes are exactly mapped to fringe
  nodes.
\end{inparaenum}
Showing the isomorphism between $\truncr{g}{d}$ and $\truncr{h}{d}$
can thus be reduced to the injectivity on retained nodes in
$\truncr{g}{d}$ which is obtained from the rigid $\bot$-homomorphism
from $g$ to $h$ by applying Lemma~\ref{lem:strongD-homInj}. The full
proof of Lemma~\ref{lem:lebotTrunc} is given in
Appendix~\ref{sec:proof-lemma-refl}.

We can use the above findings in order to obtain the following
properties of truncations that one would intuitively expect from a
truncation operation:
\begin{lemma}[smaller truncations]
  \label{lem:smallerTrunc}
  % smaller truncations %
  For all $g, h \in \itgraphs$ and $e \le d \le \omega$, the following
  holds:
  \begin{center}
    \begin{inparaenum}[\em(i)]
    \item $\truncr{g}{e} \isom \truncr{(\truncr{g}{d})}{e}$\,, \quad and\qquad
      \label{item:smallerTrunc1}
    \item $\truncr{g}{d} \isom \truncr{h}{d} \quad \implies \quad
      \truncr{g}{e} \isom \truncr{h}{e}$.
      \label{item:smallerTrunc2}
    \end{inparaenum}
  \end{center}
\end{lemma}
\proof
  \def\itema{(\ref{item:smallerTrunc1})}%
  \def\itemb{(\ref{item:smallerTrunc2})}%
  \itema{} For $d = \omega$, this is trivial. Suppose $d <
  \omega$. From Lemma~\ref{lem:truncLebot}, we obtain $\truncr{g}{d}
  \lebotr g$. Moreover, by Lemma~\ref{lem:truncDepth}, we have
  $\sdepth{\truncr{g}{d}}{\bot} \ge d$ and, a fortiori,
  $\sdepth{\truncr{g}{d}}{\bot} \ge e$. Hence, we can employ
  Lemma~\ref{lem:lebotTrunc} to get $\truncr{(\truncr{g}{d})}{e} \isom
  \truncr{g}{e}$.

  \itemb{} Since $\truncr{g}{d} \isom \truncr{h}{d}$, we also have
  $\truncr{(\truncr{g}{d})}{e} \isom \truncr{(\truncr{h}{d})}{e}$, as the
  construction of the truncation only depends on the structure of the
  term graphs. Hence, using \itema{} we can conclude
  \[
  \truncr{g}{e} \isom \truncr{(\truncr{g}{d})}{e} \isom
  \truncr{(\truncr{h}{d})}{e} \isom \truncr{h}{e}.\eqno{\qEd}
  \]

\subsection{Deriving a Metric on Term Graphs}
\label{sec:deriving-metric-term}

We may now define a rigid distance measure on canonical term graphs in
the style of Arnold and Nivat:
\begin{definition}[rigid distance]
  The \emph{rigid similarity} of two term graphs $g,h \in \ictgraphs$,
  written $\similarr{g}{h}$, is the maximum depth at which the rigid
  truncation of both term graphs coincide:
  \[
  \similarr{g}{h} = \max\setcom{d\le\omega}{\truncr{g}{d} \isom
    \truncr{h}{d}}.
  \]
  The \emph{rigid distance} between two term graphs $g,h \in
  \ictgraphs$, written $\ddr(g,h)$ is defined as
  \[
  \ddr(g,h) = 2^{-\similarr{g}{h}}, \text{ where we interpret
    $2^{-\omega}$ as $0$}.
  \]
\end{definition}
\smallskip\noindent
Indeed, the resulting distance forms an ultrametric on the set of
canonical term graphs:
\begin{proposition}[rigid ultrametric]
  The pair $(\ictgraphs,\ddr)$ forms an ultrametric space.
\end{proposition}
\begin{proof}
  The identity condition is derived as follows:
  \[
  \ddr(g,h) = 0 \iff \similarr{g}{h} = \omega \iff g \isom h
  \stackrel{\text{Prop.~\ref{prop:canon}}}\iff g = h
  \]
  The symmetry condition is satisfied by the following equational
  reasoning:
  \[
  \ddr(g,h) = 2^{-\similarr{g}{h}} = 2^{-\similarr{h}{g}} = \ddr(h,g)
  \]
  For the strong triangle condition, we have to show that
  \[
  \ddr(g_1,g_3) \le \max\set{\ddr(g_1,g_2), \ddr(g_2, g_3)},
  \]
  which is equivalent to
  \[
  \similarr{g_1}{g_3}  \ge \min\set{\similarr{g_1}{g_2}, \similarr{g_2}{g_3}}.
  \]
  Let $d = \similarr{g_1}{g_2}$ and $e = \similarr{g_2}{g_3}$. By
  symmetry, we can assume w.l.o.g.\ that $d \le e$, i.e.\ $d =
  \min\set{\similarr{g_1}{g_2}, \similarr{g_2}{g_3}}$. By definition
  of rigid similarity, we have both $\truncr{g_1}{d} \isom
  \truncr{g_2}{d}$ and $\truncr{g_2}{e} \isom \truncr{g_3}{e}$. From
  the latter we obtain, by Lemma~\ref{lem:smallerTrunc}, that
  $\truncr{g_2}{d} \isom \truncr{g_3}{d}$. That is, $\truncr{g_1}{d}
  \isom \truncr{g_2}{d} \isom \truncr{g_3}{d}$ which means that
  $\similarr{g_1}{g_3} \ge d$.
\end{proof}

\begin{example}
  \label{ex:convdiv}
  Figures~\ref{fig:gtransRed} and \ref{fig:doubleTransRed} on pages
  \pageref{fig:gtransRed} and \pageref{fig:doubleTransRed},
  respectively, show two sequences of term graphs that are converging
  in the metric space $(\ictgraphs,\ddr)$. In the sequence
  $(h_i)_{i<\omega}$ from Figure~\ref{fig:gtransRed}, we have that the
  rigid truncation at $0$ is trivially $\bot$ for all term graphs in
  the sequence. From $h_1$ onwards, the rigid truncation at $1$ is the
  term tree $\bot\cons\bot$; from $h_2$ onwards, the rigid truncation
  at $2$ is the term tree $b\cons\bot\cons\bot$; etc. Hence, for each
  $n<\omega$, the metric distance $\ddr(h_i,h_j)$ between two term
  graphs from $h_n$ onwards, i.e.\ with $n \le i, j < \omega$, is at
  most $2^{-n}$. That is, the sequence $(h_i)_{i<\omega}$ is
  Cauchy. Even more, for the term tree $h_\omega=b\cons b\cons b\cons
  \dots$ depicted in Figure~\ref{fig:gtransRed} we also have that
  $\truncr{h_\omega}{0} = \bot$, $\truncr{h_\omega}{1} =
  \bot\cons\bot$, $\truncr{h_\omega}{2} = b\cons\bot\cons\bot$,
  etc. Hence, for each $n<\omega$, the metric distance
  $\ddr(h_n,h_\omega)$ is at most $2^{-n}$. That is, the sequence
  $(h_i)_{i<\omega}$ converges to $h_\omega$. In a similar fashion,
  the sequence depicted in Figure~\ref{fig:doubleTransRed} converges
  as well.

  Figure~\ref{fig:doubleTransRedDiv} shows a sequence
  $(g_i)_{i<\omega}$ of term graphs that does not converge. In
  fact, it is not even Cauchy. To see this, notice that the $c$-node
  is at depth $1$ in $g_0$ and at depth $2$ from $g_1$ onwards. As in
  each term graph $g_i$ the $c$-node is reachable from any node in
  $g_i$ without forming a cycle, we have that each node is an
  acyclic ancestor of the $c$-node. That is, whenever the $c$-node is
  retained by a rigid truncation, so is any other node. Consequently,
  we have that $\truncr{g_i}{d} = g_i$ for each $i<\omega$
  and $d > 2$. Hence, the metric distance $\ddr(g_i,g_j)$
  between each pair of term graphs with $i\neq j$ is at least
  $2^{-2}$. That is, $(g_i)_{i<\omega}$ is not Cauchy.
\end{example}

Since we defined the metric on term graphs in the same style as Arnold
and Nivat~\cite{arnold80fi} defined the partial order $\dd$ on terms,
we can use the correspondence between the rigid truncation and the
truncation on terms in order to derive that the metric $\ddr$
generalises the metric $\dd$ on terms.
\begin{corollary}
  \label{cor:ddrGeneralise}
  For all $s,t\in \iterms$, we have that $\ddr(s,t) = \dd(s,t)$.
\end{corollary}
\begin{proof}
  Follows from Proposition~\ref{prop:truncTruncr}.
\end{proof}

From the above observation, we obtain that convergence in the metric
space $(\ictgraphs,\ddr)$ is a conservative extension of convergence
in the metric space $(\iterms,\dd)$:
\begin{proposition}
  \label{prop:limGeneralise}
  Every non-empty sequence over $\iterms$ converges to $t$ in
  $(\ictgraphs,\ddr)$ iff it converges to $t$ in $(\iterms,\dd)$.
\end{proposition}
\begin{proof}
  The ``if'' direction follows immediately from
  Corollary~\ref{cor:ddrGeneralise}. 

  For the ``iff'' direction we assume a sequence $S$ over $\iterms$
  that converges to $t$ in $(\ictgraphs,\ddr)$. Consequently, $S$ is
  also Cauchy in $(\ictgraphs,\ddr)$. Due to
  Corollary~\ref{cor:ddrGeneralise}, $S$ is then also Cauchy in
  $(\iterms,\dd)$. Since $(\iterms,\dd)$ is complete, $S$ converges to
  some term $t'$ in $(\iterms,\dd)$. Using the ``if'' direction of
  this proposition, we then obtain that $S$ converges to $t'$ in
  $(\ictgraphs,\ddr)$. Since limits are unique in metric spaces, we
  can conclude that $t = t'$.
\end{proof}

\section{Metric vs. Partial Order Convergence}
\label{sec:metric-vs.-partial}

In this section we study both the partially ordered set
$(\ipctgraphs,\lebotr)$ and the metric space $(\ictgraphs,\ddr)$. In
particular, we are interested in the notion of convergence that each
of the two structures provides.  We shall show that on total term
graphs -- i.e.\ in $\ictgraphs$ -- both structures yield the same
notion of convergence. That is, we obtain the same correspondence that
we already know from infinitary term rewriting as stated in
Theorem~\ref{thr:strongExt}. Moreover, as a side product, this finding
will also show the completeness of the metric space
$(\ictgraphs,\ddr)$.

The cornerstone of this comparison of the rigid metric $\ddr$ and the
rigid partial order $\lebotr$ is the following characterisation of the
rigid similarity $\similarr{\cdot}{\cdot}$ in terms of greatest lower
bounds in $(\ipctgraphs,\lebotr)$:
\begin{proposition}[characterisation of rigid similarity]
  \label{prop:similarGlb}
  Let $g,h \in \ictgraphs$ and $g \glb h$ the greatest lower bound of
  $g$ and $h$ in $(\ipctgraphs,\lebotr)$. Then $\similarr{g}{h} =
  \sdepth{g \glb h}{\bot}$.%
\end{proposition}
\begin{proof}
  At first assume that $g = h$. Hence, $g \glb h = g$ and,
  consequently $\sdepth{g \glb h}{\bot} = \omega$ as $g$ does not
  contain any node labelled $\bot$. On the other hand, $g = h$ implies
  $\truncr{g}{\omega} \isom \truncr{h}{\omega}$, and, therefore,
  $\similarr{g}{h} = \omega$. If $g \neq h$, then $g \nisom h$ by
  Proposition~\ref{prop:canon}. Hence, $\similarr{g}{h} <
  \omega$. Moreover, since $g \nisom h$, we know that $g \glb h$ has
  to be strictly smaller than $g$ or $h$ w.r.t.\ $\lebotr$. Hence,
  according to Proposition~\ref{prop:nonPartMax}, $g \glb h$ has to
  contain some node labelled $\bot$, i.e.\ $\sdepth{g \glb h}{\bot} <
  \omega$ as well. We prove that $\sdepth{g \glb h}{\bot} =
  \similarr{g}{h}$ by showing that both $\sdepth{g \glb h}{\bot} \le
  \similarr{g}{h}$ and $\sdepth{g \glb h}{\bot} \ge \similarr{g}{h}$
  hold.

  In order to show the former, let $d = \sdepth{g \glb
    h}{\bot}$. Since $g \glb h \lebotr g, h$, we can apply
  Lemma~\ref{lem:lebotTrunc} twice in order to obtain $\truncr{g}{d}
  \isom \truncr{(g \glb h)}{d} \isom \truncr{h}{d}$. Hence,
  $\similarr{g}{h} \ge d$.

  To show the converse direction, let $d = \similarr{g}{h}$, i.e.\
  $\truncr{g}{d} \isom \truncr{h}{d}$. According to
  Lemma~\ref{lem:truncLebot}, we have both $\truncr{g}{d} \lebotr g$
  and $\truncr{h}{d} \lebotr h$. Note that, for the canonical
  representation, we then have $\canon{\truncr{g}{d}} =
  \canon{\truncr{h}{d}}$, $\canon{\truncr{g}{d}} \lebotr g$ and
  $\canon{\truncr{h}{d}} \lebotr h$ (cf.\ Proposition~\ref{prop:canon}
  respectively Remark~\ref{rem:lebot}). That is, $\canon{\truncr{g}{d}}$ is
  a lower bound of $g$ and $h$. Thus, $\canon{\truncr{g}{d}} \lebotr g
  \glb h$ and we can reason as follows:
  \begin{align*}
    d &\le \sdepth{\truncr{g}{d}}{\bot} \tag{Lem.~\ref{lem:truncDepth}} \\
    &= \sdepth{\canon{\truncr{g}{d}}}{\bot}
    \tag{Lem.~\ref{lem:strongDepthPres}, Cor.~\ref{cor:isomStrong}} \\
    &\le \sdepth{g \glb h}{\bot} \tag{$\canon{\truncr{g}{d}} \lebotr g
      \glb h$, Lem.~\ref{lem:DhomDeltaDepth}}
  \end{align*}
\end{proof}

\begin{remark}
  From now on, we are not dealing with the concrete construction of
  rigid truncations $\truncr{g}{d}$ anymore. Therefore, we will rather
  use the canonical representation $\canon{\truncr{g}{d}}$ of
  $\truncr{g}{d}$. In order to avoid the notational overhead, we write
  $\truncr{g}{d}$ instead of $\canon{\truncr{g}{d}}$.
\end{remark}

In the next step we show that the metric space $(\ictgraphs,\ddr)$ is
indeed complete. The following proposition states even more: the limit
of Cauchy sequences in the metric space equals its limit inferior in
the partially ordered set $(\ipctgraphs,\lebotr)$:
\begin{proposition}[metric limit = limit inferior]
  \label{prop:graphLimLimInf}
  % metric limit equals limit inferior %
  Let $(g_\iota)_{\iota<\alpha}$ be a non-empty Cauchy sequence in the
  metric space $(\ictgraphs,\ddr)$ and $\liminf_{\iota \limto \alpha}
  g_\iota$ its limit inferior in $(\ipctgraphs,\lebotr)$. Then
  $\lim_{\iota \limto \alpha}g_\iota = \liminf_{\iota \limto \alpha}
  g_\iota$.
\end{proposition}
\begin{proof}
  \def\claima{(\ref{eq:metricGraphCompla})} If $\alpha$ is a successor
  ordinal, this is trivial, as the limit and the limit inferior are
  then $g_{\alpha-1}$. Assume that $\alpha$ is a limit ordinal and let
  $\ol g$ be the limit inferior of $(g_\iota)_{\iota<\alpha}$. Since,
  according to Theorem~\ref{thm:graphBcpo}, $(\ipctgraphs,\lebotr)$ is
  a complete semilattice, $\ol g$ is well-defined. Since
  $(g_\iota)_{\iota < \alpha}$ is Cauchy, we obtain that, for each
  $\epsilon \in \realsp$, there is a $\beta < \alpha$ such that we
  have $\ddr(g_\iota,g_{\iota'}) < \epsilon$ for all $\beta \le
  \iota,\iota' < \alpha$. A fortiori, we get that, for each $\epsilon
  \in \realsp$, there is a $\beta < \alpha$ such that we have
  $\ddr(g_\beta,g_\iota) < \epsilon$ for all $\beta \le \iota <
  \alpha$. By definition of $\ddr$, this is equivalent to
  $2^{-\similarr{g_\beta}{g_\iota}} < \epsilon$. Consequently, we
  have, for each $d < \omega$, a $\beta < \alpha$ such that
  $\similarr{g_\beta}{g_\iota} > d$ for all $\beta \le \iota <
  \alpha$. Due to Lemma~\ref{lem:smallerTrunc},
  $\similarr{g_\beta}{g_\iota} > d$ implies $\truncr{g_\beta}{d} =
  \truncr{g_\iota}{d}$ which in turn implies $\truncr{g_\beta}{d}
  \lebotr g_\iota$, according to Lemma~\ref{lem:truncLebot}. Hence,
  $\truncr{g_\beta}{d}$ is a lower bound for $G_\beta =
  \setcom{g_\iota}{\beta \le \iota < \alpha}$, i.e.\
  $\truncr{g_\beta}{d} \lebotr \Glb G_\beta$. Moreover, by the
  definition of the limit inferior, it holds that $\Glb G_\beta
  \lebotr \ol g$. Consequently, $\truncr{g_\beta}{d} \lebotr \ol g$,
  i.e.\ we have
  \begin{gather*}
    \forall d < \omega \exists \beta < \alpha\fcolon \quad \truncr{g_\beta}{d} \lebotr
    \ol g
    \tag{1} \label{eq:metricGraphCompla}
  \end{gather*}
  Since, by Lemma~\ref{lem:truncDepth}, we have
  $\sdepth{\truncr{g_\beta}{d}}{\bot} \ge d$, we can apply
  Lemma~\ref{lem:lebotTrunc} to obtain $\truncr{(\truncr{g_\beta}{d})}{d}
  \isom \truncr{\ol g}{d}$. Hence, by
  Lemma~\ref{lem:smallerTrunc}, we have $\truncr{g_\beta}{d} \isom
  \truncr{\ol g}{d}$ and therefore $\similarr{\ol g}{g_\beta} \ge
  d$. That is, we have shown that
  \begin{gather*}
    \forall d < \omega \exists \beta < \alpha\fcolon \quad \similarr{\ol g}{g_\beta}
    \ge d 
  \end{gather*}
  Since, for each $\epsilon \in \realsp$, we find a $d < \omega$ with
  $2^{-d} < \epsilon$, this implies
  \begin{gather*}
    \forall \epsilon \in \realsp \exists \beta < \alpha\fcolon \quad
    \ddr(\ol g,g_\beta) < \epsilon
  \end{gather*}
  This shows that $(g_\iota)_{i<\alpha}$ converges to $\ol g$. Now it
  remains to be shown that $\ol g$ is indeed in $\ictgraphs$, i.e.\ it
  does not contain any node labelled $\bot$. Suppose that $\ol g$ does
  contain a node labelled with $\bot$. Then $\sdepth{\ol g}{\bot} <
  \omega$. Let $d = \sdepth{\ol g}{\bot} + 1$. By \claima{}, there is
  a $\beta$ with $\truncr{g_\beta}{d} \lebotr \ol g$. By applying
  Lemma~\ref{lem:truncDepth} and Lemma~\ref{lem:DhomDeltaDepth}, we
  then get
  \[
  \sdepth{\ol g}{\bot} + 1 = d \le \sdepth{\truncr{g_\beta}{d}}{\bot}
  \le \sdepth{\ol g}{\bot}.
  \]
  This is a contradiction. Hence, $\ol g$ is indeed in $\ictgraphs$. 
\end{proof}

This result has two obvious but important consequences: firstly, the
limit of a converging sequence in the rigid metric space is equal to
the limit inferior in the rigid complete semilattice. Secondly, the
rigid metric space $(\ictgraphs,\ddr)$ is complete:
\begin{theorem}[completeness of rigid metric]
  \label{thr:metricGraphCompl}
  % completeness of metric on term graph %
  The metric space $(\ictgraphs,\ddr)$ is complete.
\end{theorem}
\begin{proof}
  Immediate consequence of Proposition~\ref{prop:graphLimLimInf}.
\end{proof}

In the following proposition, we show the converse direction of the
relation between the limits of the rigid metric and the limit
inferiors of the rigid partial order:
\begin{proposition}[total limit inferior = limit]
  \label{prop:graphLimInfLim}
  % total limit inferior equals limit %
  Let $(g_\iota)_{\iota < \alpha}$ be a non-empty sequence in
  $(\ictgraphs,\ddr)$ and $\liminf_{\iota \limto \alpha} g_\iota$ its
  limit inferior in $(\ipctgraphs,\lebotr)$. If $\liminf_{\iota \limto
    \alpha} g_\iota \in \ictgraphs$, then $\liminf_{\iota \limto
    \alpha} g_\iota = \lim_{\iota \limto \alpha} g_\iota$.
\end{proposition}
\begin{proof}
  \def\claima{(\ref{eq:graphLimInfLim1})}
  \def\claimb{(\ref{eq:graphLimInfLim2})}
  \def\claimc{(\ref{eq:graphLimInfLim3})}
  If $\alpha$ is a successor ordinal, then both the limit and the
  limit inferior are equal to $g_{\alpha-1}$. Let $\alpha$ be a limit
  ordinal. According to Proposition~\ref{prop:graphLimLimInf}, in
  order to show that $(g_\iota)_{\iota < \alpha}$ converges and that
  its limit coincides with its limit inferior, it suffices to prove
  that $(g_\iota)_{\iota < \alpha}$ is Cauchy.

  Let $\ol g = \liminf_{\iota \limto \alpha} g_\iota$, and define
  $G_\beta=\setcom{g_\iota}{\beta\le\iota<\alpha}$ and $h_\beta = \Glb
  G_\beta$ for each $\beta< \alpha$. Note that $\ol g = \Lub_{\beta <
    \alpha} h_\beta$. Since $\ol g$ is total, i.e.\ no node in $\ol g$
  is labelled with $\bot$, we have, according to
  Theorem~\ref{thm:graphCpo}, that for each $\pi \in \pos{\ol g}$
  there is some $\beta_\pi < \alpha$ with $h_{\beta_\pi}(\pi) = \ol
  g(\pi)$.

  Note that $(h_\iota)_{\iota<\alpha}$ is monotonic w.r.t.\ $\lebotr$,
  i.e.\ $\iota \le \iota'$ implies $h_\iota \lebotr h_{\iota'}$. Since
  $h_\iota \lebotr h_{\iota'}$ together with $h_\iota(\pi) \in \Sigma$
  implies $h_{\iota'}(\pi)=h_\iota(\pi)$ by
  Corollary~\ref{cor:chaTgraphPo}, we have $h_\gamma(\pi) = \ol
  g(\pi)$ for all $\pi \in \pos{\ol g}$ and $\beta_\pi \le \gamma <
  \alpha$.

  Let $d < \omega$. Since there are only finitely many positions in
  $\pos{\ol g}$ of length smaller than $d$, we know that $\beta_d =
  \max\setcom{\beta_\pi}{\pi \in \pos{\ol g},\len{\pi} < d}$ is a
  well-defined ordinal smaller than $\alpha$. Hence, for all $\pi \in
  \pos{\ol g}$ with $\len{\pi} < d$ we have $h_{\beta_d}(\pi) = \ol
  g(\pi)$. Since $\ol g$ is total, we thus have that
  $\sdepth{h_{\beta_d}}{\bot} \ge d$.

  Since $g_\iota,g_{\iota'} \in G_{\beta_d}$ for each $\beta_d \le
  \iota,\iota' < \alpha$, we have $h_{\beta_d} \lebotr g_\iota,
  g_{\iota'}$ and thus $h_{\beta_d} \lebotr g_\iota \glb
  g_{\iota'}$. Consequently, by Lemma~\ref{lem:DhomDeltaDepth}, we
  have that $\sdepth{g_\iota \glb g_{\iota'}}{\bot} \ge
  \sdepth{h_{\beta_d}}{\bot}$. That is,
  \[
  \text{ $\similarr{g_\iota}{g_{\iota'}}
    \stackrel{Prop.~\ref{prop:similarGlb}}= \sdepth{g_\iota \glb
      g_{\iota'}}{\bot} \ge \sdepth{h_{\beta_d}}{\bot} \ge d$ \qquad
    for each $\beta_d < \iota,\iota' < \alpha$.}
  \]

  Now, let $\epsilon \in \realsp$. We then find some $d < \omega$ with
  $\epsilon > 2^{-d}$. Consequently, we have
  \[
  \text{ $\ddr(g_\iota,g_{\iota'}) =
    2^{-\similarr{g_\iota}{g_{\iota'}}} \le 2^{-d} < \epsilon$ \qquad
    for all $\beta_d \le \iota,\iota' < \alpha$.}
  \]
  Hence, $(g_\iota)_{\iota<\alpha}$ is Cauchy.
\end{proof}

Note that Proposition~\ref{prop:graphLimInfLim} depends on the
finiteness of the arity of the symbols in the signature. (This is used
in the proof above when observing that a term graph has only finitely
many positions of a bounded length.) This restriction also applies to
terms as the following example shows:

\begin{example}
  Let $\Sigma = \set{f/\omega,a/0,b/0}$ and $(t_i)_{i<\omega}$ a
  sequence with
  \begin{align*}
    t_0 &= f(a,a,a,a,a \dots), \\
    t_1 &= f(b,a,a,a,a \dots), \\
    t_2 &= f(b,b,a,a,a \dots), \\
    t_3 &= f(b,b,b,a,a \dots), \\
    &\phantom{= f(b,b,b,}\vdots
  \end{align*}
  $(t_i)_{i<\omega}$ has the limit inferior $f(b,b,b,b,b,\dots)$. On
  the other hand, the sequence is not even Cauchy since, for each
  $i\neq j$, we have $\similarr{t_i}{t_j} = 1$ and, therefore,
  $\ddr(t_i,t_j)=\frac{1}{2}$.
\end{example}

\section{Infinitary Term Graph Rewriting}
\label{sec:infin-term-graph}

In the previous sections, we have constructed and investigated the
necessary metric and partial order structures upon which the
infinitary calculus of term graph rewriting that we shall introduce in
this section is based. After describing the framework of term graph
rewriting that we consider, we will explore two different modes of
convergence on term graphs. In the same way that infinitary term
rewriting is based on the abstract notions of $\mrs$- and
$\prs$-convergence~\cite{bahr10rta}, infinitary term graph rewriting
is an instantiation of these abstract modes of convergence for term
graphs. However, as in the overview of infinitary term rewriting in
Section~\ref{sec:preliminaries}, we restrict ourselves to weak notions
of convergence.

\subsection{Term Graph Rewriting Systems}
\label{sec:graph-rewr-syst}

In this paper, we adopt the term graph rewriting framework of
Barendregt et al.~\cite{barendregt87parle}. In order to represent
placeholders in rewrite rules, this framework uses variables -- in a
manner much similar to term rewrite rules. To this end, we consider a
signature $\Sigma_\calV = \Sigma\uplus\calV$ that extends the
signature $\Sigma$ with a set $\calV$ of nullary variable symbols.
\begin{definition}[term graph rewriting systems]
  \quad
  \begin{enumerate}[(i)]
  \item Given a signature $\Sigma$, a \emph{term graph rule} $\rho$
    over $\Sigma$ is a triple $(g,l,r)$ where $g$ is a graph over
    $\Sigma_\calV$ and $l,r \in N^g$ such that all nodes in $g$ are
    reachable from $l$ or $r$. We write $\lhs\rho$ respectively
    $\rhs\rho$ to denote the left- respectively right-hand side of
    $\rho$, i.e.\ the term graph $\subgraph{g}{l}$ respectively
    $\subgraph{g}{r}$. Additionally, we require that, for each
    variable $v\in\calV$, there is at most one node $n$ in $g$
    labelled $v$ and that $n$ is different but still reachable from
    $l$.
  \item A \emph{term graph rewriting system (GRS)} $\calR$ is a pair
    $(\Sigma,R)$ with $\Sigma$ a signature and $R$ a set of term graph
    rules over $\Sigma$.
  \end{enumerate}
\end{definition}

\noindent
The requirement that the root $l$ of the left-hand side is not
labelled with a variable symbol is analogous to the requirement that
the left-hand side of a term rule is not a variable. Similarly, the
restriction that nodes labelled with variable symbols must be
reachable from the root of the left-hand side corresponds to the
restriction on term rules that every variable occurring on the
right-hand side of a rule must also occur on the left-hand side.

Term graphs can be used to compactly represent terms, which is
formalised by the unravelling operator $\unrav{\cdot}$. We extend this
operator to term graph rules. Figure~\ref{fig:grules} illustrates two
term graph rules that both represent the term rule $a \cons x \to b
\cons a \cons x$ from Example~\ref{ex:termRewr} to which they unravel.
\begin{definition}[unravelling of term graph rules]
  Let $\rho$ be a term graph rule with $\rho_l$ and $\rho_r$ its left-
  respectively right-hand side term graph. The \emph{unravelling} of
  $\rho$, denoted $\unrav{\rho}$ is the term rule $\unrav{\rho_l} \to
  \unrav{\rho_r}$.
\end{definition}

The application of a rewrite rule $\rho$ (with root nodes $l$ and $r$)
to a term graph $g$ is performed in three steps: at first a suitable
sub-term graph of $g$ rooted in some node $n$ of $g$ is \emph{matched}
against the left-hand side of $\rho$. This amounts to finding a
$\calV$-homomorphism $\phi\fcolon \lhs{\rho} \homto_\calV
\subgraph{g}{n}$ from the term graph rooted in $l$ to the sub-term
graph rooted in $n$, the \emph{redex}. The $\calV$-homomorphism $\phi$
allows us to instantiate variables in the rule with sub-term graphs of
the redex. In the second step, nodes and edges in $\rho$ that are not
reachable from $l$ are copied into $g$, such that each edge pointing
to a node $m$ in the term graph rooted in $l$ is redirected to
$\phi(m)$. In the last step, all edges pointing to $n$ are redirected
to (the copy of) $r$ and all nodes not reachable from the root of (the
now modified version of) $g$ are removed.

The formal definition of this construction is given below:
\begin{definition}[application of term graph rewrite rules,
  \cite{barendregt87parle}]
  \label{def:termGraphApp}
  Let $\rho = (N^\rho,\glab^\rho,\gsuc^\rho,l^\rho,r^\rho)$ be a term
  graph rewrite rule in a GRS $\calR = (\Sigma,R)$, $g \in \itgraphs$
  with $N^\rho \cap N^g = \emptyset$ and $n \in N^g$. $\rho$ is called
  \emph{applicable} to $g$ at $n$ if there is a $\calV$-homomorphism
  $\phi\fcolon\rho_l \homto_\calV \subgraph{g}{n}$. $\phi$ is called
  the \emph{matching $\calV$-homomorphism} of the rule application,
  and $\subgraph{g}{n}$ is called a \emph{$\rho$-redex}. Next, we
  define the \emph{result} of the application of the rule $\rho$ to
  $g$ at $n$ using the $\calV$-homomorphism $\phi$. This is done by
  constructing the intermediate graphs $g_1$ and $g_2$, and the final
  result $g_3$.
  \begin{enumerate}[(i)]
  \item The graph $g_1$ is obtained from $g$ by adding the part of
    $\rho$ that is not contained in its left-hand side:
    \begin{align*}
      N^{g_1} &= N^g \uplus (N^\rho \setminus N^{\rho_l})\\
      \glab^{g_1}(m) &= 
      \begin{cases}
        \glab^g(m) & \text{if } m \in N^g\\
        \glab^\rho(m) & \text{if } m \in N^\rho \setminus N^{\rho_l}
      \end{cases}\\
      \gsuc^{g_1}_i(m) &= 
      \begin{cases}
        \gsuc^g_i(m) & \text{if } m \in N^g\\
        \gsuc^\rho_i(m) & \text{if } m, \gsuc^\rho_i(m) \in N^\rho \setminus N^{\rho_l}\\
        \phi(\gsuc^\rho_i(m)) & \text{if } m \in N^\rho \setminus
        N^{\rho_l}, \gsuc^\rho_i(m) \in N^{\rho_l}
      \end{cases}
    \end{align*}
  \item Let $n' = \phi(r^\rho)$ if $r^\rho \in N^{\rho_l}$ and $n'=
    r^\rho$ otherwise. The graph $g_2$ is obtained from $g_1$ by
    redirecting edges ending in $n$ to $n'$:
    \begin{align*}
      N^{g_2} = N^{g_1} \qquad
      \glab^{g_2} = \glab^{g_1}\qquad
      \gsuc^{g_2}_i(m) = 
      \begin{cases}
        \gsuc_i^{g_1}(m) &\text{if } \gsuc^{g_1}_i(m) \neq n\\
        n' &\text{if } \gsuc^{g_1}_i(m) = n
      \end{cases}
    \end{align*}
  \item The term graph $g_3$ is obtained by setting the root node
    $r'$, which is $n'$ if $n = r^g$, and otherwise $r^g$. That is,
    $g_3 = \subgraph{g_2}{r'}$. This also means that all nodes not
    reachable from $r'$ are removed.
  \end{enumerate}
  
\noindent
  This induces a \emph{pre-reduction step} $\psi=(g,n,\rho,n',g_3)$
  from $g$ to $g_3$, written $\psi\fcolon g \preto[n,\rho,n'] g_3$. In
  order to indicate the underlying GRS $\calR$, we also write
  $\psi\fcolon g \preto[\calR] g_3$.
\end{definition}

Examples for term graph (pre-)reduction steps are shown in
Figure~\ref{fig:gredEx}. We revisit them in more detail in
Example~\ref{ex:gRed} below.

The definition of term graph rewriting in the form of pre-reduction
steps is very operational in style. The result of applying a rewrite
rule to a term graph is constructed in several steps by manipulating
nodes and edges explicitly. While this is beneficial for implementing
a rewriting system, it is problematic for reasoning on term graphs up
to isomorphisms, which is necessary for introducing notions of
convergence. In our case, however, this does not cause any harm since
the construction in Definition~\ref{def:termGraphApp} is invariant
under isomorphism:
\begin{proposition}[pre-reduction steps]
  Let $\phi\fcolon g \preto[n,\rho,m] h$ be a pre-reduction step in
  some GRS $\calR$ and $\psi_1\fcolon g' \isom g$. Then there is a
  pre-reduction step $\phi'\fcolon g' \preto[n',\rho,m'] h'$ with
  $\psi_2\fcolon h' \isom h$ such that $\psi_2(n') = n$ and
  $\psi_1(m') = m$.
\end{proposition}
\begin{proof}
  Immediate from the construction in Definition~\ref{def:termGraphApp}.
\end{proof}

The above finding justifies the following definition of reduction
steps:
\begin{definition}[reduction steps]
  Let $\calR = (\Sigma,R)$ be a GRS, $\rho \in R$ and $g,h \in
  \ictgraphs$ with $n \in N^g$ and $m\in N^h$. A tuple $\phi =
  (g,n,\rho,m,h)$ is called a \emph{reduction step}, written
  $\phi\fcolon g \to[n,\rho,m] h$, if there is a pre-reduction step
  $\phi'\fcolon g' \preto[n',\rho,m'] h'$ with $\canon{g'} = g$,
  $\canon{h'} = h$, $n = \nodePos{g'}{n'}$, and $m =
  \nodePos{h'}{m'}$. Similarly to pre-reduction steps, we also write
  $\phi\fcolon g \to[\calR] h$ or simply $\phi\fcolon g \to h$ for
  short.
\end{definition}
In other words, a reduction step is a canonicalised pre-reduction
step.

Note that term graph rules do not provide a duplication
mechanism. Each variable is allowed to occur at most once. Duplication
must always be simulated by sharing. This means for example that
variables that should occur on the right-hand side must share the
occurrence of that variable on the left-hand side of the rule with its
right-hand side. This can be seen in the term graph rules in
Figure~\ref{fig:grules}. The sharing can be direct as in $\rho_1$ or
indirect as in $\rho_2$. For variables that are supposed to be
duplicated on the right-hand side, for example in the term rewrite
rule $h(x) \to f(h(x),h(x))$, we have to use sharing in order to
represent multiple occurrences of the same variable. This
representation can be seen in the corresponding term graph rules in
Figure~\ref{fig:doubleRules}.

\subsection{Convergence of Transfinite Reductions}
\label{sec:weak-convergence}

We now employ the partial order $\lebotr$ and the metric $\ddr$ for
the purpose of defining convergence of transfinite term graph
reductions. 

The notion of (transfinite) reductions carries over to GRSs
straightforwardly:
\begin{definition}[transfinite reductions]
  Let $\calR = (\Sigma,R)$ be a GRS. A \emph{(transfinite) reduction}
  in $\calR$ is a sequence $(g_\iota \to[\calR] g_{\iota+1})_{i <
    \alpha}$ of rewriting steps in $\calR$.%
  % If $S$ is finite, we write $S\fcolon g_0 \fto{*} g_\alpha$.
\end{definition}
Analogously to reductions in TRSs, we need a notion of convergence in
order to define well-behaved reductions. The two modes of convergence
that we introduced for this very purpose in
Section~\ref{sec:partial-order-lebot1} and
Section~\ref{sec:alternative-metric} are only defined on canonical
term graphs. It is therefore crucial to work on reduction steps as
opposed to pre-reduction steps.

\begin{definition}[convergence of reductions]
  Let $\calR = (\Sigma,R)$ be a GRS.
  \begin{enumerate}[(i)]
  \item Let $S = (g_\iota \to_\calR g_{\iota+1})_{\iota < \alpha}$ be
    a reduction in $\calR$. $S$ is \emph{$\mrs$-continuous} in
    $\calR$, written $S\fcolon g_0 \wmacont[\calR]$, if the underlying
    sequence of term graphs $(g_\iota)_{\iota < \wsuc\alpha}$ is
    continuous in $\calR$, i.e.\ $\lim_{\iota\limto\lambda} g_\iota =
    g_\lambda$ for each limit ordinal $\lambda < \alpha$. $S$
    \emph{$\mrs$-converges} to $g \in \ictgraphs$ in $\calR$, written
    $S\fcolon g_0 \wmato[\calR] g$, if it is $\mrs$-continuous and
    $\lim_{\iota\limto\wsuc\alpha} g_\iota = g$.
  \item Let $\calR_\bot$ be the GRS $(\Sigma_\bot, R)$ over the
    extended signature $\Sigma_\bot$ and $S = (g_\iota \to[\calR_\bot]
    g_{\iota+1})_{\iota < \alpha}$ a reduction in $\calR_\bot$. $S$ is
    \emph{$\prs$-continuous} in $\calR$, written $S\fcolon g_0
    \wpacont[\calR]$, if $\liminf_{\iota\limto\lambda} g_i =
    g_\lambda$ for each limit ordinal $\lambda < \alpha$. $S$
    \emph{$\prs$-converges} to $g\in\ipctgraphs$ in $\calR$, written
    $S\fcolon g_0 \wpato[\calR] g$, if it is $\prs$-continuous and
    $\liminf_{\iota\limto\wsuc\alpha} g_i = g$.
  \item Let $S = (g_\iota \to[\calR_\bot] g_{\iota+1})_{\iota <
      \alpha}$ be a reduction in $\calR_\bot$. The reduction $S$ is
    called \emph{$\prs$-continuous in $\ictgraphs$}, if it is
    $\prs$-continuous and $g_\iota \in \ictgraphs$ for all $\iota <
    \wsuc\alpha$. The reduction $S$ is said to \emph{$\prs$-converge
      in $\ictgraphs$} to $g$, if it is $\prs$-continuous in
    $\ictgraphs$ and $\prs$-converges to $g \in \ictgraphs$.
  \end{enumerate}
\end{definition}

\noindent
Note that, analogously to $\prs$-convergence on terms, we extended the
signature of $\calR$ to $\Sigma_\bot$ for the definition of
$\prs$-convergence. Like for terms, this approach serves two
purposes. First, by considering the extended signature $\Sigma_\bot$,
we allow any partial term graph to appear in a reduction as opposed to
only total ones. Consequently, we have the whole complete semilattice
$(\ipctgraphs,\lebotr)$ at our disposal, which means that
$\prs$-continuity coincides with $\prs$-convergence:
\begin{proposition}
  In a GRS, every $\prs$-continuous reduction is $\prs$-convergent.
\end{proposition}
\begin{proof}
  Follows immediately from Corollary~\ref{cor:lebotrLiminf}.
\end{proof}

The second reasons for the extension to $\calR_\bot$ is that by not
presupposing that the system's signature $\Sigma$ already contains a
designated $\bot$-symbol, we rule out the possibility that this $\bot$
symbol occurs in one of the rules of the system. Consequently, any
$\bot$ symbol present in the final term graph of a reduction is either
due to the initial term graph or the convergence behaviour. This is
crucial for establishing a correspondence result between $\mrs$- and
$\prs$-convergence in the vein of Theorem~\ref{thr:strongExt}.

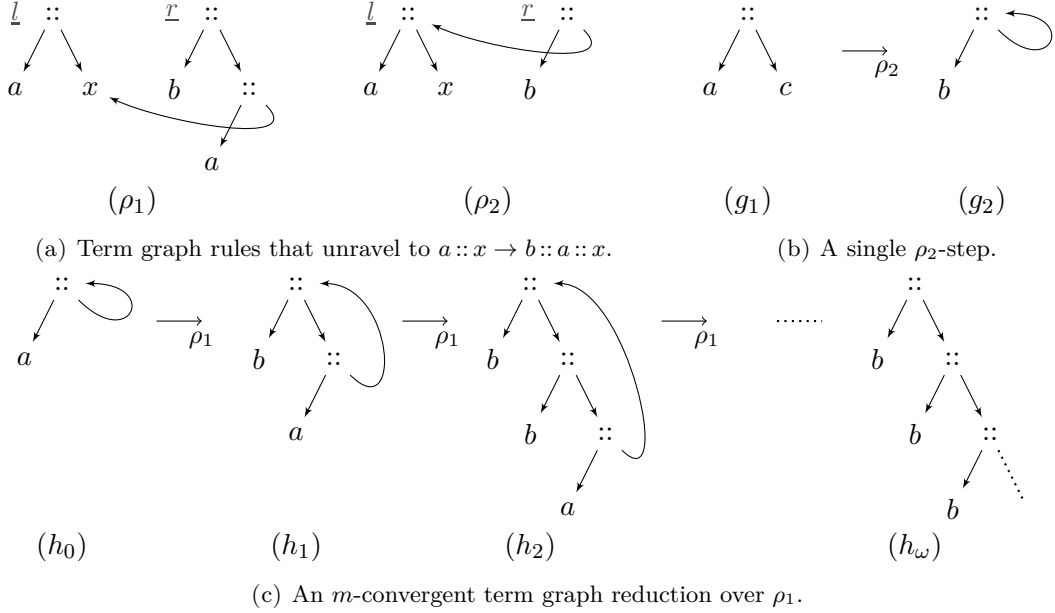
\begin{figure}
  \centering \subfloat[Term graph rules that unravel to $a \cons x
  \rightarrow b \cons a \cons x$.]{
    \hspace{2mm}
      \begin{tikzpicture}
      \node [node name=180:\underline{l}] (l) {$\cons$}%
      child {%
        node (a) {$a$}%
      } child {%
        node (x) {$x$}%
      };%

      \node[node name=180:\underline{r},node distance=1.5cm,right=of l] (r) {$\cons$}%
      child {%
        node (b) {$b$}%
      } child {%
        node (c) {$\cons$}%
        child {%
          node (a') {$a$}%
        } child [missing]
      };%
      \draw%
      (c) edge[->,out=-45,in=-25] (x);%
      
      \node () at ($(l)!.5!(r) + (0,-2.5)$) {$(\rho_1)$};
      \node [node name=180:\underline{l},node distance=2cm,right=of r] (l) {$\cons$}%
      child {%
        node (a) {$a$}%
      } child {%
        node (x) {$x$}%
      };%

      \node[node name=180:\underline{r},node distance=1.5cm,right=of l] (r) {$\cons$}%
      child {%
        node (b) {$b$}%
      } child [missing];%
      \draw%
      (r) edge[->,out=-45,in=-25] (l);%
      
      \node () at ($(l)!.5!(r) + (0,-2.5)$) {$(\rho_2)$};
    \end{tikzpicture}
    \label{fig:grules}
  \hspace{5mm}
  }
  \subfloat[A single $\rho_2$-step.]{
      \begin{tikzpicture}
        
      \node (g1) {$\cons$}%
      child {%
        node {$a$}%
      } child {%
        node {$c$}
        };%

      \node () at ($(g1) + (0,-2.5)$) {$(g_1)$};

      \node [node distance=2.5cm,right=of g1] (g2) {$\cons$}%
      child {%
        node {$b$}%
      } child [missing];%
      \draw (g2) edge[->, min distance=10mm,out=-45,in=0] (g2);%

      \node () at ($(g2) + (0,-2.5)$) {$(g_2)$};
      \node (s1) at ($(g1)!.5!(g2)-(0,.5)$) {};

      \draw[single step] ($(s1)-(.3,0)$) -- ($(s1)+(.3,0)$)
      node[pos=1,below] {{\small$\rho_2$}};
    \end{tikzpicture}
    \label{fig:gsingleStep}
  }
  \hspace{1cm}%
  \subfloat[An $\mrs$-convergent term graph reduction over $\rho_1$.]{%
      \begin{tikzpicture}
        
      \node (g1) {$\cons$}%
      child {%
        node (a) {$a$}%
      } child [missing];%
      \draw (g1) edge[->, min distance=10mm,out=-45,in=0] (g1);%

      \node () at ($(g1) + (0,-3.5)$) {$(h_0)$};
        
      \node [node distance=2.5cm,right=of g1] (g2) {$\cons$}%
      child {%
        node (b) {$b$}%
      } child {%
        node (c) {$\cons$}%
        child {%
          node (a) {$a$}%
        } child [missing]%
      };%
      \draw (c) edge[->, min distance=10mm,out=-45,in=0] (g2);%

      \node () at ($(g2) + (0,-3.5)$) {$(h_1)$};
        
      \node [node distance=2.5cm,right=of g2] (g3) {$\cons$}%
      child {%
        node (b) {$b$}%
      } child {%
        node (c) {$\cons$}%
        child {%
          node (b2) {$b$}%
        } child {%
          node (c2) {$\cons$}%
          child {%
            node (a) {$a$}%
          } child [missing]%
        }%
      };%
      \draw (c2) edge[->, min distance=10mm,out=-45,in=0] (g3);%

      \node () at ($(g3) + (0,-3.5)$) {$(h_2)$};
        
      \node [node distance=4.5cm,right=of g3] (go) {$\cons$}%
      child {%
        node (b) {$b$}%
      } child {%
        node (c) {$\cons$}%
        child {%
          node (b2) {$b$}%
        } child {%
          node (c2) {$\cons$}%
          child {%
            node (b3) {$b$}%
          } child[etc] {%
            node {} %
          }%
        }%$
      };%

      \node () at ($(go) + (0,-3.5)$) {$(h_\omega)$};

    \node (s1) at ($(g1)!.5!(g2)-(0,.5)$) {};
    \node (s2) at ($(g2)!.55!(g3)-(0,.5)$) {};
    \node (s3) at ($(g3)!.4!(go)-(0,.5)$) {};
    \node (s4) at ($(g3)!.7!(go)-(0,.5)$) {};

    \draw[single step] ($(s1)-(.3,0)$) -- ($(s1)+(.3,0)$)
    node[pos=1,below] {{\small$\rho_1$}};
    \draw[single step] ($(s2)-(.3,0)$) -- ($(s2)+(.3,0)$)
    node[pos=1,below] {{\small$\rho_1$}};
    \draw[single step] ($(s3)-(.3,0)$) -- ($(s3)+(.3,0)$)
    node[pos=1,below] {{\small$\rho_1$}};
    \draw[dotted,thick,-] ($(s4)-(.3,0)$) -- ($(s4)+(.3,0)$);
        
    \end{tikzpicture}
    \label{fig:gtransRed}
  }  
  \caption{Term graph rules and their reductions.}
  \label{fig:gredEx}
\end{figure}

\begin{example}
\label{ex:gRed}  
Consider the term graph rule $\rho_1$ in Figure~\ref{fig:grules},
which unravels to the term rule $a \cons x \to b \cons a \cons x$ from
Example~\ref{ex:termRewr}. Starting with the term tree $a \cons c$,
depicted as $g_1$ in Figure~\ref{fig:gsingleStep}, we obtain the same
transfinite reduction as in Example~\ref{ex:termRewr}:
\[
S\fcolon a \cons c \to[\rho_1] b \cons a \cons c \to[\rho_1] b \cons b \cons a \cons c
\to[\rho_1] \;\dots
\]
Since the modes of convergence of both the partial order $\lebotr$ and
the metric $\ddr$ coincide with the corresponding modes of convergence
on terms (cf.\ Proposition~\ref{prop:liminfGeneralise} respectively
Proposition~\ref{prop:limGeneralise}), we know that, for reductions
consisting only of term trees, both $\mrs$- and $\prs$-convergence in
GRSs coincide with the corresponding notions of convergence in
TRSs. This observation applies to the reduction $S$ above. Hence, also
in this setting of term graph rewriting, $S$ both $\mrs$- and
$\prs$-converges to the term tree $h_\omega$ shown in
Figure~\ref{fig:gtransRed}. Similarly, we can reproduce the
$\prs$-converging but not $\mrs$-converging reduction $T$ from
Example~\ref{ex:termRewr2}.

Notice that $h_\omega$ is a rational term tree as it can be obtained
by unravelling the finite term graph $g_2$ depicted in
Figure~\ref{fig:gsingleStep}. In fact, if we use the rule $\rho_2$,
which unravels to the term rule $a \cons x \to b \cons a \cons x$ as
well, we can immediately rewrite $g_1$ to $g_2$. In $\rho_2$, not only
the variable $x$ is shared but the whole left-hand side of the
rule. This causes each redex of $\rho_2$ to be \emph{captured} by the
right-hand side~\cite{farmer90ijfocs}.

Figure~\ref{fig:gtransRed} indicates a transfinite reduction starting
with a cyclic term graph $h_0$ that unravels to the rational term $t =
a \cons a \cons a\cons \dots$. This reduction both $\mrs$- and
$\prs$-converges to the rational term tree $h_\omega$ as well. Again,
by using $\rho_2$ instead of $\rho_1$, we can rewrite $h_0$ to the
cyclic term graph $g_2$ in one step.

For more detailed explanations of the underlying modes of partial
order and metric convergence for the reductions above, revisit
Example~\ref{ex:liminf} and Example~\ref{ex:convdiv}, respectively.
\end{example}

The following theorem shows that the total fragment of
$\prs$-converging reductions is in fact equivalent to
$\mrs$-converging reductions:
\begin{theorem}[$\prs$-convergence in $\ictgraphs$ =
  $\mrs$-convergence]
  \label{thr:graphTotalConv}
  For every reduction $S$ in a GRS the following equivalences hold:
  \\[.5em]
    \begin{inparaenum}[\em(i)]
    \begin{tabular}{@{\quad}r@{\;\,}l@{\;}l@{\qquad iff \qquad}l}
    \item\label{item:graphTotalConv1}& $S\fcolon g \wpato[\calR] h$ &in $\ictgraphs$ & $S\fcolon g \wmato[\calR] h$\\%
    \item\label{item:graphTotalConv2}& $S\fcolon g \wpacont[\calR]$ &in $\ictgraphs$ &
      $S\fcolon g \wmacont[\calR]$%
  \end{tabular}
    \end{inparaenum}
\end{theorem}
\begin{proof}
  We only show (\ref{item:graphTotalConv1}) since
  (\ref{item:graphTotalConv2}) follows similarly.

  Let $S = (g_\iota \to[\calR_\bot] g_{\iota+1})_{\iota <
    \alpha}$. For the ``only if'' direction assume $S\fcolon g
  \wpato[\calR] h$ is $\prs$-converging in $\ictgraphs$. Since $S$
  $\prs$-converges in $\ictgraphs$, it is a reduction in $\calR$. The
  $\prs$-convergence of $S$ implies that $\liminf_{\iota\limto\lambda} g_i
  = g_\lambda$ for each limit ordinal $\lambda<\alpha$. Since each
  $g_\iota$ is total, we have, according to
  Proposition~\ref{prop:graphLimInfLim}, that $\lim_{\iota<\lambda}
  g_i = \liminf_{\iota\limto\lambda} g_i = g_\lambda$ for each limit
  ordinal $\lambda<\alpha$. Hence $(g_\iota)_{\iota < \wsuc\alpha}$ is
  continuous in the metric space. Likewise, we also have
  $\lim_{\iota<\wsuc\alpha} g_i = \liminf_{\iota\limto\wsuc\alpha} g_i =
  h$. That is, $S$ $\mrs$-converges to $h$. For the ``if'' direction
  assume $S\fcolon g \wmato[\calR] h$. Since $(g_\iota)_{\iota <
    \wsuc\alpha}$ is continuous, we have that $\lim_{\iota<\lambda}
  g_i = g_\lambda$ for each limit ordinal $\lambda <
  \alpha$. According to Proposition~\ref{prop:graphLimLimInf}, we then
  have that $\liminf_{\iota\limto\lambda} g_i = g_\lambda$ for each limit
  ordinal $\lambda < \alpha$. Likewise we also have
  $\liminf_{\iota\limto\wsuc\alpha} g_i = \lim_{\iota<\wsuc\alpha} g_i =
  h$. Hence, $S$ is $\prs$-converging to $h$. Since $S$ is
  $\mrs$-converging it is by definition also in $\ictgraphs$.
\end{proof}

\begin{figure}
  \centering \subfloat[Term graph rules that unravel to $h(x) \rightarrow f(h(x),h(x))$.]{
  \hspace{7mm}
      \begin{tikzpicture}
      \node [node name=180:\underline{l}] (l) {$h$}%
      child {%
        node (x) {$x$}%
      };%

      \node[node name=180:\underline{r},node distance=1.5cm,right=of l] (r) {$f$};%
      \draw%
      (r) edge[->,out=-45,in=-50] (l)%
      (r) edge[->,out=-135,in=-25] (l);%
      
      \node () at ($(l)!.5!(r) + (0,-2.5)$) {$(\rho_1)$};

      \node [node name=180:\underline{l},node distance=2cm,right=of r] (l) {$h$}%
      child {%
        node (x) {$x$}%
      };%

      \node[node name=180:\underline{r},node distance=1.5cm,right=of l] (r) {$f$}%
      child {%
        node (g) {$h$}%
        edge from parent[transparent] 
      };%
      \draw[->]%
      (g) edge[bend left=45] (x)%
      (r)%
      edge [bend right=25] (g)%
      edge [bend left=25] (g);%
      
      \node () at ($(l)!.5!(r) + (0,-2.5)$) {$(\rho_2)$};

      \node [node name=180:\underline{l},node distance=2cm,right=of r] (l) {$h$}%
      child {%
        node (x) {$x$}%
      };%

      \node[node name=180:\underline{r},node distance=1.5cm,right=of l] (r) {$f$}%
      child {%
        node (h1) {$h$}%
      } child {%
        node (h2) {$h$}%
      };%
      \draw[->]%
      (h1) edge[bend left=45] (x)%
      (h2) edge[bend left=70] (x);%
      
      \node () at ($(l)!.5!(r) + (0,-2.5)$) {$(\rho_3)$};
    \end{tikzpicture}
  \hspace{7mm}
    \label{fig:doubleRules}
  }%
  
  \subfloat[A single $\rho_1$-step.]{
      \begin{tikzpicture}
        
      \node (g1) {$h$}%
      child {%
        node {$c$}
        };%

      \node () at ($(g1) + (0,-2.5)$) {$(g_0)$};

      \node [node distance=2.5cm,right=of g1] (g2) {$f$};%
      \draw (g2) edge[->, min distance=10mm,out=-45,in=0] (g2);%
      \draw (g2) edge[->, min distance=10mm,out=-135,in=180] (g2);%

      \node () at ($(g2) + (0,-2.5)$) {$(g_1)$};
      \node (s1) at ($(g1)!.4!(g2)-(0,.5)$) {};

      \draw[single step] ($(s1)-(.3,0)$) -- ($(s1)+(.3,0)$)
      node[pos=1,below] {{\small$\rho_1$}};
    \end{tikzpicture}
    \label{fig:doubleStep}
  }
  \subfloat[The full binary tree of $f$-nodes.]{%
    \hspace{2cm}
    \begin{tikzpicture}[etc/.style={edge from
      parent/.style={-,dotted,thick,draw}}, level 1/.style={sibling
      distance=10mm}, level 2/.style={sibling distance=5mm}, level
    3/.style={sibling distance=2.5mm} ]
    \node {$f$}%
    child {%
      node {$f$}%
      child {%
        node {$f$}%
        child[etc]{ node {}}%
        child[etc]{ node {}}%
      }%
      child {%
        node {$f$}%
        child[etc]{ node {}}%
        child[etc]{ node {}}%
      }%
    }%
    child {%
      node {$f$}%
      child {%
        node {$f$}%
        child[etc]{ node {}}%
        child[etc]{ node {}}%
      }%
      child {%
        node {$f$}%
        child[etc]{ node {}}%
        child[etc]{ node {}}%
      }%
    };%
  \end{tikzpicture}
    \hspace{2cm}
    \label{fig:bintree}
  }
  
  \subfloat[An $\mrs$-convergent term graph reduction over $\rho_2$.]{%
      \begin{tikzpicture}
        
      \node (g1) {$h$}%
      child {%
        node {$c$}
        };%

      \node () at ($(g1) + (0,-3.5)$) {$(g_0)$};
        
      \node [node distance=2.5cm,right=of g1] (g2) {$f$};%
      \node [node distance=10mm,below of=g2] (b) {$h$}%
      child {%
        node {$c$}
      };%

      \draw[->] (g2)%
      edge [bend right=25] (b)%
      edge [bend left=25] (b);%

      \node () at ($(g2) + (0,-3.5)$) {$(g_1)$};
        
      \node [node distance=2.5cm,right=of g2] (g3) {$f$};%
      \node [node distance=10mm,below of=g3] (b) {$f$};%
      \node [node distance=10mm,below of=b] (b2) {$h$}%
      child {%
        node {$c$}
      };%
      \draw[->] (g3)%
      edge [bend right=25] (b)%
      edge [bend left=25] (b);%
      \draw[->] (b)%
      edge [bend right=25] (b2)%
      edge [bend left=25] (b2);%

      \node () at ($(g3) + (0,-3.5)$) {$(g_2)$};
        
      \node [node distance=4.5cm,right=of g3] (go) {$f$}%
      child {%
        node (b) {$f$}%
        edge from parent[transparent]%
        child {%
          node (b2) {$f$}%
          edge from parent[transparent]%
        }%
      };%
      \node at (b2) {$f$}%
      child [etc] {%
        node {}%
      };%
      \draw[->] (go)%
      edge [bend right=25] (b)%
      edge [bend left=25] (b);%
      \draw[->] (b)%
      edge [bend right=25] (b2)%
      edge [bend left=25] (b2);%

      \node () at ($(go) + (0,-3.5)$) {$(g_\omega)$};

    \node (s1) at ($(g1)!.5!(g2)-(0,.5)$) {};
    \node (s2) at ($(g2)!.55!(g3)-(0,.5)$) {};
    \node (s3) at ($(g3)!.4!(go)-(0,.5)$) {};
    \node (s4) at ($(g3)!.7!(go)-(0,.5)$) {};

    \draw[single step] ($(s1)-(.3,0)$) -- ($(s1)+(.3,0)$)
    node[pos=1,below] {{\small$\rho_2$}};
    \draw[single step] ($(s2)-(.3,0)$) -- ($(s2)+(.3,0)$)
    node[pos=1,below] {{\small$\rho_2$}};
    \draw[single step] ($(s3)-(.3,0)$) -- ($(s3)+(.3,0)$)
    node[pos=1,below] {{\small$\rho_2$}};
    \draw[dotted,thick,-] ($(s4)-(.3,0)$) -- ($(s4)+(.3,0)$);
        
    \end{tikzpicture}
    \label{fig:doubleTransRed}
  }

  \subfloat[A $\prs$-convergent term graph reduction over $\rho_3$.]{%
      \begin{tikzpicture}
        
      \node (g1) {$h$}%
      child {%
        node {$c$}
        };%

      \node () at ($(g1) + (0,-3.8)$) {$(g_0)$};
        
      \node [node distance=2.5cm,right=of g1] (g2) {$f$}%
      child {%
        node (h1) {$h$}%
        child [missing]%
        child {%
          node (c) {$c$}
        }
      } child {%
        node (h2) {$h$}%
      };%

      \draw[->] (h2) edge (c);%

      \node () at ($(g2) + (0,-3.8)$) {$(g_1)$};
         
      \node [node distance=2.5cm,right=of g2] (g3) {$f$}%
      child {%
        node (f) {$f$}%
        child {%
          node (h1) {$h$}%
        } child {%
          node (h2) {$h$}%
        }
      } child {%
        node (h3) {$h$}%
        child [missing]%
        child {%
          node (c) {$c$}%
        }%
      };%

      \draw[->]%
      (h1) edge[bend right=70] (c)%
      (h2) edge[bend right=45] (c);%

      \node () at ($(g3) + (0,-3.8)$) {$(g_2)$};
         
      \node [node distance=4.5cm,right=of g3] (go) {$f$}%
      child {%
        node (f1) {$f$}%
        child {%
          node (f2) {$f$}%
          child [etc] {%
            node {}%
          } child {%
            node (h1) {$h$}%
          }%
        } child {%
          node (h2) {$h$}%
        }
      } child {%
        node (h3) {$h$}%
        child [missing]%
        child {%
          node (c) {$\bot$}%
        }%
      };%

      \draw[->]%
      (h1) edge[bend right=45] (c)%
      (h2) edge[bend right=45] (c);%

      \node () at ($(go) + (0,-3.8)$) {$(g_\omega)$};

    \node (s1) at ($(g1)!.5!(g2)-(0,.5)$) {};
    \node (s2) at ($(g2)!.55!(g3)-(0,.5)$) {};
    \node (s3) at ($(g3)!.4!(go)-(0,.5)$) {};
    \node (s4) at ($(g3)!.7!(go)-(0,.5)$) {};

    \draw[single step] ($(s1)-(.3,0)$) -- ($(s1)+(.3,0)$)
    node[pos=1,below] {{\small$\rho_3$}};
    \draw[single step] ($(s2)-(.3,0)$) -- ($(s2)+(.3,0)$)
    node[pos=1,below] {{\small$\rho_3$}};
    \draw[single step] ($(s3)-(.3,0)$) -- ($(s3)+(.3,0)$)
    node[pos=1,below] {{\small$\rho_3$}};
    \draw[dotted,thick,-] ($(s4)-(.3,0)$) -- ($(s4)+(.3,0)$);
        
    \end{tikzpicture}
    \label{fig:doubleTransRedDiv}
  }
  \caption{Term graph rules for duplicating term rewrite rules.}
  \label{fig:doubleEx}
\end{figure}

\begin{example}
  \label{ex:duplication}
  In order to represent term rewrite rules that are not right-linear,
  i.e.\ which have multiple occurrences of the same variable on the
  right-hand side, we have to use sharing to represent the occurrences
  of a variable by a single node. Consider the term rewrite rule $h(x)
  \to f(h(x),h(x))$ that duplicates the variable $x$ on the right-hand
  side. Note that by repeatedly applying this term rewrite rule
  starting with term $h(c)$, we obtain a reduction that
  $\mrs$-converges to the full binary tree depicted in
  Figure~\ref{fig:bintree}.

  Figure~\ref{fig:doubleRules} shows three different ways of
  representing the term rewrite rule $h(x) \to f(h(x),h(x))$ as a term
  graph rule. The rule $\rho_3$ has the lowest degree of sharing since
  it shares the variable node directly; $\rho_1$ has the highest
  degree of sharing as it shares its complete left-hand side with its
  right-hand side; $\rho_2$ lies in between the two.

  We have observed in Figure~\ref{fig:grules} before that, by sharing
  the complete left-hand side with the right-hand side, the redex gets
  captured by the right-hand side upon applying the rule to a term
  graph. This can be seen again in Figure~\ref{fig:doubleStep}. By
  applying $\rho_1$ to the term tree $h(c)$ once, we immediately
  obtain the cyclic term graph $g_1$, which unravels to the full
  binary tree from Figure~\ref{fig:bintree}.

  With the rule $\rho_2$, we have to go through an $\mrs$-convergent
  reduction of length $\omega$, depicted in
  Figure~\ref{fig:doubleTransRed}, in order to obtain the desired term
  graph normal form that then unravels to the full binary tree as
  well.

  The same can also be achieved via the rule $\rho_3$: Starting from
  $h(c)$ we can construct a reduction that $\mrs$-converges directly
  to the full binary tree in Figure~\ref{fig:bintree}. However, we may
  also form the reduction shown in Figure~\ref{fig:doubleTransRedDiv}
  in which we always contract the leftmost redex. As we can see in the
  picture, this means that the $c$-node remains constantly at depth
  $2$ while still reachable from any other node. As we explained in
  Example~\ref{ex:convdiv}, this means that the reduction does not
  $\mrs$-converge. On the other hand, as described in
  Example~\ref{ex:liminf} the reduction $\prs$-converges to the
  partial term graph $g_\omega$. In fact, from this term graph
  $g_\omega$ we can then construct a reduction that $\prs$-converges
  to the full binary tree.
\end{example}

\section{Term Graph Rewriting vs.\ Term Rewriting}
\label{sec:soundn--compl}

In order to assess the value of the modes of convergence on term
graphs that we introduced in this paper, we need to compare them to
the well-established counterparts on terms. We have already observed
that, if restricted to term trees, both the partial order $\lebotr$
and the metric $\ddr$ on term graphs coincide with corresponding
structures $\lebot$ and $\dd$ on terms, cf.\
Corollary~\ref{cor:lebotrGeneralise} and
Corollary~\ref{cor:ddrGeneralise}, respectively. The same holds for
the modes of convergence derived from these structures, cf.\
Proposition~\ref{prop:liminfGeneralise} and
Proposition~\ref{prop:limGeneralise}.

\subsection{Soundness \& Completeness of Infinitary Term Graph Rewriting}
\label{sec:soundn--compl-1}

Ideally, we would like to see a strong connection between converging
reductions in a GRS $\calR$ and converging reductions in the TRS
$\unrav{\calR}$ that is its unravelling. For example, for
$\mrs$-convergence we want to see that $g \wmato[\calR] h$ implies
$\unrav g \wmato[\unrav{\calR}] \unrav h$ -- i.e.\ soundness -- and
vice versa that $\unrav g \wmato[\unrav{\calR}] \unrav h$ implies $g
\wmato[\calR] h$ -- i.e.\ completeness.

Completeness is already an issue for finitary
rewriting~\cite{barendregt87parle}: a single term graph redex may
corresponds to several term redexes due to sharing. Hence, contracting
a term graph redex may correspond to several term rewriting steps. For
example, given a rewrite rule $a \to b$, we can rewrite $f(a,a)$ to
$f(a,b)$, whereas we can rewrite
\begin{center}
  \begin{tikzpicture}
    \node (r) {$f$}%
    child{%
      node (a) {$a$}%
      edge from parent[transparent]%
    };%
    \draw[->] (r)%
    edge[bend left] (a)%
    edge[bend right] (a);%

    \node[node distance=4cm,right of=r] (s) {$f$}%
    child{%
      node (b) {$b$}%
      edge from parent[transparent]%
    };%
    \draw[->] (s)%
    edge[bend left] (b)%
    edge[bend right] (b);%
    \path ($(r)!.5!(a)$) -- ($(s)!.5!(b)$) node[pos=.5] {only to};
  \end{tikzpicture}
\end{center}
which corresponds to a term reduction $f(a,a) \to f(b,b)$. That is, in
the term graph we cannot choose which of the two term redexes to
contract as they are represented by the same term graph redex.

Note that there are techniques to circumvent this problem by
incorporating reduction steps that copy nodes in order to reduce the
sharing in a term graph~\cite{plump99hggcbgt}. In this paper, however,
we are only concerned with pure term graph rewriting steps derived
from rewrite rules.

In the context of weak convergence, also soundness becomes an
issue. The underlying reason for this issue is similar to the
phenomenon explained above: a single term graph rewrite step may
represent several term rewriting steps, i.e.\ $g \to[\calR] h$ implies
$\unrav g \fto+[\unrav\calR]\unrav h$.\footnote{If the term graph $g$
  is cyclic, the corresponding term reduction may even be infinite.}
When we have a converging term graph reduction $(g_\iota \to
g_{\iota+1})_{\iota<\alpha}$, we know that the underlying sequence of
term graphs $(g_\iota)$ converges. However, the corresponding term
reduction does not necessarily produce the sequence
$(\unrav{g_\iota})$ but may intersperse the sequence
$(\unrav{g_\iota})$ with additional intermediate terms, which might
change the convergence behaviour.

A similar phenomenon is know in infinitary lambda
calculus\cite{kennaway97tcs}: while one can simulate certain term
rewriting systems with lambda terms, this simulation may fail for
infinitary rewriting since a single term rewriting step may require
several $\beta$-reduction steps. The problem that arises in this
setting is that the intermediate terms that are introduced in the
lambda reduction may cause the convergence to break.

The same can, in principle, also occur when simulating a term graph
reduction by a term reduction. Since a single term graph rewrite step
may require several term rewrite steps, we may introduce intermediate
terms into the reduction that do not directly correspond to the term
graphs in the graph reduction.

\subsection{Preservation of Convergence under Unravelling}
\label{sec:pres-conv-under}

Due to the abovementioned difficulties, we restrict ourselves in this
paper to the soundness of the modes of convergence alone. By soundness
in this setting we mean that whenever we have a sequence
$(g_\iota)_{\iota<\alpha}$ of term graphs converging to $g$, then the
sequence $(\unrav{g_\iota})_{\iota<\alpha}$ converges to
$\unrav{g}$. That is, convergence is preserved under
unravelling. Since the metric $\ddr$ on term graphs generalises the
metric $\dd$ on terms, cf.\ Corollary~\ref{cor:ddrGeneralise}, it does
not matter whether we consider the convergence of
$(\unrav{g_\iota})_{\iota<\alpha}$ in the metric space
$(\ictgraphs,\ddr)$ or $(\iterms,\dd)$, according to
Proposition~\ref{prop:limGeneralise}. The same also holds for the
limit inferior in $(\ipctgraphs,\lebotr)$ and $(\ipterms,\lebot)$, due
to Corollary~\ref{cor:lebotrGeneralise} and
Proposition~\ref{prop:liminfGeneralise}.

The cornerstone of the investigation of the unravelling of term graphs
is the following simple characterisation of unravelling in terms of
labelled quotient trees:
\begin{proposition}
  \label{prop:unravTree}%
  The unravelling $\unrav{g}$ of a term graph $g \in \itgraphs$ is
  given by the labelled quotient tree $(\pos{g},g(\cdot),\idrel{\pos{g}})$.
\end{proposition}
\begin{proof}
  Since $\idrel{\pos{g}}$ is a subrelation of $\sim_g$, we know that
  $(\pos{g},g(\cdot),\idrel{\pos{g}})$ is a labelled quotient tree and
  thus uniquely determines a term tree $t$. By
  Lemma~\ref{lem:occrephom}, there is a homomorphism from $t$ to
  $g$. Hence, $\unrav{g} = t$.
\end{proof}

Employing the above characterisation, we can easily see that the
relation $\lebotr$ is preserved under unravelling:
\begin{proposition}
  \label{prop:poUnrav}
  Given $g,h \in \ipctgraphs$, we have that $g \lebotr h$ implies
  $\unrav{g} \lebotr \unrav h$.
\end{proposition}
\begin{proof}
  Immediate consequence of Corollary~\ref{cor:chaTgraphPo} and
  Proposition~\ref{prop:unravTree}.
\end{proof}

Likewise, also least upper bounds of $\lebotr$ are preserved:
\begin{proposition}[preservation of lubs under unravelling]
  \label{prop:lubUnrav}
  Given a directed set $G$ in $(\ipctgraphs,\lebotr)$, also
  $\setcom{\unrav{g}}{g\in G}$ is directed and $\unrav{\Lub_{g\in G}
    g} = \Lub_{g\in G} \unrav{g}$.
\end{proposition}
\begin{proof}
  The fact that $\setcom{\unrav{g}}{g\in G}$ is directed in
  $(\ipctgraphs,\lebotr)$ follows from
  Proposition~\ref{prop:poUnrav}. The equality follows from the
  characterisation of the lub in Theorem~\ref{thm:graphCpo} and from
  Proposition~\ref{prop:unravTree}.
\end{proof}

For greatest lower bounds of $\lebotr$, the situation is more
complicated as we have to consider arbitrary non-empty sets of term
graphs instead of only directed sets.

We start with the characterisation of glbs in the partially ordered
set $(\ipterms,\lebot)$ of terms. Since this partially ordered set
forms a complete semi-lattice, we know that it admits glbs of
arbitrary non-empty sets. The following lemma characterises these
glbs:
\begin{lemma}[glb on terms]
  \label{lem:glbTerms}
  The glb $\Glb T$ of a non-empty set $T$ in $(\ipterms,\lebot)$ is
  given by the labelled quotient tree $(P,l,\idrel{P})$ where
  \begin{align*}
    P &= \setcom{\pi \in {\bigcap}_{t \in T} \pos{t}}{\forall \pi' <
      \pi\exists f \in \Sigma_\bot\forall t \in T: t(\pi') = f}\\
    l(\pi) &=
    \begin{cases}
      f &\text{if } \forall t \in T: f = t(\pi)\\
      \bot &\text{otherwise}
    \end{cases}
  \end{align*}
\end{lemma}
\begin{proof}
  Special case of Proposition~5.9 in \cite{bahr12rta}.
\end{proof}

By combining the above characterisation with the characterisation of
unravelled term graphs, we obtain the following:
\begin{corollary}
  \label{cor:glbUnrav}
  Given a non-empty set $G$ in $(\ipctgraphs,\lebotr)$, the glb
  $\Glb_{g\in G} \unrav{g}$ is given by the labelled quotient tree
  $(P,l,\idrel{P})$ where
  \begin{align*}
    P &= \setcom{\pi \in {\bigcap}_{g \in G} \pos{g}}{\forall \pi' <
      \pi\exists f \in \Sigma_\bot\forall g \in G: g(\pi') = f}\\
    l(\pi) &=
    \begin{cases}
      f &\text{if } \forall g \in G: f = g(\pi)\\
      \bot &\text{otherwise}
    \end{cases}
  \end{align*}
\end{corollary}
\begin{proof}
  Follows immediately from Lemma~\ref{lem:glbTerms} and
  Proposition~\ref{prop:unravTree}.
\end{proof}

Before we deal with the preservation of glbs under unravelling, we
need the following property that relates the unravelling of a glb to
the original term graphs:
\begin{lemma}[unravelling of a glb]
  \label{lem:unravGlb}
  For each non-empty set $G$ in $(\ipctgraphs,\lebotr)$, the term $t =
  \unrav{\Glb G}$ satisfies the following for all $g\in G$ and $\pi
  \in \pos{t}$:
  \begin{center}
    \begin{inparaenum}[(i)]
    \item $\pi \in \pos{g}$ \hspace{2cm}
      \label{item:unravGlb1}
    \item $t(\pi) \in \Sigma \quad \implies \quad t(\pi) = g(\pi)$
      \label{item:unravGlb2}
    \end{inparaenum}
  \end{center}
\end{lemma}
\begin{proof}
  Let $g\in G$, $\pi \in \pos{t}$, and $h = \Glb G$. Then $\pi \in
  \pos{h}$ and $h(\pi) = t(\pi)$ according to
  Proposition~\ref{prop:unravTree}. Since $h \lebotr g$, we may apply
  Corollary~\ref{cor:chaTgraphPo} to obtain (\ref{item:unravGlb1})
  that $\pi \in \pos{g}$ and (\ref{item:unravGlb2}) that $t(\pi) =
  g(\pi)$ whenever $t(\pi) \in \Sigma$.
\end{proof}

\begin{proposition}[weak preservation of glbs under unravelling]
  \label{prop:glbUnravWeak}
  If $G$ is a non-empty set in $(\ipctgraphs,\lebotr)$, then
  $\unrav{\Glb_{g\in G} g} \lebotr \Glb_{g\in G} \unrav{g}$.
\end{proposition}
\begin{proof}
  Let $s = \unrav{\Glb_{g\in G} g}$ and $t = \Glb_{g\in G}
  \unrav{g}$. Since both $s$ and $t$ are terms, we can use the
  characterisation of $\lebot$ instead of $\lebotr$. That is, we will
  show that for each $\pi \in \pos{s}$, we have that $\pi \in \pos{t}$
  and that $t(\pi) = s(\pi)$ whenever $s(\pi) \in \Sigma$.
  
  If $\pi \in \pos{s}$, then $\pi' \in \pos{s}$ for all $\pi' <
  \pi$. Since $s(\pi')$ cannot be a nullary symbol if $\pi' < \pi$, we
  know that $s(\pi') \neq \bot$. Hence, we can apply
  Lemma~\ref{lem:unravGlb} in order to obtain for all $g\in G$ that
  $\pi \in \pos{g}$ and that $s(\pi') = g(\pi')$ for all $\pi' <
  \pi$. According to Corollary~\ref{cor:glbUnrav}, this means that
  $\pi \in \pos{t}$. In order to show the second part, assume that
  $s(\pi) \in \Sigma$. Then, by Lemma~\ref{lem:unravGlb}, $g(\pi) =
  s(\pi)$ for all $g \in G$, which, according to
  Corollary~\ref{cor:glbUnrav}, implies that $t(\pi) = s(\pi)$.
\end{proof}

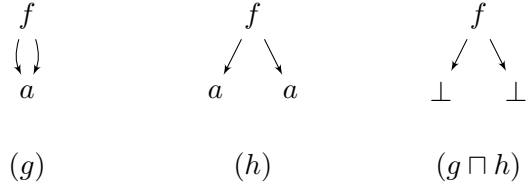
\begin{figure}
  \centering
  \begin{tikzpicture}[node distance=3cm]
    \node (g) {$f$}%
    child {%
      node (a) {$a$}%
      edge from parent[draw=none]%
    };%
    \draw[->]%
    (g) edge[bend right=20] (a)
    edge[bend left=20] (a);%

    \node[right of=g] (h) {$f$}%
    child {%
      node  {$a$}%
    } child {%
      node {$a$}%
    };%

    \node[right of=h] (b) {$f$}%
    child {%
      node  {$\bot$}%
    } child {%
      node {$\bot$}%
    };%
    
    \begin{scope}[node distance=20mm]
      \node[below of=g] {$(g)$};%
      \node[below of=h] {$(h)$};%
      \node[below of=b] {$(g \glb h)$};%
    \end{scope}
  \end{tikzpicture}
  \caption{Failure of preservation of glbs under unravelling.}
  \label{fig:counterGlbUnrav}
\end{figure}

In general, glbs are not fully preserved under unravelling as the
following example shows:
\begin{example}
  \label{ex:counterGlbUnrav}
  Consider the term graphs $g$ and $h$ in
  Figure~\ref{fig:counterGlbUnrav}. The only difference between the
  two term graphs is the sharing of the arguments of the root
  node. Due to this difference in sharing, the glb $g \glb h$ of the
  two term graphs is a proper partial term graph as depicted in
  Figure~\ref{fig:counterGlbUnrav}. On the other hand, since the
  unravelling of the two term graphs coincides, viz.\ $\unrav g =
  \unrav h = h$, we have that $\unrav g \glb \unrav h = h$. In
  particular, we have the strict inequality $\unrav{g \glb h} \lbotr
  \unrav g \glb \unrav h$.
\end{example}

Unfortunately, this also means that the limit inferior is only weakly
preserved under unravelling as well:
\begin{theorem}
  \label{thr:liminfUnrav}
  For every sequence $(g_\iota)_{\iota<\alpha}$ in
  $(\ipctgraphs,\lebotr)$, we have that
  \[
  \unrav{\liminf_{\iota \limto \alpha} g_\iota} \lebotr \liminf_{\iota
    \limto \alpha} \unrav{g_\iota}.
  \]
\end{theorem}
\begin{proof}
  This follows from Proposition~\ref{prop:glbUnravWeak} and
  Proposition~\ref{prop:lubUnrav}.
\end{proof}

Again, we can construct a counterexample that shows that the converse
inequality does not hold in general:
\begin{example}
  \label{ex:counterLiminfUnrav}
  Let $(g_\iota)_{\iota<\omega}$ be the sequence alternating between
  $g$ and $h$ from Figure~\ref{fig:counterGlbUnrav}, i.e.\ $g_{2\iota}
  = g$ and $g_{2\iota+1} = h$ for all $\iota<\omega$. Then
  $\Glb_{\alpha \le \iota<\omega} g_\iota = g \glb h$ for each $\alpha
  < \omega$ and, consequently, $\liminf_{\iota\limto\omega}g_\iota = g
  \glb h$. As we have seen in Example~\ref{ex:counterGlbUnrav}, $g
  \glb h$ is the proper partial term graph depicted in
  Figure~\ref{fig:counterGlbUnrav}. On the other hand, since $\unrav g
  = \unrav h = h$, we have that
  $\liminf_{\iota\limto\omega}\unrav{g_\iota} = h$. In particular, we
  have the strict inequality
  $\unrav{\liminf_{\iota\limto\omega}g_\iota} \lbotr
  \liminf_{\iota\limto\omega}\unrav{g_\iota}$.
\end{example}

Moreover, we cannot expect that any other partial order with
properties comparable to those of $\lebotr$ fully preserves the limit
inferior under unravelling.

The example above shows that any partial order $\le$ on partial term
graphs whose limit inferior is preserved under unravelling must also
satisfy either $g \le h$ or $h \le g$ for the term graphs in
Figure~\ref{fig:counterGlbUnrav}. That is, such a partial order has to
give up the property that total term graphs are maximal, cf.\
Proposition~\ref{prop:nonPartMax}. This observation is independent of
whether this partial order specialises to $\lebot$ on terms.

The sacrifice for full preservation under unravelling goes even
further. If a partial order $\le$ on partial term graphs satisfies
preservation of its limit inferior under unravelling, the limit
inferior $\liminf_{\iota\limto\omega}g_\iota$ of the sequence
$(g_\iota)_{\iota<\omega}$ from Example~\ref{ex:counterLiminfUnrav}
has to unravel to $h$, a total term. That is,
$\liminf_{\iota\limto\omega}g_\iota$ has to be a total term graph. On
the other hand, there is no metric -- or any Hausdorff topology for
that matter -- for which $(g_\iota)_{\iota<\omega}$ converges at all
because $(g_\iota)_{\iota<\omega}$ alternates between two distinct
term graphs. In other words, the correspondence between $\mrs$- and
$\prs$-convergence, which we have for $\lebotr$ as stated in
Theorem~\ref{thr:graphTotalConv}, cannot be satisfied for such a
partial order $\le$, regardless of the metric on term graphs.

The simple partial order $\lebots$, which we briefly discussed in
comparison to the rigid partial order $\lebotr$ in
Section~\ref{sec:partial-order-lebot1}, takes the other side of the
trade-off illustrated above: it satisfies $g \lebots h$ and the
preservation of the limit inferior under unravelling but sacrifices
the correspondence between total term graphs and maximality as well as
the correspondence between $\mrs$- and
$\prs$-convergence~\cite{bahr12rta}.

Using the correspondence between the limit inferior in
$(\ipctgraphs,\lebotr)$ and the limit in $(\ictgraphs,\ddr)$, we can
derive full preservation of limits under unravelling:
\begin{theorem}
  \label{thr:limUnrav}
  For every convergent sequence $(g_\iota)_{\iota<\alpha}$ in
  $(\ictgraphs,\ddr)$, also $(\unrav{g}_\iota)_{\iota<\alpha}$ is
  convergent and
  \[
  \unrav{\lim_{\iota \limto \alpha} g_\iota} = \lim_{\iota \limto
    \alpha} \unrav{g_\iota}.
  \]
\end{theorem}
\proof
  We prove the equality as follows:
  \[
  \unrav{\lim_{\iota\limto\alpha} g_\iota} \stackrel{(1)}=
  \unrav{\liminf_{\iota\limto\alpha} g_\iota} \stackrel{(2)}=
  \liminf_{\iota\limto\alpha}\unrav{g_\iota}
  \stackrel{(3)}=\lim_{\iota\limto\alpha}\unrav{g_\iota}
  \]
  \begin{enumerate}[(1)]
  \item Since $(g_\iota)_{\iota<\alpha}$ is convergent, and thus
    Cauchy, we can apply Proposition~\ref{prop:graphLimLimInf} to
    obtain that $\lim_{\iota\limto\alpha} g_\iota =
    \liminf_{\iota\limto\alpha} g_\iota$.
  \item Since $\liminf_{\iota\limto\alpha} g_\iota$ is total, so is
    $\unrav{\liminf_{\iota\limto\alpha} g_\iota}$. By
    Proposition~\ref{prop:nonPartMax}, this means that
    $\unrav{\liminf_{\iota\limto\alpha} g_\iota}$ is maximal w.r.t.\
    $\lebotr$. Consequently, the inequality $\unrav{\liminf_{\iota
        \limto \alpha} g_\iota} \lebotr \liminf_{\iota \limto \alpha}
    \unrav{g_\iota}$ due to Theorem~\ref{thr:liminfUnrav} yields (2).
  \item This equality follows from
    Proposition~\ref{prop:graphLimInfLim} and the totality of %fact that
    $\liminf_{\iota \limto \alpha} \unrav{g_\iota}$.\qed% is total.
  \end{enumerate}

\subsection{Finite Representations of Transfinite Term Reductions}
\label{sec:rational-terms-}

One of the motivations for considering modes of convergence on term
graphs in the first place is the study of finite representation of
transfinite term reductions as finite term graph reductions. Since
both the metric $\ddr$ and the partial order $\lebotr$ specialise to
the corresponding structures on terms, we can use both the metric
space $(\ictgraphs,\ddr)$ and the partially ordered set
$(\ipctgraphs,\lebotr)$ to move seamlessly from terms to term graphs
and vice versa.

For instance, Example~\ref{ex:duplication} illustrates reductions that
perform essentially the same computations, however, at different
levels of sharing / parallelism. This includes the complete lack of
sharing as well, i.e.\ term rewriting. For each of the cases we can
use the partially ordered set $(\ipctgraphs,\lebotr)$ and the metric
space $(\ictgraphs,\ddr)$ as a unifying framework to determine the
convergence behaviour.

In order to use the partial order $\lebotr$ and the metric $\ddr$ as a
tool to study finite representability of infinite term reductions as
finite term graph reductions there still is some work to be done,
though.

First and foremost, we need a unifying framework for performing both
term and term graph rewriting. A straightforward approach to achieve
this, is to include copying steps in term graph reductions that allow
us to revert the sharing produced by applying term graph
rules~\cite{plump99hggcbgt}. For example while the rule $\rho_3$ from
Figure~\ref{fig:doubleRules} is the term graph rule with the least
sharing that unravels to $\rho\fcolon h(x) \rightarrow f(h(x),h(x))$,
it still has some sharing in order to represent the duplication of the
variable $x$ on the right-hand side.

The result of this sharing is seen in
Figure~\ref{fig:doubleTransRedDiv}, which shows that even if we start
with a term tree, the rule $\rho_3$ turns it into a proper term
graph. Consequently it is slightly different from the corresponding
term reduction
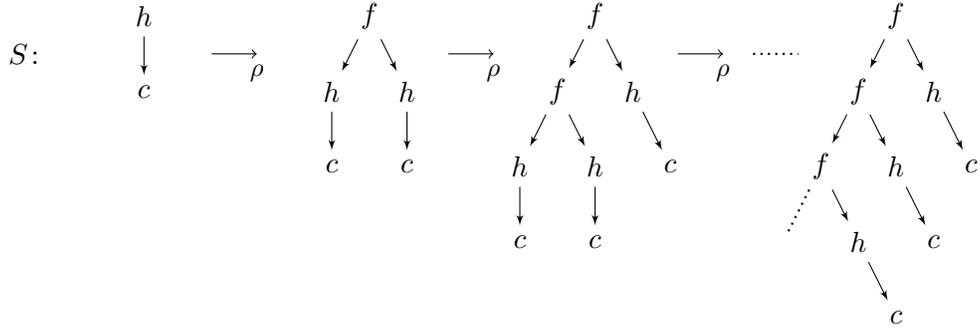
\begin{center}
  \begin{tikzpicture}
        
    \node (g1) {$h$}%
    child {%
      node {$c$} };%

    \node [node distance=2.5cm,right=of g1] (g2) {$f$}%
    child {%
      node (h1) {$h$}%
      child {%
        node (c) {$c$}%
      }%
    } child {%
      node (h2) {$h$}%
      child {%
        node (c) {$c$}%
      }%
    };%
         
    \node [node distance=2.5cm,right=of g2] (g3) {$f$}%
    child {%
      node (f) {$f$}%
      child {%
        node (h1) {$h$}%
        child {%
          node (c2) {$c$}%
        }%
      } child {%
        node (h2) {$h$}%
        child {%
          node (c2) {$c$}%
        }%
      }%
    } child {%
      node (h3) {$h$}%
      child [missing]%
      child {%
        node {$c$}%
      }%
    };%

    \node [node distance=3.5cm,right=of g3] (go) {$f$}%
    child {%
      node (f1) {$f$}%
      child {%
        node (f2) {$f$}%
        child [etc] {%
          node {}%
        } child {%
          node (h1) {$h$}%
          child [missing]%
          child {%
            node {$c$}%
          }%
        }%
      } child {%
        node (h2) {$h$}%
        child [missing]%
        child {%
          node {$c$}%
        }%
      } } child {%
      node (h3) {$h$}%
      child [missing]%
      child {%
        node {$c$}%
      }%
    };%

    \node (S) at ($(g1)-(1.5,.5)$) {$S\fcolon$};%
    \node (s1) at ($(g1)!.4!(g2)-(0,.5)$) {};%
    \node (s2) at ($(g2)!.45!(g3)-(0,.5)$) {};%
    \node (s3) at ($(g3)!.35!(go)-(0,.5)$) {};%
    \node (s4) at ($(g3)!.6!(go)-(0,.5)$) {};%

    \draw[single step] ($(s1)-(.3,0)$) -- ($(s1)+(.3,0)$)%
    node[pos=1,below] {{\small$\rho$}};%
    \draw[single step] ($(s2)-(.3,0)$) -- ($(s2)+(.3,0)$)%
    node[pos=1,below] {{\small$\rho$}};%
    \draw[single step] ($(s3)-(.3,0)$) -- ($(s3)+(.3,0)$)%
    node[pos=1,below] {{\small$\rho$}};%
    \draw[dotted,thick,-] ($(s4)-(.3,0)$) -- ($(s4)+(.3,0)$);%
        
  \end{tikzpicture}
\end{center}
In fact, while the above term reduction $S$ is $\mrs$-convergent, the
term graph reduction via $\rho_3$, depicted in
Figure~\ref{fig:doubleTransRedDiv}, is not.

However, by interspersing the term graph reduction with reduction
steps that copy nodes -- and in general sub-term graphs -- we instead
obtain the following term graph reduction:
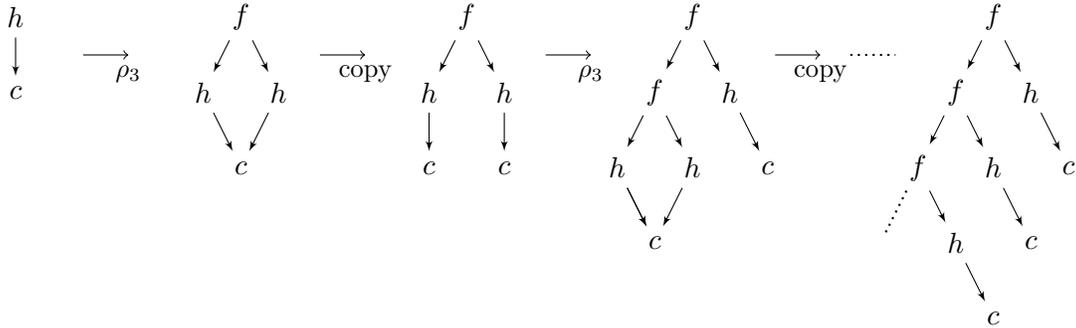
\begin{center}
  \begin{tikzpicture}
        
    \node (g1) {$h$}%
    child {%
      node {$c$} };%
        
    \node [node distance=2.5cm,right=of g1] (g2) {$f$}%
    child {%
      node (h1) {$h$}%
      child [missing]%
      child {%
        node (c) {$c$} } } child {%
      node (h2) {$h$}%
    };%

    \draw[->] (h2) edge (c);%

    \node [node distance=2.5cm,right=of g2] (g2c) {$f$}%
    child {%
      node (h1) {$h$}%
      child {%
        node (c) {$c$}%
      }%
    } child {%
      node (h2) {$h$}%
      child {%
        node (c) {$c$}%
      }%
    };%
         
    \node [node distance=2.5cm,right=of g2c] (g3) {$f$}%
    child {%
      node (f) {$f$}%
      child {%
        node (h1) {$h$}%
        child [missing]%
        child {%
          node (c2) {$c$}%
        }%
      } child {%
        node (h2) {$h$}%
      }%
    } child {%
      node (h3) {$h$}%
      child [missing]%
      child {%
        node (c) {$c$}%
      }%
    };%

    \draw[->]%
    (h1) edge (c2)%
    (h2) edge (c2);%

    \node [node distance=3.5cm,right=of g3] (go) {$f$}%
    child {%
      node (f1) {$f$}%
      child {%
        node (f2) {$f$}%
        child [etc] {%
          node {}%
        } child {%
          node (h1) {$h$}%
          child [missing]%
          child {%
            node {$c$}%
          }%
        }%
      } child {%
        node (h2) {$h$}%
        child [missing]%
        child {%
          node {$c$}%
        }%
      } } child {%
      node (h3) {$h$}%
      child [missing]%
      child {%
        node {$c$}%
      }%
    };%

    \node (s1) at ($(g1)!.4!(g2)-(0,.5)$) {};%
    \node (s2) at ($(g2)!.45!(g2c)-(0,.5)$) {};%
    \node (s3) at ($(g2c)!.45!(g3)-(0,.5)$) {};%
    \node (s4) at ($(g3)!.35!(go)-(0,.5)$) {};%
    \node (s5) at ($(g3)!.6!(go)-(0,.5)$) {};%

    \draw[single step] ($(s1)-(.3,0)$) -- ($(s1)+(.3,0)$)%
    node[pos=1,below] {{\small$\rho_3$}};%
    \draw[single step] ($(s2)-(.3,0)$) -- ($(s2)+(.3,0)$)%
    node[pos=1,below] {{\small copy}};%
    \draw[single step] ($(s3)-(.3,0)$) -- ($(s3)+(.3,0)$)%
    node[pos=1,below] {{\small$\rho_3$}};%
    \draw[single step] ($(s4)-(.3,0)$) -- ($(s4)+(.3,0)$)%
    node[pos=1,below] {{\small copy}};%
    \draw[dotted,thick,-] ($(s5)-(.3,0)$) -- ($(s5)+(.3,0)$);%
        
  \end{tikzpicture}
\end{center}
This reduction simulates the corresponding term reduction $S$ more
closely and in fact both reductions $\mrs$-converge to the same
term. Nevertheless, this approach creates the same issue that we have
already noted for soundness: the additional term graphs that get
interspersed with the original term reduction may affect the
convergence behaviour.

The second ingredient that we need in order to study the finite
representability of transfinite term reductions is a compression
property~\cite{kennaway95ic} for transfinite term graph reductions
that allows us to compress a transfinite term graph reduction that
ends in a finite term graph to a term graph reduction of finite
length.

Unfortunately, experience from infinitary term rewriting already shows
us that a general compression property -- allowing any reduction to be
compressed to length at most $\omega$ -- is not possible for weak
convergence~\cite{kennaway03book}. However, the more restrictive
version of the compression property that we need, viz.\ that
reductions ending in a finite term graph may be compressed to finite
length, does hold for weakly $\mrs$-converging term
reductions~\cite{lucas01rta} and there is hope that this carries over
to the term graph rewriting setting.

\section{Conclusions \& Future Work}
\label{sec:conclusions}

With the goal of generalising infinitary term rewriting to term
graphs, we have presented two different modes of convergence for an
infinitary calculus of term graph rewriting. The success of this
generalisation effort was demonstrated by a number of results. Many of
the properties of the modes of convergence on terms have been
maintained in this transition to term graphs. First and foremost, this
includes the intrinsic completeness properties of the underlying
structures, i.e.\ the metric space is still complete and the partially
ordered set still forms a complete semilattice. Moreover, we were also
able to maintain the correspondence between $\prs$- and
$\mrs$-convergence as well as the intuition of the partial order to
capture a notion of information preservation.

An important check for the appropriateness of the modes of convergence
on term graphs is their relation to the corresponding modes of
convergence on terms. For both the partial order and the metric
approach, we have that convergence on term graphs is a conservative
extension of the convergence on terms. Conversely, convergence on term
graphs carries over to convergence on terms via the unravelling
mapping. Unfortunately, this preservation of convergence under
unravelling is only weak in the case of the partial order setting;
cf.\ Theorem~\ref{thr:liminfUnrav}. However, as we have explained in
Section~\ref{sec:pres-conv-under}, this phenomenon is an unavoidable
side effect of the generalisation to term graphs unless other
important properties are sacrificed. Fortunately, this phenomenon
vanishes in the metric setting and we in fact obtain full preservation
of limits under unravelling; cf.\ Theorem~\ref{thr:limUnrav}.

As a result, we have obtained two modes of convergence, which allow us
to combine both infinitary term rewriting and term graph rewriting
within one theoretical framework. Our motivation for this effort is
derived from studying \emph{lazy evaluation} and the correspondence
between infinitary term rewriting and finitary term graph
rewriting. For both applications, we still require more understanding
of the matter, though: for the former, we still lack at least a
treatment of higher-order rewriting whereas we are much closer to the
latter. We have discussed issues concerning the correspondence between
infinitary term rewriting and finitary term graph rewriting in detail
in Section~\ref{sec:rational-terms-}: while the unified modes of
convergence are already helpful for studying infinitary rewriting with
a varying degree of sharing, we identified two shortcomings that have
to be addressed, viz.\ the lack of a unifying notion of rewriting for
terms and term graphs and a compression property for transfinite term
graph reductions.

Apart from the abovementioned issues, future work should also be
concerned with establishing a stronger correspondence between
infinitary term rewriting and infinitary term graph rewriting beyond
the preservation of limits under unravellings, which we showed in this
paper. Despite the difficulties that we encountered in
Section~\ref{sec:soundn--compl-1}, we think that obtaining such
results is possible. However, a more promising way of approaching this
issue is to restrict the notion of convergence to strong convergence
as we know it from infinitary term rewriting~\cite{kennaway95ic}. Such
a stricter notion of convergence takes the location of a reduction
step into consideration and, thus, provides a closer correspondence
between term graph reductions and their term rewriting
counterparts. Indeed, this technique has been applied successfully to
convergence on term graphs based on the simple partial order
$\lebots$, which we briefly discussed in comparison to the rigid
partial order $\lebotr$ in Section~\ref{sec:partial-order-lebot1}, and
a corresponding metric~\cite{bahr12rta}.

\section*{Acknowledgement}
The author would like to thank the anonymous referees of RTA 2011 as
well as the referees for the special issue of LMCS whose comments
greatly helped to improve the presentation of the material.

\bibliographystyle{plain}
\bibliography{itgr}

\appendix

\makeatletter
\providecommand\@dotsep{5}
\makeatother
\listoftodos\relax

\section{Proof of Lemma~\ref{lem:graphCompLub}}
\label{sec:proof-lemma-refl-1}

\theoremstyle{plain}\newtheorem*{lemGraphCompLub}{Lemma~\ref{lem:graphCompLub}}
\begin{lemGraphCompLub}[compatible elements have lub]
  % compatible elements have lub %
  Each pair $g_1,g_2$ of compatible term graphs in
  $(\ipctgraphs,\lebotr)$ has a least upper bound.
\end{lemGraphCompLub}
\begin{proof}[Proof of Lemma~\ref{lem:graphCompLub}]
  Since $\set{g_1,g_2}$ is not necessarily directed, its lub might
  have positions that are neither in $g_1$ nor in $g_2$. Therefore, it
  is easier to employ a different construction here: Following
  Remark~\ref{rem:lebot}, we will use the structure
  $(\quotient{\iptgraphs}{\isom},\lebotr)$ which is isomorphic to
  $(\ipctgraphs,\lebotr)$. To this end, we will construct a term graph
  $\ol g$ such that $\eqc{\ol g}{\isom}$ is the lub of
  $\set{\eqc{g_1}{\isom},\eqc{g_2}{\isom}}$. Since we assume that
  $\set{\eqc{g_1}{\isom},\eqc{g_2}{\isom}}$ has an upper bound, say
  $\eqc{\oh g}{\isom}$, there are two rigid $\bot$-homomorphisms
  $\phi_i\fcolon g_i \homto_\bot \oh g$.

  Let $g_j = (N^j,\gsuc^j,\glab^j,r^j)$, $j = 1,2$. Since we are
  dealing with isomorphism classes, we can assume w.l.o.g.\ that the
  nodes in $g_j$ are of the form $n^j$ for $j = 1,2$. Let $\ol M = N^1
  \uplus N^2$ and define the relation $\sim$ on $\ol M$ as follows:
  \[
  n^j \sim m^k \quad \text{iff} \quad \nodePos{g_j}{n^j} \cap
  \nodePos{g_k}{m^k} \neq \emptyset
  \]
  $\sim$ is clearly reflexive and symmetric. Hence, its transitive
  closure $\sim^+$ is an equivalence relation on $\ol M$.  Now define
  the term graph $\ol g = (\ol N, \ol \glab, \ol \gsuc, \ol r)$ as
  follows:
  \begin{align*}
    \ol N &= \quotient{\ol M}{\sim^+} &%
    \ol\glab(N) &=
    \begin{cases}
      f & \text{if } f \in\Sigma, \exists n^j \in N.\;\; \glab^j(n^j)
      = f \\
      \bot & \text{otherwise}
    \end{cases}\\%
    \ol r &= \eqc{r^1}{\sim^+} &%
    \ol \gsuc_i(N) &= N' \quad \text{iff} \quad \exists n^j \in N.\;\;
    \gsuc^j_i(n^j)\in N'
  \end{align*}
  Note that since $\emptyseq \in \nodePos{g_1}{r^1} \cap
  \nodePos{g_2}{r^2}$, we also have $\ol r = \eqc{r^2}{\sim^+}$.

  Before we argue about the well-definedness of $\ol g$, we need to
  establish some auxiliary claims:
  \begin{align*}
    n^j \sim^+ m^k \quad &\implies \quad \phi_j(n^j) = \phi_k(m^k)
    \qquad &&\text{for all } n^j,m^k \in \ol M
    \tag{1} \label{eq:graphCompLuba} \\
    \phi_j(n^j) = \phi_k(m^k) \quad &\implies \quad n^j \sim m^k
    \qquad
    &&\begin{aligned}
      &\text{for all } n^j,m^k \in \ol M\\[-5pt]
      &\text{with }
      \glab^j(n^j),\glab^k(m^k)\in\Sigma
    \end{aligned}
    \tag{1'} \label{eq:graphCompLubap}
  \end{align*}
  \def\claima{(\ref{eq:graphCompLuba})}
  \def\claimap{(\ref{eq:graphCompLubap})}
  
\noindent
  We show \claima{} by proving that $n^j \sim^p m^k$ implies
  $\phi_j(n^j) = \phi_k(m^k)$ by induction on $p > 0$. If $p = 1$,
  then $n^j \sim m^k$. Hence, $\nodePos{g_j}{n^j} \cap
  \nodePos{g_k}{m^k} \neq \emptyset$. Additionally, from
  Lemma~\ref{lem:canhom} we obtain both $\nodePos{g_j}{n^j} \subseteq
  \nodePos{\oh g}{\phi_j(n^j)}$ and $\nodePos{g_k}{m^k} \subseteq
  \nodePos{\oh g}{\phi_k(m^k)}$. Consequently, we also have that
  $\nodePos{\oh g}{\phi_j(n^j)} \cap \nodePos{\oh g}{\phi_k(m^k)} \neq
  \emptyset$, i.e.\ $\phi_j(n^j) = \phi_k(m^k)$. If $p = q+1 > 1$, then
  there is some $o^l \in \ol M$ with $n^j \sim o^l$ and $o^l \sim^q
  m^k$. Applying the induction hypothesis immediately yields
  $\phi_j(n^j) = \phi_l(o^l) = \phi_k(m^k)$.

  For \claimap{}, let $n^j,m^k \in \ol M$ with
  $\glab^j(n^j),\glab^k(m^k)\in\Sigma$ and $\phi_j(n^j) =
  \phi_k(m^k)$. Since $\phi_j$ and $\phi_k$ are rigid
  $\bot$-homomorphisms, we have the following equations:
  \[
  \nodePosAcy{g_j}{n^j} =
  \nodePosAcy{\oh g}{\phi_j(n^j)} = \nodePosAcy{\oh g}{\phi_k(m^k)} =
  \nodePosAcy{g_k}{m^k}.
  \]
  Hence, $\nodePos{g_j}{n^j} \cap \nodePos{g_k}{m^k} \neq \emptyset$
  and, therefore, $n^j \sim m^k$.

  Next we show that $\ol \glab$ is well-defined. To this end, let $N
  \in \ol N$ and $n^j,m^k \in N$ such that $\glab^j(n^j) = f_1
  \in\Sigma$ and $\glab^k(m^k) = f_2 \in\Sigma$. We need to show that
  $f_1 = f_2$. By \claima{}, we have that $\phi_j(n^j) =
  \phi_k(m^k)$. Since $f_1,f_2 \in\Sigma$, we can employ the labelling
  condition for $\phi_j$ and $\phi_k$ in order to obtain that
  \[
  f_1 = \glab^j(n^j) = \oh\glab(\phi_j(n^j)) = \oh\glab(\phi_k(m^k)) =
  \glab^k(m^k) = f_2.
  \]

  To argue that $\ol\gsuc$ is well-defined, we first have to show for
  all $N \in \ol N$ that $\ol\gsuc_i(N)$ is defined iff $i <
  \srank{\ol\glab(N)}$. Suppose that $\ol\gsuc_i(N)$ is defined. Then
  there is some $n^j\in N$ such that $\gsuc^j_i(n^j)$ is
  defined. Hence, $i < \srank{\glab^j(n^j)}$. Since then also
  $\glab^j(n^j) \in\Sigma$, we have $\ol\glab(N) =
  \glab^j(n^j)$. Therefore, $i < \srank{\ol\glab(N)}$. If, conversely,
  there is some $i\in\nats$ with $i < \srank{\ol \glab(N)}$, then we know that
  $\ol\glab(N) = f \in\Sigma$. Hence, there is some $n^j \in N$ with
  $\glab^j(n^j) = f$. Hence, $i < \srank{\glab^j(n^j)}$ and,
  therefore, $\gsuc^j_i(n^j)$ is defined. Hence, $\ol\gsuc_i(N)$ is
  defined.

  To finish the argument showing that $\ol \gsuc$ is well-defined, we
  have to show that, for all $N,N_1,N_2 \in \ol N$ and $n^j, m^k \in N$
  such that $\gsuc^j_i(n^j) \in N_1$ and $\gsuc^k_i(m^k) \in N_2$, we
  indeed have $N_1 = N_2$. As $n^j,m^k \in N$, we have $n^j \sim^+ m^k$
  and, therefore, $\phi_j(n^j) = \phi_k(n^k)$ according to
  \claima. Since both $\gsuc^j_i(n^j)$ and $\gsuc^k_i(m^k)$ are
  defined, we have $\glab^j(n^j),\glab^k(m^k) \in\Sigma$.  By \claimap\
  we then have $n^j \sim m^k$, i.e.\ there is some $\pi \in
  \nodePos{g_j}{n^j} \cap \nodePos{g_k}{m^k}$. Consequently,
  $\pi\concat\seq i \in \nodePos{g_j}{\gsuc^j_i(n^j)} \cap
  \nodePos{g_k}{\gsuc^k_i(m^k)}$. Hence, $\gsuc^j_i(n^j) \sim
  \gsuc^k_i(m^k)$ and, therefore, $N_1 = N_2$.

  Before we begin the main argument we need establish the following
  auxiliary claims:
  \begin{align*}
    \nodePos{g_j}{n^j} \subseteq \nodePos{\ol g}{\eqc{n^j}{\sim^+}}
    \qquad &\text{for all } n^j \in \ol M
    \tag{2} \label{eq:graphCompLubb}\\
    \forall \pi \in \nodePosAcy{\ol g}{N} \; \exists n^j \in N.\;\;
    \glab^j(n^j) \in\Sigma, \pi \in \nodePosAcy{g_j}{n^j} \qquad
    &\text{for all } N \in \ol N \text{ with } \ol\glab(N) \in\Sigma \tag{3} \label{eq:graphCompLubc}\\
    n^j \sim^+ m^k \quad \implies \quad \nodePosAcy{g_j}{n^j} =
    \nodePosAcy{g_k}{m^k} \quad &
    \begin{aligned}
      &\text{for all } n^j,m^j \in \ol M \\[-5pt]
      &\text{with } \glab^j(n^j),
      \glab^k(m^k) \in\Sigma
    \end{aligned}
    \tag{4} \label{eq:graphCompLubd}
  \end{align*}
  \def\claimb{(\ref{eq:graphCompLubb})}
  \def\claimc{(\ref{eq:graphCompLubc})}
  \def\claimd{(\ref{eq:graphCompLubd})}
  
  For \claimb{}, we will show that $\pi \in \nodePos{g_j}{n^j}$
  implies $\pi \in \nodePos{\ol g}{\eqc{n^j}{\sim^+}}$ by induction on
  the length of $\pi$. The case $\pi = \emptyseq$ is trivial. If $\pi
  = \pi' \concat \seq i$, then $\pi' \concat \seq i \in
  \nodePos{g_j}{n^j}$, i.e.\ for $m^j = \nodeAtPos{g_j}{\pi'}$, we
  have $\gsuc^j_i(m^j) = n^j$. Employing the induction hypothesis, we
  obtain $\pi' \in \nodePos{\ol g}{\eqc{m^j}{\sim^+}}.$ Additionally,
  according to the construction of $\ol g$, we have
  $\ol\gsuc_i(\eqc{m^j}{\sim^+}) = \eqc{n^j}{\sim^+}$. Consequently,
  $\pi'\concat \seq i \in \nodePos{\ol g}{\eqc{n^j}{\sim^+}}$ holds.

  Similarly, we also show \claimc{} by induction on the length of
  $\pi$. If $\pi = \emptyseq$, then we have $\emptyseq \in
  \nodePosAcy{\ol g}{N}$, i.e.\ $N = \ol r$. Since, by assumption,
  $\ol \glab(\ol r) \in\Sigma$ holds, there is some $j \in \set{1,2}$
  such that $\glab^j(r^j) \in\Sigma$. Moreover, we clearly have
  $\emptyseq \in \nodePosAcy{g_j}{r^j}$. If $\pi = \pi' \concat \seq
  i$, then we have $\pi' \concat \seq i \in \nodePosAcy{\ol
    g}{N}$. Let $N' = \nodeAtPos{\ol g}{\pi'}$. Since $\pi' \concat
  \seq i$ is acyclic in $\ol g$, so is $\pi'$, i.e.\ $\pi' \in
  \nodePosAcy{\ol g}{N'}$. Moreover, we have that $\ol \gsuc_i(N')$ is
  defined, i.e.\ $\ol \glab(N')$ is not nullary and in particular not
  $\bot$. Thus, we can apply the induction hypothesis to obtain some
  $n^j \in N'$ with $\glab^j(n^j) \in\Sigma$ and $\pi' \in
  \nodePosAcy{g_j}{n^j}$. Hence, according to the construction of $\ol
  g$, we have $\glab^j(n^j) = \ol\glab(N')$, i.e.\ $\gsuc^j_i(n^j) =
  m^j$ is defined. Furthermore, we then get $m^j \in N$. Since $\pi'
  \concat\seq i \in \nodePos{g_j}{m^j}$, it remains to be shown that
  $\pi' \concat\seq i$ is acyclic in $g_j$. Suppose that $\pi'
  \concat\seq i$ is cyclic in $g_j$. As $\pi'$ is acyclic in $g_j$,
  this means that there is some position $\pi^* \in
  \nodePos{g_j}{m^j}$ with $\pi^* < \pi' \concat\seq i$. Using
  \claimb{}, we obtain that $\pi^* \in \nodePos{\ol g}{N}$. This
  contradicts the assumption of $\pi' \concat\seq i$ being acyclic in
  $\ol g$. Hence, $\pi' \concat\seq i$ is acyclic.

  For \claimd{}, suppose that $n^j \sim^+ m^k$ holds with
  $\glab^j(n^j), \glab^k(m^k) \in\Sigma$. From \claima{}, we obtain
  $\phi_j(n^j) = \phi_k(n^k)$. Moreover, since both $n^j$ and $m^k$
  are not labelled with $\bot$, we know that $\phi_j$ and $\phi_k$ are
  rigid in $n^j$ and $m^k$, respectively, which yields the equations
  \[
  \nodePosAcy{g_j}{n^j} = \nodePosAcy{\oh g}{\phi_j(n^j)} =
  \nodePosAcy{\oh g}{\phi_k(m^k)} = \nodePosAcy{g_k}{m^k}.
  \]

  Next we show that $\eqc{g_1}{\isom}, \eqc{g_1}{\isom} \lebotr
  \eqc{\ol g}{\isom}$ holds by giving two rigid $\bot$-homomorphisms
  $\psi_j\fcolon g_j \homto_\bot \ol g$, $j=1,2$. Define
  $\psi_j\fcolon N^j \funto \ol N$ by $n^j \mapsto
  \eqc{n^j}{\sim^+}$. From \claimb{} and the fact that, according to
  the construction of $\ol g$, $\glab^j(n^j) \in\Sigma$ implies
  $\glab^j(n^j) = \ol\glab(\eqc{n^j}{\sim^+})$, we immediately get
  that $\psi_j$ is a $\bot$-homomorphism by applying
  Lemma~\ref{lem:canhom}. In order to argue that $\psi_j$ is rigid,
  assume that $n^j \in N^j$ with $\glab^j(n^j) \in\Sigma$. According
  to Lemma~\ref{lem:presShar}, it suffices to show that
  $\nodePosAcy{\ol g}{\psi_j(n^j)} \subseteq
  \nodePos{g_j}{n^j}$. Suppose that $\pi \in \nodePosAcy{\ol
    g}{\psi_j(n^j)}$.  Note that, by the construction of $\ol g$,
  $\psi_j(n^j)$ is not labelled with $\bot$ either. Hence, we can
  apply \claimc\ to obtain some $m^k \in \psi_j(n^j)$ with
  $\glab^k(m^k)\in\Sigma$ and $\pi \in \nodePosAcy{g_k}{m^k}$. By
  definition, $m^k \in \psi_j(n^j)$ is equivalent to $n^j \sim^+
  m^k$. Therefore, we can employ \claimd{}, which yields
  $\nodePosAcy{g_k}{m^k} = \nodePosAcy{g_j}{n^j}$. Hence, $\pi \in
  \nodePosAcy{g_j}{n^j}$.

  Note that the construction of $\ol g$ did not depend on $\oh g$,
  viz.\ for any other upper bound $\eqc{\oh h}{\isom}$ of
  $\eqc{g_1}{\isom}, \eqc{g_2}{\isom}$, we get the same term graph
  $\ol g$. Hence, it is still just an arbitrary upper bound which
  means that in order to show that $\eqc{\ol g}{\isom}$ is the least
  upper bound, it suffices to show $\eqc{\ol g}{\isom} \lebotr
  \eqc{\oh g}{\isom}$. For this purpose, we will devise a rigid
  $\bot$-homomorphism $\psi\fcolon \ol g \homto_\bot \oh g$. Define
  $\psi\fcolon \ol N \funto \oh N$ by $\eqc{n^j}{\sim^+} \mapsto
  \phi_j(n^j)$. \claima\ shows that $\psi$ is well-defined. The root
  condition for $\psi$ follows from the root condition for $\phi_1$:
  \[
  \psi(\ol r) = \psi(\eqc{r^1}{\sim^+}) = \phi_1(r^1) = \oh r.
  \]

  For the labelling condition, assume that $\ol\glab(N) = f \in\Sigma$
  for some $N \in \ol N$. Then there is some $n^j \in N$ with
  $\glab^j(n^j) = f$. Therefore, the labelling condition for $\phi_j$
  yields 
  \[
  \oh \glab(\psi(N)) = \oh\glab(\phi_j(n^j)) = \glab^j(n^j) = f
  \]

  For the successor condition, let $\ol \gsuc_i(N) = N'$ for some $N,N'
  \in \ol N$. Then there is some $n^j \in N$ with $\gsuc^j_i(n^j) \in
  N'$. Therefore, the successor condition for $\psi$ follows from the
  successor condition for $\phi_j$ as follows:
  \begin{align*}
    \psi(\ol \gsuc_i(N)) &= \psi(N') = \psi (
    \eqc{\gsuc^j_i(n^j)}{\sim^+}) = \phi_j(\gsuc^j_i(n^j))\\
    &= \oh \gsuc_i(\phi_j(n^j)) = \oh
    \gsuc_i(\psi(\eqc{n^j}{\sim^+})) = \oh \gsuc_i(\psi(N))
  \end{align*}

  Finally, we show that $\psi$ is rigid. To this end, let $N \in \ol
  N$ with $\ol \glab(N) \in\Sigma$. That is, there is some $n^j \in N$
  with $\glab^j(n^j) \in\Sigma$. Recall, that we have shown that
  $\psi_j\fcolon g^j \homto_\bot \ol g$ is rigid. That is, we have 
  \[
  \nodePosAcy{g_j}{n^j} = \nodePosAcy{\ol g}{\psi_j(n^j)} =
  \nodePosAcy{\ol g}{\eqc{n^j}{\sim^+}}.
  \]
  Analogously, we have $\nodePosAcy{\oh g}{\phi_j(n^j)} =
  \nodePosAcy{g_j}{n^j}$ as $\phi_j$ is rigid, too. Using this, we can
  obtain the following equations:
  \begin{gather*}
    \nodePosAcy{\oh g}{\psi(N)} = \nodePosAcy{\oh
      g}{\psi(\eqc{n^j}{\sim^+})} = \nodePosAcy{\oh g}{\phi_j(n^j)} =
    \nodePosAcy{g_j}{n^j} = \nodePosAcy{\ol g}{\eqc{n^j}{\sim^+}} =
    \nodePosAcy{\ol g}{N}
  \end{gather*}
  Hence, $\psi$ is a rigid $\bot$-homomorphism from $\ol g$ to $\oh
  g$.
\end{proof}

\section{Proof of Lemma~\ref{lem:lebotTrunc}}
\label{sec:proof-lemma-refl}

In this appendix, we will give the full proof of
Lemma~\ref{lem:lebotTrunc}. Before we can do this we have to establish
a number of technical auxiliary lemmas.

The lemma below will serve as a tool for the two lemmas that are to
follow afterwards. We know that the set of retained nodes
$\tNodes{g}{d}$ contains at least all nodes at depth smaller than $d$
due to the closure condition \trna{}. However, due to the closure
condition \trnb{} also nodes at a larger depth may be included in
$\tNodes{g}{d}$. The following lemma shows that this is not possible
for nodes labelled with nullary symbols:
\begin{lemma}[labelling]
  \label{lem:depthLabelling}
  % labelling %
  Let $g \in \itgraphs$, $\Delta \subseteq \Sigma^{(0)}$ and $d <
  \omega$. If $\sdepth{g}{\Delta} \ge d$, then $\glab^g(n) \nin
  \Delta$ for all $n \in \tNodes{g}{d}$.
\end{lemma}
\begin{proof}
  We will show that $N_\nabla = \setcom{n \in N^g}{\glab^g(n) \nin
    \Delta}$ satisfies the properties \trna{} and \trnb{} of
  Definition~\ref{def:truncGraph} for the term graph $g$ and depth
  $d$. Since $\tNodes{g}{d}$ is the least such set, we then obtain
  $\tNodes{g}{d} \subseteq N_\nabla$ and, thereby, the claimed
  statement.
  
  For \trna{}, let $n \in N^g$ with $\depth{g}{n} < d$. Since
  $\sdepth{g}{\Delta} \ge d$, we have $\glab^g(n) \nin \Delta$ and,
  therefore, $n\in N_\nabla$. For \trnb{}, let $n \in N_\nabla$ and $m
  \in \predAcy{g}{n}$. Then $m$ cannot be labelled with a nullary
  symbol, a fortiori $\glab^g(m) \nin \Delta$. Hence, we have $m \in
  N_\nabla$.
\end{proof}

The following two lemmas are rather technical. They state that rigid
$\Delta$-homomorphisms preserve retained nodes and in a stricter sense
also fringe nodes.

\begin{lemma}[preservation of retained nodes]
  \label{lem:presTruncNodes}
  % preservation of retained nodes %
  Let $g,h \in \itgraphs$, $d < \omega$, $\phi\fcolon g \homto_\Delta
  h$ rigid, and $\sdepth{g}{\Delta} \ge d$. Then $\phi(\tNodes{g}{d})
  = \tNodes{h}{d}$.
\end{lemma}
\begin{proof}
  Let $N_\nabla = \setcom{ n \in N^g}{\glab^g(n) \nin \Delta}$. At
  first we will show that $\phi(\tNodes{g}{d}) \subseteq
  \tNodes{h}{d}$. To this end, we will show that
  $\phi^{-1}(\tNodes{h}{d}) \cap N_\nabla$ satisfies \trna{} and
  \trnb{} of Definition~\ref{def:truncGraph} for term graph $g$ and
  depth $d$. Since $\tNodes{g}{d}$ is the least such set, we then
  obtain $\tNodes{g}{d} \subseteq \phi^{-1}(\tNodes{h}{d}) \cap
  N_\nabla$ and, a fortiori, $\tNodes{g}{d} \subseteq
  \phi^{-1}(\tNodes{h}{d})$ which is equivalent to
  $\phi(\tNodes{g}{d}) \subseteq \tNodes{h}{d}$.

  For \trna{}, let $n \in N^g$ with $\depth{g}{n} < d$. Because
  $\sdepth{g}{\Delta} \ge d$, we know that $\glab^g(n)\nin\Delta$,
  which means by Lemma~\ref{lem:strongDepthPres} that we also have
  $\depth{h}{\phi(n)} < d$. Hence, $\phi(n) \in \tNodes{h}{d}$ by
  \trna{}. Since $\glab^g(n) \nin \Delta$, we thus have $n \in
  \phi^{-1}(\tNodes{h}{d}) \cap N_\nabla$.

  For \trnb{}, let $n \in \phi^{-1}(\tNodes{h}{d}) \cap
  N_\nabla$. That is, we have $\phi(n) \in \tNodes{h}{d}$ and
  $\glab^g(n) \nin \Delta$. Hence, by \trnb{}, it holds that
  $\predAcy{h}{\phi(n)} \subseteq \tNodes{h}{d}$. We have to show now
  that $\predAcy{g}{n} \subseteq \phi^{-1}(\tNodes{h}{d}) \cap
  N_\nabla$. Let $m \in \predAcy{g}{n}$. That is, there is some $\pi
  \concat\seq i \in \nodePosAcy{g}{n}$ with $\pi \in \nodePos{g}{m}$. As
  $\glab^g(n) \nin \Delta$ and $\phi$ is rigid, we know that $\phi$ is
  rigid in $n$. Consequently, $\pi \concat\seq i \in
  \nodePosAcy{h}{\phi(n)}$. Moreover, we have $\pi \in
  \nodePos{h}{\phi(m)}$ by Lemma~\ref{lem:canhom}. Hence, $\phi(m) \in
  \predAcy{h}{\phi(n)}$ and, therefore, $\phi(m) \in \tNodes{h}{d}$ by
  \trnb{}. Additionally, as $m$ has a successor in $g$, it cannot be
  labelled with a symbol in $\Delta$. Hence, $m \in
  \phi^{-1}(\tNodes{h}{d}) \cap N_\nabla$.

  In order to prove the converse inclusion $\phi(\tNodes{g}{d})
  \supseteq \tNodes{h}{d}$, we will show that $\phi(\tNodes{g}{d})$
  satisfies \trna{} and \trnb{} for term graph $h$ and depth $d$. This
  will prove the abovementioned inclusion since $\tNodes{h}{d}$ is the
  least such set.

  For \trna{}, let $n \in N^h$ with $\depth{h}{n} < d$. By
  Lemma~\ref{lem:homDepthRev}, there is some $m \in N^g$ with
  $\depth{g}{m} < d$ and $\phi(m) = n$. Hence, according to \trna{},
  we have $m \in \tNodes{g}{d}$ and, therefore, $n \in
  \phi(\tNodes{g}{d})$.
  
  For \trnb{}, let $n \in \phi(\tNodes{g}{d})$. That is, there is some
  $m \in \tNodes{g}{d}$ with $\phi(m) = n$. By \trnb{}, we have
  $\predAcy{g}{m} \subseteq \tNodes{g}{d}$. We must show that
  $\predAcy{h}{n} \subseteq \phi(\tNodes{g}{d})$. Let $n' \in
  \predAcy{h}{n}$. That is, there is some $\pi \concat\seq i \in
  \nodePosAcy{h}{n}$ with $\pi\in \nodePos{h}{n'}$. Since $m \in
  \tNodes{g}{d}$, we have $\glab^g(m) \nin \Delta$ by
  Lemma~\ref{lem:depthLabelling}. Consequently, $\phi$ is rigid in $m$
  which yields that $\pi \concat\seq i \in \nodePosAcy{g}{m}$. Note that
  then also $\pi \in \pos{g}$. Let $m' = \nodeAtPos{g}{\pi}$. Thus,
  $m' \in \predAcy{g}{m}$ and, therefore, $m' \in \tNodes{g}{d}$
  according to \trnb{}. Moreover, because $\pi \in \nodePos{g}{m'}
  \cap \nodePos{h}{n'}$, we are able to obtain from
  Lemma~\ref{lem:canhom} that $\phi(m') = n'$. Hence, $n' \in
  \phi(\tNodes{g}{d})$.
\end{proof}

\begin{lemma}[preservation of fringe nodes]
  \label{lem:presFNodes}
  % preservation of fringe nodes %
  Let $g, h \in \itgraphs$, $\phi\fcolon g \homto_\Delta h$ rigid, $0
  < d < \omega$, $\sdepth{g}{\Delta} \ge d$, $n\in N^g$, and $0\le i <
  \rank{g}{n}$. Then $n^i \in \fNodes{g}{d}$ iff $\phi(n)^i \in
  \fNodes{h}{d}$.
\end{lemma}
\begin{proof}
  Note that, by Lemma~\ref{lem:depthLabelling}, we have that
  $\glab^g(n) \nin \Delta$ for all nodes $n \in
  \tNodes{g}{d}$. Additionally, by Lemma~\ref{lem:presTruncNodes}, we
  obtain $\phi(\tNodes{g}{d}) = \tNodes{h}{d}$ and, therefore,
  according to the labelling condition for $\phi$, we get that
  $\glab^h(n) \nin \Delta$ for all $n \in \tNodes{h}{d}$.

  At first we will show the ``only if'' direction. To this end, let
  $n^i \in \fNodes{g}{d}$. By definition, we then have $\depth{g}{n} \ge
  d - 1$. Hence, by Lemma~\ref{lem:strongDepthPres},
  $\depth{h}{\phi(n)} \ge d - 1$. Furthermore, we have that
  $\gsuc^g_i(n) \nin \tNodes{g}{d}$ or $n \nin
  \predAcy{g}{\gsuc^g_i(n)}$. We show now that in either case we can
  conclude $\phi(n)^i \in \fNodes{h}{d}$.

  Let $\gsuc^g_i(n) \nin \tNodes{g}{d}$. If we have
  $\gsuc^h_i(\phi(n)) \nin \tNodes{h}{d}$, then $\phi(n)^i \in
  \fNodes{h}{d}$. So suppose $\gsuc^h_i(\phi(n)) \in
  \tNodes{h}{d}$. Since $\tNodes{h}{d} = \phi(\tNodes{g}{d})$,
  according to Lemma~\ref{lem:presTruncNodes}, we find some $m \in
  \tNodes{g}{d}$ with $\phi(m) = \gsuc^h_i(\phi(n))$. However, since
  $\gsuc^g_i(n) \nin \tNodes{g}{d}$, we know that $m \neq
  \gsuc^g_i(n)$. We now show $\phi(n) \nin
  \predAcy{h}{\gsuc^h_i(\phi(n))}$ by showing that $\pi\concat\seq i
  \nin \nodePosAcy{h}{\gsuc^h_i(\phi(n))}$ whenever $\pi \in
  \nodePosAcy{h}{\phi(n)}$:
  \begin{align*}
    \pi \in \nodePosAcy{h}{\phi(n)}%
    \iff& \pi \in \nodePosAcy{g}{n}%
    \tag{$\phi$ is rigid in $n$}\\%
    \implies& \pi\concat\seq i \nin \nodePosAcy{g}{m}%
    \tag{$m \neq \gsuc^g_i(n)$}\\%
    \iff& \pi\concat\seq i \nin \nodePosAcy{h}{\phi(m)}%
    \tag{$\phi$ is rigid in $m$}\\%
    \iff& \pi\concat\seq i \nin \nodePosAcy{h}{\gsuc^h_i(\phi(n))}%
    \tag{$\phi(m) = \gsuc^h_i(\phi(n))$}%
  \end{align*}
  Together with $\depth{h}{\phi(n)} \ge d - 1$, this implies that
  $\phi(n)^i \in \fNodes{h}{d}$.

  Let $n \nin \predAcy{g}{\gsuc^g_i(n)}$. If $\phi(n) \nin
  \predAcy{h}{\gsuc^h_i(\phi(n))}$, then $\phi(n)^i \in
  \fNodes{h}{d}$. So suppose that $\phi(n) \in
  \predAcy{h}{\gsuc^h_i(\phi(n))}$. Hence, $\phi(n) \in
  \predAcy{h}{\phi(\gsuc^g_i(n))}$ as $\phi$ is homomorphic in $n$. If
  $\glab^g(\gsuc^g_i(n)) \nin \Delta$, then $\phi$ is rigid in
  $\gsuc^g_i(n)$ and we would also get $n \in
  \predAcy{g}{\gsuc^g_i(n)}$ which contradicts the assumption. Hence,
  $\glab^g(\gsuc^g_i(n)) \in \Delta$ and, therefore, $\gsuc^g_i(n)
  \nin \tNodes{g}{d}$ according to
  Lemma~\ref{lem:depthLabelling}. Thus, we can employ the argument for
  this case that we have already given above.

  We now turn to the converse direction. For this purpose, let
  $\phi(n)^i \in \fNodes{h}{d}$. Then $\depth{h}{\phi(n)} \ge d - 1$
  and, consequently $\depth{g}{n} \ge d - 1$ by
  Lemma~\ref{lem:strongDepthPres}. Additionally, we also have
  $\gsuc^h_i(\phi(n)) \nin \tNodes{h}{d}$ or $\phi(n) \nin
  \predAcy{h}{\gsuc^h_i(\phi(n))}$. Again we will show that in either
  case we can conclude $n^i \in \fNodes{g}{d}$.

  If $\gsuc^h_i(\phi(n)) \nin \tNodes{h}{d}$, then $\phi(\gsuc^g_i(n))
  \nin \tNodes{h}{d}$ and, therefore, $\phi(\gsuc^g_i(n)) \nin
  \phi(\tNodes{g}{d})$ according to
  Lemma~\ref{lem:presTruncNodes}. Consequently, $\gsuc^g_i(n)
  \nin\tNodes{g}{d}$ which implies that $n^i \in \fNodes{g}{d}$.

  Let $\phi(n) \nin \predAcy{h}{\gsuc^h_i(\phi(n))}$. If $n \nin
  \predAcy{g}{\gsuc^g_i(n)}$, then we get $n^i \in \fNodes{g}{d}$
  immediately. So assume that $n \in \predAcy{g}{\gsuc^g_i(n)}$. If
  $\glab^g(\gsuc^g_i(n)) \nin \Delta$, then $\phi$ would be rigid in
  $\gsuc^g_i(n)$. Thereby, we would get $\phi(n) \in
  \predAcy{h}{\phi(\gsuc^g_i(n))}$ which contradicts the
  assumption. Hence, $\glab^g(\gsuc^g_i(n)) \in \Delta$. According to
  Lemma~\ref{lem:depthLabelling}, we then have $\gsuc^g_i(n) \nin
  \tNodes{g}{d}$ and, therefore, $n^i \in \fNodes{g}{d}$.
\end{proof}

\begin{figure}
  \centering
  \begin{tikzpicture}[node distance=4cm, allow upside down, on grid]
    \node[node name=175:r] (r) {$h$}%
    child {%
      node[node name=175:n] (n) {$h$}%
        child {%
          node[node name=175:m] (m) {$h$}%
          child {%
            node[node name=175:o] (o) {$\bot$}%
          }%
      }%
    };%

    \node[node name=5:\ol{r},right=of r] (rp) {$h$}
    child {
      node[node name=5:\ol{n}] (np) {$h$}
    };
    \draw (np) edge[min distance=1cm,out=-45,in=0,->] (rp);
    \begin{scope}[|-to,black!50,dashed, thin] %|
      \path %
      (r) edge (rp)%
      (n) edge (np)%
      (m) edge[bend right=20] (rp)%
      (o) edge[bend right=20] (np)%
      ;
      \path(r) -- (rp) node[midway,above,black!70] {$\phi$};
    \end{scope}

    \begin{scope}[black!90, node distance=3.7cm]
      \node[below=of r] (g) {$h$};
      \node[node distance=5mm,left=of g] {$\phi\fcolon$};
      \node[below=of rp] (h) {$g$};
      \draw[fun,shorten=5mm] (g) -- (h) node[very near end,below] {\small$\bot$};
    \end{scope}
  \end{tikzpicture}
  \caption{Fringe nodes and rigid $\bot$-homomorphisms.}
  \label{fig:fNodes}
\end{figure}
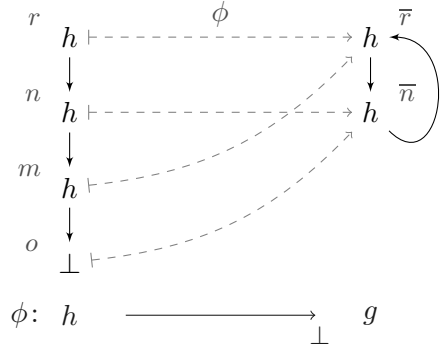

The above lemma depends on the peculiar definition of fringe nodes --
in particular those fringe nodes that are due to the condition
\[
\depth{g}{n} \ge d - 1 \text{ and } n\nin \predAcy{g}{\gsuc^g_i(n)}.
\]
Recall that this condition produces a fringe node for each edge from a
retained node that closes a cycle. Let us have a look at the term
graph $g$ depicted in Figure~\ref{fig:fNodes}. The rigid truncation
$\truncr{g}{2}$ of $g$ is shown in Figure~\ref{fig:loopTrunc}. If the
abovementioned alternative condition for fringe nodes would not be
present, then the set $\fNodes{g}{2}$ would be empty (and, thus,
$\truncr{g}{2} = g$). Then, however, the rigid $\bot$-homomorphism
$\phi$ illustrated in Figure~\ref{fig:fNodes} would violate
Lemma~\ref{lem:presFNodes}. Since the node $m$ is cut off from $h$ in
the truncation $\truncr{h}{2}$, there is a fringe node $n^0$ in
$\truncr{h}{2}$. On the other hand, there would be no fringe node $\ol
n^0$ in $\truncr{g}{2}$ if not for the alternative condition above.

\theoremstyle{plain}\newtheorem*{lemLebotTrunc}{Lemma~\ref{lem:lebotTrunc}}
\begin{lemLebotTrunc}[$\lebotr$ and rigid truncation]
  % $\lebot$ and truncation %
  Given $g, h \in \iptgraphs$ and $d < \omega$ with $g \lebotr h$ and
  $\sdepth{g}{\bot} \ge d$, we have that $\truncr{g}{d} \isom
  \truncr{h}{d}$.
\end{lemLebotTrunc}
\begin{proof}[Proof of Lemma~\ref{lem:lebotTrunc}]
  For $d = 0$, this is trivial. So assume $d > 0$. Since $g \lebotr h$,
  there is a rigid $\bot$-homomorphism $\phi\fcolon g \homto_\bot
  h$. Define the function $\psi$ as follows:
  \begin{align*}
    \psi\fcolon N^{\truncr{g}{d}} &\funto N^{\truncr{h}{d}}\\
    \tNodes{g}{d} \ni n &\mapsto \phi(n) \\
    \fNodes{g}{d} \ni n^i &\mapsto \phi(n)^i
  \end{align*}
  At first we have to argue that $\psi$ is well-defined. For this
  purpose, we first need that $\phi(\tNodes{g}{d}) \subseteq
  N^{\truncr{g}{d}}$. Lemma~\ref{lem:presTruncNodes} confirms
  this. Secondly, we need that $n^i \in \fNodes{g}{d}$ implies
  $\phi(n)^i \in N^{\truncr{g}{d}}$. This is guaranteed by
  Lemma~\ref{lem:presFNodes}.

  Next we show that $\psi$ is a homomorphism from $\truncr{g}{d}$ to
  $\truncr{h}{d}$. The root condition is inherited from $\phi$ as
  $r^{\truncr{g}{d}} \in \tNodes{g}{d}$. Note that, according to
  Lemma~\ref{lem:depthLabelling}, we have $\glab^g(n) \in\Sigma$ for
  all $n \in \tNodes{g}{d}$. Hence, $\phi$ is homomorphic in
  $\tNodes{g}{d}$ which means that the labelling condition for nodes
  in $\tNodes{g}{d}$ is also inherited from $\phi$. For nodes $n^i \in
  \fNodes{g}{d}$, we have $\glab^{\truncr{g}{d}}(n^i) = \bot$. Since, by
  definition, $\psi(n^i) \in \fNodes{h}{d}$, we can conclude
  $\glab^{\truncr{h}{d}}(\psi(n^i)) = \bot$.

  The successor condition is trivially satisfied by nodes in
  $\fNodes{g}{d}$ as they do not have any successors. Let $n \in
  \tNodes{g}{d}$ and $0 \le i < \rank{\truncr{g}{d}}{n}$. We
  distinguish two cases: At first assume that $n^i \nin
  \fNodes{g}{d}$. Hence, $\gsuc^{\truncr{g}{d}}_i(n) = \gsuc^g_i(n) \in
  \tNodes{g}{d}$. Since, by Lemma~\ref{lem:presFNodes}, also
  $\phi(n)^i \nin \fNodes{h}{d}$, we additionally have
  $\gsuc^{\truncr{h}{d}}_i(\phi(n)) = \gsuc^h_i(\phi(n))$. Hence, using
  the successor condition for $\phi$, we can reason as follows:
  \begin{gather*}
    \psi(\gsuc^{\truncr{g}{d}}_i(n)) = \psi(\gsuc^g_i(n)) =
    \phi(\gsuc^g_i(n)) = \gsuc^h_i(\phi(n)) =
    \gsuc^{\truncr{h}{d}}_i(\phi(n)) = \gsuc^{\truncr{h}{d}}_i(\psi(n))
  \end{gather*}
  If, on the other hand, $n^i \in \fNodes{g}{d}$, then
  $\gsuc^{\truncr{g}{d}}_i(n) = n^i$. Moreover, since then $\phi(n)^i
  \in \fNodes{h}{d}$ by Lemma~\ref{lem:presFNodes}, we have
  $\gsuc^{\truncr{h}{d}}_i(\phi(n)) = \phi(n)^i$, too. Hence, we can
  reason as follows:
  \begin{gather*}
    \psi(\gsuc^{\truncr{g}{d}}_i(n)) = \psi(n^i) = \phi(n)^i =
    \gsuc^{\truncr{h}{d}}_i(\phi(n)) = \gsuc^{\truncr{h}{d}}_i(\psi(n))
  \end{gather*}
  This shows that $\psi$ is a homomorphism. Note that, according to
  Lemma~\ref{lem:strongD-homInj}, $\phi$ is injective in
  $\tNodes{g}{d}$. Then also $\psi$ is injective in
  $\tNodes{g}{d}$. For the same reason, $\psi$ is also injective in
  $\fNodes{g}{d}$. Moreover, we have $\psi(\tNodes{g}{d}) \subseteq
  \tNodes{h}{d}$ and $\psi(\fNodes{g}{d}) \subseteq \fNodes{h}{d}$,
  i.e.\ $\psi(\tNodes{g}{d}) \cap \psi(\fNodes{g}{d}) =
  \emptyset$. Hence, $\psi$ is injective which implies, by
  Lemma~\ref{lem:isomBij}, that $\psi$ is an isomorphism from
  $\truncr{g}{d}$ to $\truncr{h}{d}$.
\end{proof}

\end{document}